\documentclass[11pt]{article}
\usepackage[margin=1in]{geometry}
\usepackage{amsmath,amssymb,amsthm,mathtools}
\usepackage{booktabs}
\usepackage{graphicx}
\usepackage{placeins}
\usepackage{microtype}
\usepackage{enumitem}
\usepackage{cite}
\usepackage[hidelinks]{hyperref}
\hypersetup{linktoc=all,bookmarksdepth=3}
\usepackage{bm}
\usepackage{array}
\usepackage{authblk}
\usepackage{xcolor}

\newtheorem{theorem}{Theorem}[section]
\newtheorem{proposition}[theorem]{Proposition}
\newtheorem{corollary}[theorem]{Corollary}
\newtheorem{lemma}[theorem]{Lemma}
\theoremstyle{remark}
\newtheorem{remark}[theorem]{Remark}
\newtheorem{example}[theorem]{Example}

\newcommand{\E}{\mathbb{E}}
\newcommand{\Pp}{\mathbb{P}}
\newcommand{\X}{\mathcal{X}}
\newcommand{\Z}{\mathcal{Z}}
\newcommand{\F}{\mathcal{F}}
\newcommand{\KL}[2]{D\!\left(#1\middle\|#2\right)}
\newcommand{\Pois}{\operatorname{Pois}}
\newcommand{\Ber}{\operatorname{Bern}}
\newcommand{\coI}{\mathsf{I}}
\newcommand{\cR}{\mathcal{R}}
\newcommand{\cC}{\mathcal{C}}
\newcommand{\cS}{\mathcal{S}}
\newcommand{\cJ}{\mathcal{J}}
\newcommand{\ind}{\mathbf{1}}

\newcommand{\Rmis}{\mathcal R_{\mathrm{mis}}}
\newcommand{\Rmat}{\cR_{P,P}}
\newcommand{\Dout}{\mathcal D_{\mathrm{out}}}
\newcommand{\Dpost}{\mathcal D_{\mathrm{post}}}
\newcommand{\Cmis}{\mathcal C_{\mathrm{mis}}}
\newcommand{\Smis}{\mathcal S_{\mathrm{mis}}}
\newcommand{\Cmat}{\cC_{P,P}}
\newcommand{\Smat}{\cS_{P,P}}

\title{\textbf{Information, Estimation, and Prediction for Discrete Memoryless
Channels via Poissonization}}
\author{Abhiram Gorle, Tsachy Weissman\\
Department of Electrical Engineering, Stanford University, CA \\
\normalsize\texttt{abhiramg@stanford.edu{,}tsachy@stanford.edu}}
\date{}

\begin{document}
\maketitle
\begin{abstract}
\noindent We develop information--estimation relations for an arbitrary discrete memoryless channel with finite input and output alphabets by varying the expected number of observations while keeping the channel fixed. A single hidden input is observed through a \textit{Poisson} number of conditionally independent channel uses, whose mean specifies the observation budget. The derivative of mutual information with respect to this budget equals the minimum expected relative-entropy loss for estimating the input's output distribution. The posterior average of the channel's output distributions achieves this minimum and is the Bayesian predictor for the next observation.

\noindent When the input distribution, the channel, or both are misspecified, the relative entropy between the true and assumed observation laws is the integral of the excess prediction loss over the budget. We also extend these relations to finite-state processes taking values in the input alphabet and observed through the channel at independent Poisson sampling times. Both in the matched setting and under model mismatch, the expected prediction loss accumulated over time using past observations equals the corresponding loss using the full observation record (noncausal estimation), averaged over lower sampling rates.

\noindent We also characterize how redundancy between two outputs selects among input distributions that achieve one-use capacity. For the coupon-collector channel, where each observation is sampled uniformly from an unknown subset, we derive the exact capacity loss of any input distribution. In particular, balanced subset families can substantially reduce the input alphabet while preserving the leading capacity terms.

\end{abstract}

\clearpage
\tableofcontents
\clearpage

\section{Introduction}

Suppose a single latent input is observed repeatedly through a noisy channel. This gives rise to \emph{two} fundamental questions: i) how much have the observations revealed about the input, and ii) how well can they predict the outcome of the next observation? Mutual information measures the first quantity. For a general discrete memoryless channel (DMC), we connect it to the second through an exact differential identity. We keep the channel fixed and vary the number of observations: the input is drawn once and observed through a \textit{Poisson} number of conditionally independent
channel uses with mean \(\gamma\). Poissonization supplies a continuous observation parameter without requiring a distinguished signal-to-noise parameter in the transition matrix.

\noindent Let \(X\sim P\) take values in a finite alphabet \(\X\), and let \(M\) be a fixed channel with finite output alphabet \(\Z\). Conditional on \(X\), the observations \(Z_1,Z_2,\ldots\) are independent with common distribution \(W_X=M(\cdot| X)\). We observe
\[
Y_\gamma=(N_\gamma,Z_1,\ldots,Z_{N_\gamma}),
\text{ } N_\gamma\sim\Pois(\gamma),
\]
where the count is independent of the input and of the channel-sampling randomness. Let \(Z'\) be a further observation with conditional distribution \(W_X\), independent of the record given \(X\). Its posterior predictive distribution is \(\widehat W_P^\gamma=\E_P[W_X| Y_\gamma]\).
Our basic identity is
\[
\frac{d}{d\gamma}I_P(X;Y_\gamma)
=\E_P D(W_X\|\widehat W_P^\gamma)
=I_P(X;Z'| Y_\gamma).
\]
Thus the marginal information gain equals the information supplied by one additional observation, given the observations already collected. Equivalently, it is the expected excess logarithmic prediction loss relative to a predictor that knows the true observation distribution \(W_X\). The posterior predictive distribution minimizes this expected KL loss, a special case of conditional-mean optimality under Bregman losses \cite{banerjee2005bregman}. The estimation object is a probability vector, so the formulation requires no numerical encoding of the input symbols. All logarithms are natural, with the usual conventions for zero probabilities, and derivatives at zero are generally taken from the right.

\subsection{Main results and scope}

Our results center on three themes: i) the relation between information gain and prediction loss, ii) the excess prediction loss caused by an incorrect model, and iii) the connection between causal prediction from past observations and noncausal estimation from the full record. We use these connections to explain diminishing information gains, determine what repeated observations can identify, and characterize capacity.

For a fixed input prior, repeated observations yield diminishing information gains. Let \(Z_1',Z_2'\) be two further observations through the same channel, independent of each other and of the record conditional on \(X\). Then
\[
\frac{d^2}{d\gamma^2}I_P(X;Y_\gamma)
=-I_P(Z_1';Z_2'|Y_\gamma).
\]
The dependence remaining between these observations measures the reduction in information gain: on average, observing one makes the other less informative about the input. Thus mutual information is nondecreasing and concave in the observation budget. We also derive higher-order identities and show that derivative signs need not continue to alternate (Section~\ref{sec:curvature}).

The same prediction framework applies when an estimator uses an incorrect input distribution, an incorrect channel, or both. Under the stated support and finiteness conditions, the relative entropy between the true and assumed observation laws equals the excess expected KL prediction loss relative to the matched predictor, integrated over observation budgets (Section~\ref{sec:static-mismatch}).

Repeated observations can identify the channel-output distribution without identifying the input itself. Inputs producing the same output distribution remain indistinguishable, regardless of the observation budget. If the assumed channel is incorrect, even identifying the input need not lead to accurate predictions of further observations. We also characterize when accumulated excess prediction loss stays bounded and when it grows linearly as the observation budget increases (Sections~\ref{sec:posterior-mismatch} and~\ref{sec:high-intensity}).

Keeping the channel and observation budget fixed, we maximize mutual information over the input distribution. The resulting capacity is measured in nats per static experiment: one input is selected and its entire Poisson-sized observation record is collected. Each candidate distribution uses its own matched predictor, giving a prediction-loss characterization of capacity and its optimality conditions. When several input distributions achieve the capacity of a single use of \(M\), those with the least redundancy between two outputs determine the second-order term of the small-budget capacity expansion (Section~\ref{sec:capacity}).

A coupon-collector channel gives a concrete input-design application. Its input is a subset of symbols, and each observation samples one member uniformly. For any input distribution, we express the loss relative to capacity through the distributions it induces on observed subsets. This identifies balanced input families that preserve specified orders of the small-budget capacity expansion. In particular, seven carefully chosen triples match the linear and quadratic terms obtained using all 35 triples of a 7-symbol alphabet, with an exact formula for the remaining loss (Section~\ref{sec:coupon-design}).

We next allow the hidden input to evolve in time, observing it through the fixed DMC at independent Poisson sampling times. We compare prediction from past observations (causal estimation) with smoothing from the complete record (noncausal estimation). Their time-integrated expected KL losses satisfy an exact relation: the causal prediction loss at a positive sampling rate equals the average noncausal smoothing loss over rates from zero to that rate~\cite{atar2012poisson}. Under the stated conditions, the same relation holds even when the assumed process law and channel are both misspecified (Sections~\ref{sec:ct-model}--\ref{sec:ct-mismatch}).

We work out explicit examples for binary symmetric,  erasure channels and a sampled Markov process. An image experiment further illustrates that recovering the input is not the same as predicting its outputs: even when the estimator identifies the true photograph, misspecified measurement noise leaves a positive KL prediction-loss floor for future noisy pixel measurements (Section~\ref{sec:examples}).

\subsection{Relation to prior work}

For Poisson channels, derivatives of mutual information are expressed through conditional mean estimates \cite{guo2008poisson}. Related identities connect mutual information, relative entropy, and causal/noncausal estimation \cite{atar2012poisson}, while vector Poisson gradients admit generalized Bregman representations \cite{wang2013bregman}.

The connection to our experiment follows from Poisson thinning \cite{last2017poisson}: conditional on the input, the output-symbol counts are independent Poisson variables. Because channel rows and their posterior means both sum to one, the associated Poisson estimation loss reduces to KL divergence (Remark~\ref{rem:poisson-kl}). The first-order identity is therefore a probability-simplex specialization of established Poisson results.

Different assumed priors and channels induce different distributions over channel rows. Treating these rows as Poisson inputs connects the prediction identities to existing input-law mismatch theory \cite{atar2012poisson}. We also handle zero channel probabilities under explicit support conditions. Repeated or multi-view channels likewise collect conditionally independent observations of one input. The analysis in~\cite{rameshwar2025multiview} establishes concavity in the number of views for each fixed input distribution and introduces a Poisson approximation channel for the binary symmetric channel, with the same conditional count law as our binary example. Here we study prediction and mismatch for the Poissonized experiment with an arbitrary finite DMC, using the mean count as a continuous observation budget.

For Poissonized observations, an expression for second derivatives in terms of negative conditional mutual information appears in the proof of~\cite[Proposition~4.1]{kireeva2024concavity}. Here we use this dependence measure to determine the second-order capacity term when the input distribution is optimized.

The coupon-collector channel and its capacity for a fixed number of reads were studied in~\cite{sokolovskii2025coupon} for combinatorial motif-based DNA storage. That work also considers restricted input families to reduce decoding complexity. In~\cite{sokolovskii2025coupon}, it is also noted that randomizing the read count allows capacity to be viewed as a function of the mean number of reads. Our application gives an exact expression for the information deficit of any input distribution, separating the number of distinct symbols observed from the imbalance in their subset distribution. The resulting balance conditions connect small-budget capacity to weighted combinatorial designs and characterize when a smaller input family retains full-input capacity for a fixed number of reads. The lower bound on the support needed to match the first two capacity terms is a weighted form of Fisher's design inequality~\cite{fisher1940blocks}.

Gaussian information--estimation relations connect mutual information to causal mean-square error \cite{duncan1970information}, and its signal-to-noise derivative to minimum mean-square error \cite{guo2005immse}. Relative entropy characterizes the accumulated estimation cost of an incorrect input law \cite{verdu2010mismatched}. The relation between causal and noncausal estimation also extends to mismatch \cite{weissman2010causal}.

Beyond these models, general channel-derivative formulas express changes in mutual information through posterior input distributions \cite{palomar2007input}. The generalized area theorem gives related entropy-derivative identities for memoryless channels \cite{measson2009area}. L\'evy channels admit information--estimation relations with losses adapted to the observation model \cite{jiao2017levy}, while semimartingale channels provide a broader continuous-time framework for causal estimation and mismatch \cite{jiao2017semimartingale}.

Recent diffusion work also connects estimation losses with data likelihoods. Such representations have been developed for Gaussian denoising \cite{kong2023diffusion} and discrete diffusion \cite{jeon2025discrete}; Poisson reconstruction loss provides a related connection for a discrete generative model \cite{bhattacharya2025itdpdm}. GEXIT-based conservation laws express data--model cross-entropy as an integral along a memoryless noise path \cite{aharoni2026conservation}. Here the DMC remains fixed: the continuous parameter controls the expected number of outputs collected, and prediction concerns the distribution of a further output.

Sections~\ref{sec:static-model}--\ref{sec:capacity} develop the static model and its information--estimation relations. Sections~\ref{sec:ct-model}--\ref{sec:ct-mismatch} give the continuous-time counterparts. Section~\ref{sec:dictionary} collects the main identities, Section~\ref{sec:examples} presents analytical examples and an \hyperref[sec:catalogue-example]{image-reconstruction example}, and Section~\ref{sec:discussion} discusses their implications. The appendices contain the longer proofs and computational details.

\section{Static Poissonized repeated-DMC model}\label{sec:static-model}

Let $\mathcal X$ and $\mathcal Z$ be nonempty finite alphabets, and let $M(z|x)$ be a fixed channel transition matrix. The input $X$ is drawn once with distribution $P$. 
Define
\begin{equation}
W_x(z)\triangleq M(z| x),\text{ } W\triangleq W_X=M(\cdot| X).
\end{equation}
Thus \(W\) is a random probability vector on \(\Z\). For each deterministic parameter $\gamma\geq0$, independently of the input and the channel-sampling randomness, let
\begin{equation}
 N_\gamma\sim\operatorname{Poisson}(\gamma).
 \end{equation}

\noindent Conditional on $X$, the observations $Z_1,Z_2,\ldots$ are independent with common distribution $W_X$:
\begin{equation}
Z_1,\ldots,Z_{N_\gamma}|X \stackrel{\mathrm{iid}}{\sim}\ W_X.
\end{equation}

\noindent We observe the first $N_\gamma$ outputs and define
\begin{equation}
Y_\gamma\triangleq (N_\gamma,Z^{N_\gamma}),
\end{equation}
where $Z^n=(Z_1,\ldots,Z_n)$ and $Z^0=\varnothing$. \\
The parameter $\gamma=\E[N_\gamma]$ is the observation budget: increasing $\gamma$ raises the expected number of observations while the prior $P$ and channel $M$ remain fixed.
An equivalent sufficient statistic is the vector $C=(C_z)_{z\in\Z}$ of output-symbol counts, where

\begin{equation}
C_z=\#\{1\le i\le N_\gamma:Z_i=z\},\text{ } z\in\Z.
\end{equation}
By Poisson thinning \cite{last2017poisson}, applied conditionally on $X$,
\begin{equation}\label{eq:vectorpois}
C_z| X\sim\Pois\!\left(\gamma W_X(z)\right),\text{ } z\in\Z,
\end{equation}
independently across \(z\). Conditional on a histogram $c$, every compatible ordering has probability $\prod_z c_z!/(\sum_z c_z)!$, independently of the input. 

\noindent Consequently, the histogram preserves the posterior on $X$. Moreover, its conditional distribution depends on $X$ only through the selected row $W$, so
\begin{equation}
I_P(X;Y_\gamma)=I_P(X;C)=I_P(W;C).
\end{equation}
\emph{Hence, the repeated-DMC experiment is equivalently a vector Poisson channel with random intensity vector \(\gamma W\).}

\noindent For a prior $S$ on $\X$, let $\E_S$ denote expectation under the same channel and sampling model with $X\sim S$. Define the posterior predictive row
\begin{equation}\label{eq:predrow}
\widehat W_S^\gamma\triangleq \E_S[W| Y_\gamma].
\end{equation}
If $Z_{\mathrm{new}}$ is a fresh observation through $M$, conditionally independent of $Y_\gamma$ given $X$, then
\begin{equation}
\widehat W_S^\gamma(z)
=
S(Z_{\mathrm{new}}=z| Y_\gamma).
\end{equation}

\noindent To see how information changes with the observation budget, we express each quantity of interest as a Poisson average of its fixed-sample counterparts. The elementary lemma below converts differentiation in $\gamma$ into a forward difference in the sample size $n$. In its first-order form, it replaces $a_n$ by the increment $a_{n+1}-a_n$ associated with one additional observation.
\begin{lemma}[Poisson differentiation]\label{lem:poisson-diff}
Let \(\{a_n\}_{n\ge0}\) satisfy \(|a_n|\le K(1+n)\) for some finite \(K\), let \(N_\gamma\sim\Pois(\gamma)\), and define \(F(\gamma)=\E[a_{N_\gamma}]\).  With the forward difference \(\Delta a_n=a_{n+1}-a_n\),
\begin{equation}\label{eq:poisson-k}
F^{(k)}(\gamma)=\E\!\left[\Delta^k a_{N_\gamma}\right],\text{ } k\ge1.
\end{equation}
At \(\gamma=0\), derivatives are interpreted from the right on the  parameter space \([0,\infty)\).
\end{lemma}

\begin{proof}
Termwise differentiation of
\(
F(\gamma)=e^{-\gamma}\sum_{n\ge0}a_n\gamma^n/n!
\)
gives the case \(k=1\); iteration gives \eqref{eq:poisson-k}.  The linear growth bound and the factorial denominator make the series and each of its termwise derivatives converge uniformly on compact \(\gamma\)-intervals, justifying these interchanges.
\end{proof}

\noindent Since $I_P(X;Z^n)\leq H_P(X)$, the lemma applies directly to the matched static information curve. Applications to mismatched likelihood ratios and continuous-time observations require the additional growth bounds established below.

\section{First-order information--estimation identity}\label{sec:first-order}

If the hidden input $X$ were known, a fresh output from the channel would have distribution $W=M(\cdot| X)$. Given only the observation record $Y_\gamma$, the Bayesian predictive distribution is the posterior average $\widehat W_P^\gamma=\E_P[W| Y_\gamma]$. Under logarithmic loss $-\log q(z)$, the expected excess loss from using $\widehat W_P^\gamma$ in place of the true row $W$, conditional on $X$ and $Y_\gamma$, is $\KL{W}{\widehat W_P^\gamma}$. Averaging over $(X,Y_\gamma)$ defines the \emph{matched KL prediction risk}
\begin{equation}\label{eq:matched-risk}
\cR_{P,P}(\gamma)
\triangleq
\E_P\left[\KL{W}{\widehat W_P^\gamma}\right].
\end{equation}
Here \emph{matched} means that the predictor uses the true prior $P$ and channel $M$. Because the estimation target is the probability vector $W$, the risk measures uncertainty about the selected channel row rather than the randomness of the next output. A predictor that knows $X$ knows $W$ exactly and incurs zero excess loss, even when the next output remains random.

\begin{theorem}[Poissonized-DMC I--estimation identity]\label{thm:first}
For every $\gamma\ge0$, 
\begin{equation}\label{eq:first}
\boxed{
\frac{d}{d\gamma}I_P(X;Y_\gamma)
=
\cR_{P,P}(\gamma)
=
\E_P\left[\KL{W}{\E_P[W| Y_\gamma]}\right].
}
\end{equation}
At $\gamma=0$, the derivative is taken from the right. Equivalently, if \(Z'\) is one fresh DMC observation conditionally independent of \(Y_\gamma\) given \(X\), then
\begin{equation}\label{eq:first-cmi}
\frac{d}{d\gamma}I_P(X;Y_\gamma)=I_P(X;Z'| Y_\gamma).
\end{equation}
\end{theorem}

\begin{proof}
For $n\ge0$, let $I_n=I_P(X;Z^n)$, with $I_0=0$. Since $N_\gamma$ is independent of $(X,Z_1,Z_2,\ldots)$, conditioning on the observation count gives
\begin{equation}
\begin{aligned}
I_P(X;Y_\gamma)
&=I_P(X;N_\gamma)+I_P(X;Z^{N_\gamma}| N_\gamma)=\E[I_{N_\gamma}].
\end{aligned}
\end{equation}
For a fixed sample count, the chain rule identifies the information gained from one additional observation:
\begin{equation}
I_{n+1}-I_n=I_P(X;Z_{n+1}| Z^n).
\end{equation}
Since $0\le I_n\le H_P(X)<\infty$, Lemma~\ref{lem:poisson-diff} gives
\begin{equation}
\begin{aligned}
\frac{d}{d\gamma}I_P(X;Y_\gamma)
=\E[I_{N_\gamma+1}-I_{N_\gamma}]=\sum_{n=0}^\infty \Pr(N_\gamma=n)
I_P(X;Z_{n+1}| Z^n)&=I_P(X;Z'| Y_\gamma).
\end{aligned}
\end{equation}
The last equality follows by conditioning on $N_\gamma$ and using that $Z'$ has the same conditional law as a fresh channel output.

\noindent Conditioned on both $X$ and $Y_\gamma$, the law of $Z'$ is $W_X$, whereas its law given only $Y_\gamma$ is $\widehat W_P^\gamma$. The conditional KL representation of mutual information therefore gives
\begin{equation}
I_P(X;Z'| Y_\gamma)
=
\E_P\left[\KL{W}{\widehat W_P^\gamma}\right],
\end{equation}
which proves \eqref{eq:first}. This calculation permits zero channel entries: at every positive-probability history, the posterior mixture assigns positive mass wherever a channel row with positive posterior probability does.
\end{proof}

\begin{remark}[How does KL appear?]\label{rem:poisson-kl}
For the equivalent vector Poisson channel \eqref{eq:vectorpois}, the natural Poisson loss identified by Atar and Weissman \cite{atar2012poisson} sums coordinatewise:
\begin{equation}
\sum_{z\in\Z}\left[W(z)\log\frac{W(z)}{\widehat W(z)}-W(z)+\widehat W(z)\right].
\end{equation}
Because \(W\) and \(\widehat W\) are probability vectors, the linear terms cancel, leaving exactly
\(
\KL{W}{\widehat W}
\).
Thus Theorem~\ref{thm:first} recovers the simplex specialization of the Atar--Weissman Poisson information--estimation identity. The vector mismatch theorem in \cite{atar2012poisson} is stated for intensity inputs whose components are bounded between two positive constants; the finite-alphabet proof above handles zero channel entries directly.
\end{remark}

\noindent At $\gamma=0$, the observation is the deterministic empty record, so $I_P(X;Y_0)=0$. The derivative is continuous by Lemma~\ref{lem:poisson-diff}; integrating \eqref{eq:first} therefore gives
\begin{equation}\label{eq:integrated-matched}
\boxed{
I_P(X;Y_\gamma)=\int_0^\gamma \cR_{P,P}(\alpha)\,d\alpha.
}
\end{equation}

\section{Mismatched estimation and relative entropy}\label{sec:static-mismatch}
The matched predictor in Section~\ref{sec:first-order} uses the true statistical model. We now keep the true experiment $(P,M)$ and the target row $W$ fixed, first allowing the estimator to use an incorrect prior and then an incorrect channel as well. In both cases, relative entropy between the true and nominal observation laws accounts for the resulting excess prediction loss.

\subsection{Prior mismatch with a fixed channel}\label{sec:prior-mismatch}

Let $P$ be the true input prior and $Q$ the prior used by the estimator. The channel $M$ and the Poisson sampling law are the same under both models. The estimator therefore uses the correct likelihood for each possible input, but combines these likelihoods with different prior weights. This can lead to different posterior predictions from the same observed record.

\noindent Assume $P\ll Q$, meaning that $Q(x)>0$ whenever $P(x)>0$. Thus the nominal model does not exclude any input that can occur under the true model. Every history of positive probability under $P$ also has positive probability under $Q$, and the nominal predictive row is positive wherever the true predictive row is positive.

Define
\begin{equation}\label{eq:mismatch-risk}
\cR_{P,Q}(\gamma)
\triangleq
\E_P\left[
\KL{W}{\widehat W_Q^\gamma}
\right],
\text{ }
\widehat W_Q^\gamma=\E_Q[W| Y_\gamma].
\end{equation}

\noindent Here $\E_Q$ constructs the predictor, whereas $\E_P$ evaluates its performance under the true experiment. The additional cost of using $Q$ rather than $P$ is the risk difference $\cR_{P,Q}(\gamma)-\cR_{P,P}(\gamma)$.

\noindent Let $P_{Y_\gamma}$ and $Q_{Y_\gamma}$ denote the laws of $Y_\gamma$ under the two priors.

\begin{theorem}[Mismatched information--estimation identity]\label{thm:mismatch}
For \(\gamma\ge0\),
\begin{equation}\label{eq:mismatch-master}
\boxed{
\frac{d}{d\gamma}
\left[
I_P(X;Y_\gamma)+\KL{P_{Y_\gamma}}{Q_{Y_\gamma}}
\right]
=
\cR_{P,Q}(\gamma).
}
\end{equation}
The derivative is finite and continuous in $\gamma$; at $\gamma=0$, it is taken from the right. Consequently,
\begin{equation}\label{eq:mismatch-integrated}
\boxed{
I_P(X;Y_\gamma)+\KL{P_{Y_\gamma}}{Q_{Y_\gamma}}
=
\int_0^\gamma \cR_{P,Q}(\alpha)\,d\alpha.
}
\end{equation}
\end{theorem}

\begin{proof}
For fixed \(n\), let \(P_n,Q_n\) denote the output laws of \(Z^n\), and set
\begin{equation}
G_n=I_P(X;Z^n)+\KL{P_n}{Q_n}.
\end{equation}
Combining logarithms,
\begin{equation}\label{eq:mismatch-fixed-count}
G_n
=
\E_P\log\frac{P(Z^n| X)}{Q_n(Z^n)}.
\end{equation}
Since \(Q_n(Z^n)\ge Q(X)P(Z^n| X)\) holds \(P\)-almost surely,
\[
0\le G_n\le \E_P[-\log Q(X)]<\infty.
\]
Thus \((G_n)\) is bounded, justifying Poisson differentiation below.
Hence
\begin{align}
G_{n+1}-G_n
&=\E_P\log\frac{W_X(Z_{n+1})}{Q(Z_{n+1}| Z^n)}\notag\\
&=\E_P\left[\KL{W}{\E_Q[W| Z^n]}\right].\label{eq:mismatch-increment}
\end{align}
The last equality conditions on $(X,Z^n)$: the new output has true law $W_X$, whereas its nominal predictive law is $\E_Q[W| Z^n]$. Because the observation count has the same law under both models and is independent of the input and output sequence, conditioning on it gives
\begin{equation*}
I_P(X;Y_\gamma)+\KL{P_{Y_\gamma}}{Q_{Y_\gamma}}
=
\E[G_{N_\gamma}].
\end{equation*}
Applying Lemma~\ref{lem:poisson-diff} and using \eqref{eq:mismatch-increment} proves \eqref{eq:mismatch-master}, including continuity of the derivative. At $\gamma=0$, both models produce the same deterministic empty record, so the mutual information and output relative entropy vanish. Integration therefore gives \eqref{eq:mismatch-integrated}.
\end{proof}

Subtracting the matched identity yields the exact mismatch-regret formula
\begin{equation}\label{eq:out-div-deriv}
\boxed{
\frac{d}{d\gamma}\KL{P_{Y_\gamma}}{Q_{Y_\gamma}}
=
\cR_{P,Q}(\gamma)-\cR_{P,P}(\gamma).
}
\end{equation}
Therefore
\begin{equation}\label{eq:out-div-int}
\KL{P_{Y_\gamma}}{Q_{Y_\gamma}}
=
\int_0^\gamma
\left[\cR_{P,Q}(\alpha)-\cR_{P,P}(\alpha)\right]d\alpha.
\end{equation}

Thus output relative entropy is the accumulated excess prediction loss over observation budgets up to $\gamma$. The next result identifies this excess loss at each budget directly: it is the expected KL disagreement between the matched and mismatched predictions for a fresh output.

\begin{proposition}[Predictive disagreement]\label{prop:predictive-disagreement}
For every $\gamma\ge0$,
\begin{equation}\label{eq:pyth}
\boxed{
\cR_{P,Q}(\gamma)-\cR_{P,P}(\gamma)
=
\E_P\left[
\KL{\widehat W_P^\gamma}{\widehat W_Q^\gamma}
\right].
}
\end{equation}
Consequently,
\begin{equation}\label{eq:out-div-pred}
\boxed{
\frac{d}{d\gamma}\KL{P_{Y_\gamma}}{Q_{Y_\gamma}}
=
\E_P\left[
\KL{\widehat W_P^\gamma}{\widehat W_Q^\gamma}
\right].
}
\end{equation}
\end{proposition}

\begin{proof}
Fix a history $y$ of positive probability under $P$, and evaluate both posterior predictors at this history. For any probability vector $q$ positive on the support of $\widehat W_P^\gamma(y)$, split the logarithm in $\KL{W}{q}$ through $\widehat W_P^\gamma(y)$. Taking the conditional expectation and using $\E_P[W| Y_\gamma=y]=\widehat W_P^\gamma(y)$ gives
\begin{equation}\label{eq:conditional-kl-decomposition}
\E_P[\KL{W}{q}| y]
=
\E_P[\KL{W}{\widehat W_P^\gamma}| y]
+
\KL{\widehat W_P^\gamma}{q}.
\end{equation}
The choice $q=\widehat W_Q^\gamma(y)$ satisfies the support condition because $P\ll Q$. Substituting this choice and averaging over $Y_\gamma$ proves \eqref{eq:pyth}. Combining it with \eqref{eq:out-div-deriv} gives \eqref{eq:out-div-pred}.

\end{proof}

Thus the derivative of output relative entropy is exactly the expected KL disagreement between the matched and mismatched predictions for the next sample.

\subsection{Simultaneous prior and channel mismatch}\label{sec:joint-mismatch}

Suppose the estimator uses both a nominal prior \(Q\) and a nominal channel \(\widetilde M\) with rows \(\widetilde W_x=\widetilde M(\cdot| x)\), while the true model remains \((P,M)\). Both models use the same independent \(\Pois(\gamma)\) count, and both prior--channel pairs remain fixed as \(\gamma\) varies. Unless a weaker condition is explicitly stated, the static results below assume
\begin{equation}\label{cm:eq:standing-support}
P\ll Q,\text{ } W_x\ll\widetilde W_x
\text{ }\text{for every }x\text{ with }P(x)>0.
\end{equation}
For finite alphabets these assumptions imply \(D(P\|Q)<\infty\) and
\(\kappa=\sum_{x:P(x)>0}P(x)D(W_x\|\widetilde W_x)<\infty\). Here \(\kappa\) abbreviates \(\kappa(P;M,\widetilde M)\). These assumptions permit zero entries elsewhere and ensure that the nominal likelihood and its required next-symbol probabilities are positive on true histories.

Write $P_\gamma^M$ and $Q_\gamma^{\widetilde M}$ for the true and nominal laws of $Y_\gamma$. Let $P_M(x|y)$ and $Q_{\widetilde M}(x|y)$ denote the posterior probabilities of the input computed under $(P,M)$ and $(Q,\widetilde M)$, respectively. For a fixed observed record $y$, the two predictive rows are
\begin{align}
\widehat W_{P,\gamma}(z| y)
&=\sum_xP_M(x| y)W_x(z)
,\label{cm:eq:truepred}\\
\widehat{\widetilde W}_{Q,\gamma}(z| y)
&=\sum_xQ_{\widetilde M}(x| y)\widetilde W_x(z).
\label{cm:eq:nompred}
\end{align}
Evaluating the true predictor at the random record recovers the earlier notation:
$\widehat W_{P,\gamma}(\cdot| Y_\gamma)=\widehat W_P^\gamma$.
When the history argument is omitted, both predictive rows are understood to be evaluated at $Y_\gamma$.

Two parts of the estimator change: the posterior is computed with the nominal likelihood, and its prediction averages the nominal rows. The loss still compares this prediction with the true row \(W_X\). Define
\begin{equation}\label{cm:eq:Rmis-def}
\Rmis(\gamma)=\E_{P,M}D(W_X\|\widehat{\widetilde W}_{Q,\gamma}).
\end{equation}
We continue to use $\cR_{P,P}(\gamma)$ for the matched risk.

When \(\widetilde M=M\), the predictor in \eqref{cm:eq:nompred} reduces to \(\widehat W_Q^\gamma\) and \(\Rmis=\cR_{P,Q}\). The following result therefore extends Theorem~\ref{thm:mismatch} with the same prediction-loss interpretation.

\begin{theorem}[Simultaneous prior and channel mismatch]\label{cm:thm:main}
Let
\begin{equation}
\mathcal J(\gamma)
\triangleq
I_{P,M}(X;Y_\gamma)
+
\KL{P_\gamma^M}{Q_\gamma^{\widetilde M}}.
\end{equation}
For every $\gamma\ge0$,
\begin{equation}\label{cm:eq:main-identity}
\boxed{
\frac{d}{d\gamma}\mathcal J(\gamma)=\Rmis(\gamma).
}
\end{equation}
The derivative is finite and continuous in $\gamma$; at $\gamma=0$, it is taken from the right. Since $\mathcal J(0)=0$,
\begin{equation}\label{cm:eq:main-integrated}
\boxed{
I_{P,M}(X;Y_\gamma)
+
\KL{P_\gamma^M}{Q_\gamma^{\widetilde M}}
=
\int_0^\gamma \Rmis(s)\,ds.
}
\end{equation}
\end{theorem}

\begin{proof}
For exactly $n$ observations, write $P_n^M$ and $Q_n^{\widetilde M}$ for the corresponding output laws and define
\begin{equation}\label{cm:eq:Jn-def}
J_n
\triangleq
I_{P,M}(X;Z^n)+\KL{P_n^M}{Q_n^{\widetilde M}}.
\end{equation}
The likelihood representation \eqref{eq:mismatch-fixed-count} applies with \(Q_n\) replaced by \(Q_n^{\widetilde M}\): the true conditional likelihood remains \(\prod_{i=1}^nW_X(Z_i)\). On a true history with $X=x$, the nominal denominator is at least
$Q(x)\prod_{i=1}^n\widetilde W_x(Z_i)$. Consequently,
\begin{equation}\label{cm:eq:Jn-bound}
0\le J_n\le\E_P[-\log Q(X)]+n\kappa
=H(P)+D(P\|Q)+n\kappa.
\end{equation}
Unlike the prior-only case, the bound may now grow linearly with the sample count. The linear-growth form of Lemma~\ref{lem:poisson-diff} is sufficient for this extension. For every $x$ with $P(x)>0$, both $W_x$ and $\widetilde W_x$ are positive on the support of $W_x$. Since the alphabets are finite, the corresponding log likelihood ratios are bounded in absolute value by constants of order $1+n$. In particular, the subtraction below involves finite expectations. 

\noindent We now write $\widehat{\widetilde W}_{Q,n}(\cdot| Z^n)$ for the nominal predictive row after exactly $n$ observations. Factoring off one observation as in \eqref{eq:mismatch-increment}, now using the nominal-channel predictor \eqref{cm:eq:nompred}, gives
\begin{equation}\label{cm:eq:Jn-increment}
J_{n+1}-J_n
=\E_{P,M}
\KL{W_X}{\widehat{\widetilde W}_{Q,n}(\cdot| Z^n)}.
\end{equation}
Because the count has the same Poisson law under both models and, under each model, is independent of the input and the full output sequence,
\begin{equation}
\mathcal J(\gamma)=\E[J_{N_\gamma}].
\end{equation}
The bound \eqref{cm:eq:Jn-bound} permits Lemma~\ref{lem:poisson-diff} to be applied to \eqref{cm:eq:Jn-increment}, yielding \eqref{cm:eq:main-identity} with a right derivative at zero. The fixed-count risks, being nonnegative increments, also satisfy an $O(1+n)$ bound; their Poisson averages are therefore finite and continuous, including at zero. Since \(J_0=\mathcal J(0)=0\), integration gives \eqref{cm:eq:main-integrated}.
\end{proof}

Subtracting \eqref{eq:first} again isolates excess prediction loss, just as in \eqref{eq:out-div-deriv}. The conditional KL decomposition behind Proposition~\ref{prop:predictive-disagreement} also applies to the nominal-channel predictor.

\begin{corollary}[Derivative of output relative entropy]\label{cm:cor:Dout}
Define
\begin{equation}
\Dout(\gamma)
\triangleq
\KL{P_\gamma^M}{Q_\gamma^{\widetilde M}}.
\end{equation}
Then
\begin{align}
\Dout'(\gamma)
&=\Rmis(\gamma)-\Rmat(\gamma)\label{cm:eq:Dout-diff-risk}\\
&=
\E_{P,M}
\KL{\widehat W_{P,\gamma}(\cdot| Y_\gamma)}
{\widehat{\widetilde W}_{Q,\gamma}(\cdot| Y_\gamma)}.
\label{cm:eq:Dout-predictive}
\end{align}
Consequently,
\begin{equation}\label{cm:eq:Dout-integrated}
\boxed{
\Dout(\gamma)
=
\int_0^\gamma
\E_{P,M}
\KL{\widehat W_{P,s}}{\widehat{\widetilde W}_{Q,s}}
\,ds.
}
\end{equation}
In particular, $\Dout(\gamma)$ is nondecreasing in $\gamma$.
\end{corollary}

\begin{proof}
Subtract \eqref{eq:first} from \eqref{cm:eq:main-identity} for the first equality. For the second, apply \eqref{eq:conditional-kl-decomposition} with \(q=\widehat{\widetilde W}_{Q,\gamma}(\cdot| y)\) and average over the true observation law. The support condition \eqref{cm:eq:standing-support} and the risk bounds in the preceding proof make these quantities finite. Integration from \(\Dout(0)=0\) gives \eqref{cm:eq:Dout-integrated}.
\end{proof}

Equation~\eqref{cm:eq:Dout-predictive} gives a direct predictive interpretation:
\begin{quote}
The rate at which the true and nominal output laws separate is exactly the expected KL disagreement between their predictions for the next observation.
\end{quote}
This statement includes prior mismatch and channel mismatch on equal footing.

\begin{remark}[Weaker support for output and predictive identities]\label{cm:rem:active-cover}
Theorem~\ref{cm:thm:main} and Corollary~\ref{cm:cor:Dout} also hold if, for every $x$ with $P(x)>0$, there exists a possibly different active label $j_x$ such that $Q(j_x)>0$ and $W_x\ll\widetilde W_{j_x}$. Indeed, their proofs then use
\begin{equation}\label{cm:eq:active-cover-bound}
0\le J_n\le
\sum_{x:P(x)>0}P(x)\left[-\log Q(j_x)+nD(W_x\|\widetilde W_{j_x})\right]
\end{equation}
in place of \eqref{cm:eq:Jn-bound}. The same likelihood and risk bounds justify differentiation and conditioning. This is a support condition, not an exact-representation condition: $\widetilde W_{j_x}$ need not equal $W_x$, and the nominal label $j_x$ need not equal the true label $x$. Prediction of the observations requires only this support coverage. The posterior conservation law below instead compares beliefs about the same latent labels, so we retain the stronger assumptions \eqref{cm:eq:standing-support} for that law.

\end{remark}

\subsection{Posterior mismatch and channel discrepancy}\label{sec:posterior-mismatch}

The prediction identities above concern the induced observation law. To understand their implications for inference about the latent label, we apply the KL chain rule to the joint law of $(X,Y_\gamma)$. This gives a conservation law separating observable prediction discrepancy from disagreement between label posteriors, under the same-label support assumptions \eqref{cm:eq:standing-support}.

\noindent Recall the same-label average channel discrepancy
\begin{equation}\label{cm:eq:kappa}
\kappa(P;M,\widetilde M)
=
\E_P\KL{W_X}{\widetilde W_X}.
\end{equation}
This quantity compares the true and nominal output laws at the same input label. It is the expected excess logarithmic loss incurred by using $\widetilde W_X$ instead of $W_X$ even when $X$ is known. It therefore measures channel discrepancy without uncertainty about the input and is independent of the observation budget $\gamma$. Also define the expected posterior mismatch
\begin{equation}\label{cm:eq:Dpost-def}
\Dpost(\gamma)
\triangleq
\E_{P_\gamma^M}
\KL{P_M(\cdot| Y_\gamma)}
{Q_{\widetilde M}(\cdot| Y_\gamma)}.
\end{equation}
Both arguments are posterior distributions on the same input alphabet $\X$, formed from the same observed record. The expectation averages their disagreement over records generated by the true model.

\begin{theorem}[Relative-entropy conservation]\label{cm:thm:conservation}
For every $\gamma\ge0$,
\begin{equation}\label{cm:eq:conservation}
\boxed{
D(P\|Q)
+
\gamma\,\kappa(P;M,\widetilde M)
=
\Dout(\gamma)+\Dpost(\gamma).
}
\end{equation}
Therefore
\begin{equation}\label{cm:eq:Dpost-derivative}
\boxed{
\Dpost'(\gamma)
=
\kappa(P;M,\widetilde M)
-
\E_{P,M}\KL{\widehat W_{P,\gamma}}{\widehat{\widetilde W}_{Q,\gamma}}.
}
\end{equation}
\end{theorem}

\begin{proof}
By the KL chain rule,
\begin{align}
D\!\left(P^M_{X,Y_\gamma}\middle\|Q^{\widetilde M}_{X,Y_\gamma}\right) =
D(P\|Q)
+
\E_P D\!\left(P^M_{Y_\gamma| X}\middle\|Q^{\widetilde M}_{Y_\gamma| X}\right).
\end{align}
Conditional on $X=x$, the count has the same law under both models. For a fixed count $n$, independence of the marks makes the conditional relative entropy equal to $nD(W_x\|\widetilde W_x)$. Averaging over the count and using $\E[N_\gamma]=\gamma$ therefore gives
\begin{equation}
D\!\left(P^M_{Y_\gamma| X=x}\middle\|Q^{\widetilde M}_{Y_\gamma| X=x}\right)
=
\gamma D(W_x\|\widetilde W_x).
\end{equation}
Applying the KL chain rule in the opposite order gives
\begin{equation}
D\!\left(P^M_{X,Y_\gamma}\middle\|Q^{\widetilde M}_{X,Y_\gamma}\right)
=
\Dout(\gamma)+\Dpost(\gamma),
\end{equation}
which proves \eqref{cm:eq:conservation}. Differentiation and Corollary~\ref{cm:cor:Dout} give \eqref{cm:eq:Dpost-derivative}.
\end{proof}

\noindent When only the prior is wrong, $\widetilde M=M$ and $\kappa=0$, so the expected posterior divergence $\Dpost(\gamma)$ is nonincreasing in the observation budget. This is an averaged statement, not a claim that posterior disagreement decreases after every realized observation. With channel mismatch, the joint-model relative entropy on the left side of \eqref{cm:eq:conservation} grows at rate $\kappa$. The change in expected posterior disagreement is the difference between this rate and the current expected predictive disagreement, as expressed by \eqref{cm:eq:Dpost-derivative}. Additional observations can therefore make the two label posteriors more different on average.

\begin{remark}[No canonical additive split into prior and channel regret]
The prediction regret
\begin{equation}
\E_{P,M}\KL{\widehat W_{P,\gamma}}{\widehat{\widetilde W}_{Q,\gamma}}
\end{equation}
does not, in general, admit a canonical decomposition into nonnegative prior-only and channel-only contributions. The two errors can cancel, as Example~\ref{cm:ex:label-swap} illustrates. Equation~\eqref{cm:eq:conservation} instead partitions initial prior discrepancy plus accumulated same-label channel discrepancy into output discrepancy and residual posterior discrepancy.
\end{remark}

\begin{example}[Label swapping: correct predictions, increasing posterior mismatch]\label{cm:ex:label-swap}
Let $\X=\Z=\{0,1\}$, $P=Q=(1/2,1/2)$, and
\begin{equation}
W_0=\operatorname{Bern}(\eta),\text{ }
W_1=\operatorname{Bern}(1-\eta),\text{ }
\widetilde W_0=W_1,\text{ }\widetilde W_1=W_0,
\text{ } 0<\eta<1/2.
\end{equation}
Relabeling the equally weighted mixture leaves its observation law unchanged, so $P_\gamma^M=Q_\gamma^{\widetilde M}$ for every $\gamma$. The nominal posterior satisfies
$Q_{\widetilde M}(x| y)=P_M(1-x| y)$, and its mixture of the swapped rows is exactly the true predictor. Therefore
\begin{equation}
\Dout(\gamma)=0,\text{ } \Rmis(\gamma)=\Rmat(\gamma),\text{ }
\kappa=(1-2\eta)\log\frac{1-\eta}{\eta}>0.
\end{equation}
The conservation law gives $\Dpost(\gamma)=\gamma\kappa$: label-posterior disagreement grows without bound even though the two output laws and predictors agree exactly.

The same cancellation also permits both a wrong prior and wrong rows. For $P=(p,1-p)$ and $Q=(1-p,p)$, $p\in(0,1)\setminus\{1/2\}$, keep the swapped rows. Again the full output laws coincide. In contrast, either changing only the prior or swapping only the rows changes the one-step marginal, giving positive initial predictive regret by Proposition~\ref{cm:prop:small}. Thus two separately observable discrepancies can cancel when combined.
\end{example}

\subsection{Initial slopes under mismatch}

At $\gamma=0$ no data have been observed. Define the true and nominal one-step marginal output laws
\begin{equation}
\overline W_P\triangleq\sum_xP(x)W_x,
\text{ }
\overline{\widetilde W}_Q\triangleq\sum_xQ(x)\widetilde W_x.
\end{equation}
In the proposition below, $(X,Z)$ denotes one use of the true channel: $X\sim P$ and $Z| X=x\sim W_x$.

\begin{proposition}[Initial slopes]\label{cm:prop:small}
At zero intensity,
\begin{equation}\label{cm:eq:Dout0}
\boxed{
\Dout'(0+)
=
\KL{\overline W_P}{\overline{\widetilde W}_Q}.
}
\end{equation}
Moreover,
\begin{equation}\label{cm:eq:Rmis0}
\boxed{
\Rmis(0)
=
I_{P,M}(X;Z)
+
\KL{\overline W_P}{\overline{\widetilde W}_Q}.
}
\end{equation}
\end{proposition}

\begin{proof}
At $\gamma=0$, the predictors in \eqref{cm:eq:truepred}--\eqref{cm:eq:nompred} reduce to $\overline W_P$ and $\overline{\widetilde W}_Q$. Equation~\eqref{cm:eq:Dout-predictive} gives \eqref{cm:eq:Dout0}. Finally,
\begin{align}
\E_P\KL{W_X}{\overline{\widetilde W}_Q}
&=
\E_P\KL{W_X}{\overline W_P}
+
\KL{\overline W_P}{\overline{\widetilde W}_Q}\\
&=
I_{P,M}(X;Z)
+
\KL{\overline W_P}{\overline{\widetilde W}_Q}.
\end{align}
\end{proof}
\noindent Before any observations are available, the mismatched prediction risk therefore consists of the matched risk, equal to the information in one true channel use, plus the divergence between the true and nominal one-step output marginals. The latter is the entire initial penalty for using the nominal prior--channel pair.

\section{Curvature, diminishing returns, and higher derivatives}\label{sec:curvature}

Returning to the matched model with fixed \((P,M)\), we differentiate \(I_P(\gamma)\triangleq I_P(X;Y_\gamma)\) a second time. Let \(Z_1',Z_2',\ldots\) be fresh observations drawn i.i.d.\ through \(M\) given \(X\), conditionally independently of \(Y_\gamma\) given \(X\). Although independent given \(X\), these outputs may remain dependent given \(Y_\gamma\), since they share the same uncertain channel row. Observing one can therefore change the prediction of the other.

\begin{theorem}[Second derivative and predictive dependence]\label{thm:second}
The function \(\gamma\mapsto I_P(X;Y_\gamma)\) is nondecreasing and concave, with
\begin{equation}\label{eq:second}
\boxed{
\frac{d^2}{d\gamma^2}I_P(X;Y_\gamma)
=-I_P(Z_1';Z_2'| Y_\gamma)\le0.
}
\end{equation}
If
\begin{equation}
\widehat W_{P,+}^\gamma
\triangleq
\E_P[W| Y_\gamma,Z_1'],
\end{equation}
then equivalently
\begin{equation}\label{eq:predictive-update}
\boxed{
-I_P''(\gamma)
=
\E_P\left[
\KL{\widehat W_{P,+}^\gamma}{\widehat W_P^\gamma}
\right].
}
\end{equation}
\end{theorem}

\begin{proof}
Recall that $I_n=I_P(X;Z^n)$ and $0\le I_n\le H_P(X)$. Applying Lemma~\ref{lem:poisson-diff} twice gives
\begin{equation*}
I_P''(\gamma)=\E[\Delta^2 I_{N_\gamma}].
\end{equation*}
For each $n\ge0$, the fixed-sample increment satisfies
\begin{align}
\Delta^2 I_n
&=I_P(X;Z_{n+2}| Z^n,Z_{n+1})
-I_P(X;Z_{n+1}| Z^n)\\
&=I_P(X;Z_{n+2}| Z^n,Z_{n+1})
-I_P(X;Z_{n+2}| Z^n)\\
&=I_P(Z_{n+1};Z_{n+2}| X,Z^n)
-I_P(Z_{n+1};Z_{n+2}| Z^n)\\
&=-I_P(Z_{n+1};Z_{n+2}| Z^n).
\end{align}
The second equality uses conditional exchangeability of the two fresh outputs. The third follows by expanding $I_P(X,Z_{n+1};Z_{n+2}| Z^n)$ in two orders. Conditional independence given $X$ makes the first term in that line zero. Averaging over $N_\gamma$ proves \eqref{eq:second}. Together with the nonnegative first derivative in \eqref{eq:first}, this establishes that the information curve is nondecreasing and concave.

The conditional law of $Z_2'$ given $Y_\gamma$ is $\widehat W_P^\gamma$. After additionally observing $Z_1'$, that law becomes $\widehat W_{P,+}^\gamma$. The conditional KL representation of mutual information therefore gives
\begin{equation*}
I_P(Z_1';Z_2'| Y_\gamma)
=
\E_P\left[
\KL{\widehat W_{P,+}^\gamma}{\widehat W_P^\gamma}
\right],
\end{equation*}
which proves \eqref{eq:predictive-update}. Thus observations have diminishing returns: the expected information supplied by one further output, $I_P'(\gamma)=\cR_{P,P}(\gamma)$, is nonincreasing in the observation budget. This monotonicity holds after averaging over observation records; it does not assert that the conditional prediction risk decreases after every realized sample.
\end{proof}

\noindent At the fixed-count level, the decreasing increments of \(I(X;Z^n)\) are the repeated-observation instance of submodularity of information gain under conditional independence \cite{krause2007submodular,krause2005nonmyopic}. At zero budget, this redundancy also distinguishes priors that are equally informative for one observation. Section~\ref{sec:capacity} uses it to determine the second-order capacity term; the designs in Section~\ref{sec:coupon-design} show why agreement to second order need not imply optimality at positive budgets.

\noindent The same dependence can be expressed through posterior fluctuations of the channel row. Fix a history $y$ of positive probability under $P$, and write $\pi(x)=P(X=x| Y_\gamma=y)$. Let $q_\pi$ denote the common predictive distribution of each fresh output and $r_\pi$ their joint predictive distribution:
\begin{align}
q_\pi(z)&=\sum_x\pi(x)W_x(z),\\
r_\pi(z,z')&=\sum_x\pi(x)W_x(z)W_x(z').
\end{align}
Then
\begin{equation}\label{eq:posterior-second}
I(Z_1';Z_2'| Y_\gamma=y)
=
\KL{r_\pi}{q_\pi\otimes q_\pi},
\end{equation}
and
\begin{equation}
r_\pi(z,z')-q_\pi(z)q_\pi(z')
=
\operatorname{Cov}_\pi(W_X(z),W_X(z')).
\end{equation}
Thus the difference between the joint predictive law and the product of its marginals is exactly the posterior covariance of the corresponding channel-row coordinates. Averaging \eqref{eq:posterior-second} over the observed record recovers $-I_P''(\gamma)$.

\noindent For the local expansion, let $\Z_\pi=\{z\in\Z:q_\pi(z)>0\}$. Pairs outside $\Z_\pi\times\Z_\pi$ have zero probability under both $r_\pi$ and $q_\pi\otimes q_\pi$ and can be omitted. When the relative deviations of $r_\pi$ from $q_\pi\otimes q_\pi$ are small on this support, the divergence admits the quadratic expansion
\begin{equation}\label{eq:quad-second}
I_P(Z_1';Z_2'| Y_\gamma=y)
=
\frac12\sum_{z,z'\in\Z_\pi}
\frac{\operatorname{Cov}_\pi(W_X(z),W_X(z'))^2}
{q_\pi(z)q_\pi(z')}
+\mathcal R.
\end{equation}
Appendix~\ref{app:covariance} gives a precise relative-deviation condition and an explicit bound on $\mathcal R$. This is a local approximation: full support alone does not justify a uniform unweighted remainder as predictive probabilities approach zero. This is the closest finite-alphabet analogue of the squared-posterior-variance expression in the second derivative of Gaussian I--MMSE \cite{guo2011gaussian}.

\subsection{A higher-order co-information hierarchy}
Repeated Poisson differentiation produces higher forward differences of the fixed-sample information; these differences can be expressed as alternating sums of conditional entropies of fresh outputs. For finite-valued random variables $U_1,\ldots,U_k$ and a conditioning variable $V$, write $[k]=\{1,\ldots,k\}$ and $U_S=(U_i)_{i\in S}$. Define conditional co-information \cite{mcgill1954multivariate} with the following sign convention:
\begin{equation}\label{eq:coinf-def}
\coI_k(U_1;\ldots;U_k| V)
\triangleq
\sum_{\emptyset\ne S\subseteq[k]}
(-1)^{|S|+1}H(U_S| V).
\end{equation}
Thus \(\coI_2(U_1;U_2| V)=I(U_1;U_2| V)\), while
\begin{equation}
\coI_3(U_1;U_2;U_3| V)
=I(U_1;U_2| V)-I(U_1;U_2| U_3,V).
\end{equation}

\begin{proposition}[Higher derivatives]\label{prop:higher}
For every integer \(k\ge2\),
\begin{equation}\label{eq:higher}
\boxed{
\frac{d^k}{d\gamma^k}I_P(X;Y_\gamma)
=
(-1)^{k-1}
\coI_k(Z_1';\ldots;Z_k'| Y_\gamma).
}
\end{equation}
In particular,
\begin{equation}\label{eq:third}
I_P^{(3)}(\gamma)
=
\coI_3(Z_1';Z_2';Z_3'| Y_\gamma).
\end{equation}
\end{proposition}

\begin{proof}
Let \(h_n=H(Z^n)\).  Conditional i.i.d. observations imply
\begin{equation}
I(X;Z^n)=h_n-nH(Z_1| X).
\end{equation}
For \(k\ge2\), the linear term disappears under \(\Delta^k\). By conditional exchangeability, the conditional entropy of any $j$ of the fresh outputs given $Z^n$ is $h_{n+j}-h_n$. Grouping the terms in \eqref{eq:coinf-def} by subset size therefore gives
\begin{equation}
\begin{aligned}
\coI_k(Z_{n+1};\ldots;Z_{n+k}| Z^n)
&=
\sum_{j=1}^k(-1)^{j+1}\binom{k}{j}(h_{n+j}-h_n)\\
&=
(-1)^{k-1}\Delta^k h_n.
\end{aligned}
\end{equation}
Since $\Delta^k I_n=\Delta^k h_n$ for $k\ge2$, applying Lemma~\ref{lem:poisson-diff} to the bounded sequence $I_n=I_P(X;Z^n)$ and conditioning on $N_\gamma$ proves \eqref{eq:higher}.
\end{proof}

\begin{remark}[No alternating-sign theorem beyond curvature]
The mutual-information curve is always nondecreasing and concave, but its higher derivatives do not in general alternate in sign.  This can be checked within the present model: for instance, consider
\[
P(X=1)=\frac1{10},\text{ } W_0=\left(\frac45,\frac15\right),\text{ } W_1=(0,1).
\]
Then
\[
I_P^{(3)}(0+)
=3H(Z_1)-3H(Z_1,Z_2)+H(Z_1,Z_2,Z_3)
\approx-0.01299748<0.
\]
For a noiseless channel with a nondegenerate prior, \(I_P(\gamma)=(1-e^{-\gamma})H(P)\), so \(I_P^{(3)}(\gamma)=e^{-\gamma}H(P)>0\).  Thus both signs occur already at the third derivative.  The hierarchy \eqref{eq:higher} is an exact identity, not a complete-monotonicity statement.
\end{remark}

\section{High-intensity limits and identifiability}\label{sec:high-intensity}

We now examine the information and accumulated prediction losses as the observation budget grows. We first use the common channel of Section~\ref{sec:prior-mismatch}, and then return to the nominal-channel predictor of Section~\ref{sec:joint-mismatch}. This distinguishes uncertainty that observations can remove from approximation error that remains in the active nominal row family.

\subsection{Row identifiability with a fixed channel}

The observation law depends on \(X\) only through the row \(W=W_X=M(\cdot| X)\).  Let \(P_W,Q_W\) be the induced priors on the finite set of distinct rows of \(M\).  The observation laws and predictive risks depend on \(P,Q\) only through these row priors.  Thus the preceding mismatch identities apply on the row alphabet whenever \(P_W\ll Q_W\), even if \(P\not\ll Q\) on input labels.

\begin{proposition}[High-intensity limit]\label{prop:high-intensity}
Assume \(P_W\ll Q_W\).  Then
\begin{equation}\label{eq:high-kl}
\lim_{\gamma\to\infty}\KL{P_{Y_\gamma}}{Q_{Y_\gamma}}
=
\KL{P_W}{Q_W}.
\end{equation}
Consequently,
\begin{equation}\label{eq:infinite-regret}
\boxed{
\KL{P_W}{Q_W}
=
\int_0^\infty
\left[\cR_{P,Q}(\gamma)-\cR_{P,P}(\gamma)\right]d\gamma.
}
\end{equation}
If the rows \(W_x\) are all distinct, \(W\) identifies \(X\) and the right side equals \(D(P\|Q)\).
\end{proposition}

\begin{proof}
Because the set of distinct rows is finite, nearest-row decoding of the empirical output frequencies is consistent under every row as \(N_\gamma\to\infty\) in probability.  Let \(\widehat W_\gamma(Y_\gamma)\) be such a decoder, assigning any row to the empty record. Unlike the posterior predictive row, this decoder returns one of the finitely many distinct channel rows. It is consistent under both row priors.  The Markov chain \(W\to Y_\gamma\to\widehat W_\gamma\) and data processing give
\begin{equation}
D(P_W\|Q_W)
\ge
D(P_{Y_\gamma}\|Q_{Y_\gamma})
\ge
D(P_{\widehat W_\gamma}\|Q_{\widehat W_\gamma}).
\end{equation}
Consistency gives coordinatewise convergence of the decoded-row laws to \(P_W,Q_W\).  Lower semicontinuity of relative entropy, including at zero coordinates, therefore gives
\[
\liminf_{\gamma\to\infty}
D(P_{\widehat W_\gamma}\|Q_{\widehat W_\gamma})
\ge D(P_W\|Q_W).
\]
The squeeze proves \eqref{eq:high-kl}.  Equation \eqref{eq:infinite-regret} follows from \eqref{eq:out-div-int} by letting the upper integration limit tend to infinity, since the excess risk is nonnegative by Proposition~\ref{prop:predictive-disagreement}.
\end{proof}

The matched limit does not require any assumption on $Q$. Under the true law, the same consistent decoder has vanishing row error probability. Fano's inequality on the finite row alphabet therefore gives \(H(W| Y_\gamma)\to0\). Thus the matched limit is
\begin{equation}\label{eq:entropy-row}
\boxed{
H(W)=\int_0^\infty \cR_{P,P}(\gamma)\,d\gamma,
}
\end{equation}
because \(I(X;Y_\gamma)\to I(X;W)=H(W)\).  In particular, if \(x\mapsto W_x\) is injective, \(H(W)=H(X)\).

\subsection{KL projection under channel mismatch}

For the nominal-channel predictor, the identity \eqref{cm:eq:Dout-integrated} turns long-run approximation error into growth of output relative entropy. The relevant approximation is the KL projection of the true row onto the \emph{active} nominal row family. The resulting posterior behavior is a finite-alphabet instance of Bayesian misspecification studied by Berk~\cite{berk1966posterior}. Write
\[
\mathcal A_Q=\{j\in\X:Q(j)>0\},
\]
and, for each \(x\) with \(P(x)>0\), define
\begin{equation}\label{cm:eq:dx}
d(x)\triangleq\min_{j\in\mathcal A_Q}\KL{W_x}{\widetilde W_j},
\text{ }
\bar d\triangleq\sum_{x:P(x)>0}P(x)d(x).
\end{equation}
For observations generated from $W_x$, minimizing this divergence is equivalent to maximizing the limiting average log likelihood assigned by a nominal row. The minimization is over individual active rows, not arbitrary mixtures of them: the nominal model also draws one latent input and uses its row to generate the entire record. For the output and predictive-loss statements in this section, it suffices that \(d(x)<\infty\) for every \(P\)-supported \(x\).  Other pairwise divergences may be infinite, and same-label domination is not needed.  Same-label assumptions will be invoked explicitly when using the posterior conservation law.

\begin{proposition}[High-intensity output-divergence slope]\label{cm:prop:high}
Assume \(d(x)<\infty\) whenever \(P(x)>0\).  Then
\begin{equation}\label{cm:eq:high-Dout}
\boxed{
\lim_{\gamma\to\infty}\frac{\Dout(\gamma)}{\gamma}=\bar d.
}
\end{equation}
If, in addition, for each such \(x\) the active-label minimizer
\[
j^*(x)=\arg\min_{j\in\mathcal A_Q}\KL{W_x}{\widetilde W_j}
\]
is unique, then, conditionally on \(X=x\), the nominal posterior predictive distribution converges in probability to \(\widetilde W_{j^*(x)}\), and
\begin{equation}\label{cm:eq:high-risk}
\boxed{
\lim_{\gamma\to\infty}\Rmis(\gamma)=\bar d.
}
\end{equation}
\end{proposition}

\begin{proof}
The logarithm of the finite nominal likelihood mixture is asymptotically controlled by its KL-minimizing active row. Bounded normalized log likelihood ratios give convergence in expectation and preserve the slope under Poisson averaging. With a unique minimizer, exponentially small posterior-concentration errors and a linear bound on exceptional-history losses give the expected-risk limit. Appendix~\ref{app:channel-high} provides the full argument, including zero entries.
\end{proof}

\subsection{Representable rows and posterior discrepancy}

A zero asymptotic slope in \eqref{cm:eq:high-Dout} does not by itself establish a finite limit for $\Dout(\gamma)$. We now show that exact representability of every true row by a nominal row with positive prior mass gives this stronger result. Define probability measures on the finite union of true and nominal row values by
\[
P_W(v)=\sum_{x:W_x=v}P(x),\text{ }
Q_{\widetilde W}(v)=\sum_{j:\widetilde W_j=v}Q(j).
\]
Thus identical rows are grouped even when their latent labels differ.

\begin{proposition}[Finite divergence limit under row representability]\label{cm:prop:representable}
Suppose every row \(W_x\) with \(P(x)>0\) equals \(\widetilde W_j\) for at least one \(j\in\mathcal A_Q\).  Equivalently, \(P_W\ll Q_{\widetilde W}\).  Then
\begin{equation}\label{cm:eq:representable-Dout}
\boxed{
\lim_{\gamma\to\infty}\Dout(\gamma)
=\KL{P_W}{Q_{\widetilde W}}<\infty.
}
\end{equation}
No uniqueness of the nominal latent label is required.  Consequently,
\begin{equation}\label{cm:eq:representable-regret}
\boxed{
\int_0^\infty
\E_{P,M}\KL{\widehat W_{P,\gamma}}
{\widehat{\widetilde W}_{Q,\gamma}}\,d\gamma
=\KL{P_W}{Q_{\widetilde W}}.
}
\end{equation}
\end{proposition}

\begin{proof}
Regard the finite union of true and nominal row values as an input alphabet, with the conditional output distribution at input \(v\) equal to \(v\), and with priors \(P_W,Q_{\widetilde W}\). Conditional on the row value, both models use the same repeated-sampling experiment. Since \(P_W\ll Q_{\widetilde W}\), Proposition~\ref{prop:high-intensity} gives \eqref{cm:eq:representable-Dout}. Letting the upper integration limit in \eqref{cm:eq:Dout-integrated} tend to infinity then gives \eqref{cm:eq:representable-regret}, since its integrand is nonnegative.
\end{proof}

In the prior-only case $\widetilde M=M$, the row priors are $P_W$ and $Q_W$. The finite total regret is therefore $D(P_W\|Q_W)$, which equals $D(P\|Q)$ when the channel rows are pairwise distinct. Irreducible channel mismatch need not have a finite total regret, as shown below.

For finite row families with \(d(x)<\infty\), representability is equivalent to \(\bar d=0\).  Otherwise some \(P\)-supported true row is absent from the active nominal family, so \(\bar d>0\) and
\[
\Dout(\gamma)=\gamma\bar d+o(\gamma),\text{ }
\int_0^\infty
\E_{P,M}\KL{\widehat W_{P,\gamma}}
{\widehat{\widetilde W}_{Q,\gamma}}\,d\gamma=\infty.
\]
These divergence and integrated-regret conclusions do not require a unique KL minimizer.  The instantaneous risk limit \eqref{cm:eq:high-risk} is asserted under the uniqueness condition of Proposition~\ref{cm:prop:high}. \\
\textbf{Remaining posterior discrepancy:}

If \(D(P\|Q)<\infty\) and
\(\kappa(P;M,\widetilde M)<\infty\), the same-label nominal row is active and admissible for every \(P\)-supported \(x\).  Therefore
\[
\bar d\le\kappa(P;M,\widetilde M).
\]
Dividing the conservation law \eqref{cm:eq:conservation} by \(\gamma\) and using \eqref{cm:eq:high-Dout} gives the further limit
\begin{equation}\label{cm:eq:high-posterior-slope}
\boxed{
\lim_{\gamma\to\infty}\frac{\Dpost(\gamma)}{\gamma}
=\kappa(P;M,\widetilde M)-\bar d.
}
\end{equation}
This limit describes the linear growth rate of the expected posterior divergence. A zero rate does not imply that the divergence itself tends to zero. Thus the output-distribution slope measures the best active-row approximation error, while the remainder of the same-label discrepancy appears in the posterior divergence.  In particular, matching the observable row family need not make the two latent-label posteriors agree.

\section{Capacity of the static Poissonized DMC}\label{sec:capacity}

We return to the matched experiment and optimize mutual information over the input prior for the fixed true channel \(M\). The prediction representation therefore uses the matched risk \(\cR_{P,P}\) from \eqref{eq:matched-risk}. For fixed \(M\) and \(\gamma\), regard the experiment \(x\mapsto Y_\gamma\) itself as a channel.  Let
\begin{equation}
\mathsf C_M(\gamma)\triangleq \max_{P\in\Delta(\X)} I_P(X;Y_\gamma)
\end{equation}
be its capacity in nats per use of the static experiment. One use selects a single input and observes its entire Poisson-sized output record. The budget $\gamma$ is prescribed; only the input prior is optimized. This differs from waveform Poisson-channel capacity, where the input controls a time-dependent intensity and information is measured per unit time \cite{kabanov1978capacity}.

\noindent The output alphabet is countably infinite, but the input prior ranges over a finite-dimensional simplex.  To justify attainment, write \(I_n(P)=I_P(X;Z^n)\).  For every fixed \(n\), this is continuous in \(P\), and \(0\le I_n(P)\le\log|\X|\).  The Poisson series for \(I_P(X;Y_\gamma)\), and its derivative series with coefficients \(I_{n+1}(P)-I_n(P)\), therefore converge uniformly in \(P\), locally uniformly in \(\gamma\ge0\).  Thus the information and its intensity derivative are jointly continuous in \((P,\gamma)\), with a right derivative at zero.  Compactness of the input simplex gives a maximizing prior at every budget.

\noindent Using the count-vector sufficient statistic \(C=(C_z:z\in\Z)\), for \(\gamma>0\) the induced channel law is explicit:
\begin{equation}\label{eq:cap-induced-channel}
P_\gamma(c| x)
=e^{-\gamma}\frac{\gamma^{|c|}}{\prod_{z\in\Z}c_z!}
\prod_{z\in\Z}W_x(z)^{c_z},
\text{ } |c|=\sum_z c_z.
\end{equation}
Here a zero count contributes the factor \(W_x(z)^0=1\), including when \(W_x(z)=0\).  At \(\gamma=0\), the count vector is zero almost surely under every input.  Thus
\begin{equation}\label{eq:cap-explicit}
\mathsf C_M(\gamma)
=\max_P\sum_xP(x)\sum_c P_\gamma(c| x)
\log\frac{P_\gamma(c| x)}{\sum_{x'}P(x')P_\gamma(c| x')}.
\end{equation}
Equivalently, conditioning on \(N_\gamma\),
\begin{equation}\label{eq:cap-poisson-mixture}
\mathsf C_M(\gamma)
=\max_P e^{-\gamma}\sum_{n\ge0}\frac{\gamma^n}{n!}I_P(X;Z^n).
\end{equation}
The maximization remains outside the Poisson average: one input prior must serve all possible observation counts.

\subsection{Estimation envelope and capacity geometry}

The first-order identity immediately gives an estimation-theoretic representation of capacity.

\begin{proposition}[Capacity as accumulated predictive risk]\label{prop:cap-envelope}
For every \(\gamma\ge0\),
\begin{equation}\label{eq:cap-envelope}
\boxed{
\mathsf C_M(\gamma)
=\max_{P\in\Delta(\X)}
\int_0^\gamma \cR_{P,P}(s)\,ds
=\max_P\int_0^\gamma
\E_P\!\left[\KL{W}{\E_P[W| Y_s]}\right]ds.
}
\end{equation}
\end{proposition}

\noindent Each candidate prior $P$ is evaluated using its own matched predictor over $[0,\gamma]$, so capacity selects the input distribution with the largest accumulated matched prediction risk. Since $P$ is held fixed across the entire interval, the maximization cannot in general be moved inside the integral. \\
There is also a useful geometric identity.  Let \(P_{\gamma,x}\) denote the output law in \eqref{eq:cap-induced-channel} given \(X=x\).  For \(\gamma>0\), its product-Poisson form gives
\begin{equation}\label{eq:pairwise-poisson-kl}
\boxed{
\KL{P_{\gamma,x}}{P_{\gamma,x'}}
=\gamma\KL{W_x}{W_{x'}}.
}
\end{equation}
This equality includes infinite divergences.  At \(\gamma=0\), both output laws are the point mass at the zero count vector and their divergence is zero.
The induced channel has the standard information-radius representation
\begin{equation}\label{eq:cap-radius}
\mathsf C_M(\gamma)
=\inf_Q\max_{x\in\X}\KL{P_{\gamma,x}}{Q},
\end{equation}
where \(Q\) ranges over output distributions on count vectors.  Restricting the center \(Q\) to product-Poisson laws with mean vector \(\gamma q\), and also degrading the observation to its first mark when at least one arrival occurs, gives the simple bounds
\begin{equation}\label{eq:cap-bounds}
\boxed{
(1-e^{-\gamma})C(M)
\le \mathsf C_M(\gamma)
\le \min\!\left\{\gamma C(M),\,(1-e^{-\gamma})\log r\right\},
}
\end{equation}
where \(C(M)=\max_P I_P(X;Z)\) is the ordinary one-use DMC capacity and \(r\) is the number of distinct rows of \(M\).

\begin{proof}[Proof of the bounds]
At \(\gamma=0\), all terms in the bounds are zero.  For \(\gamma>0\), map the full observation to an erasure if \(N_\gamma=0\), and otherwise retain only the first DMC output.  This degraded channel is an independent erasure of one use of \(M\), and has capacity \((1-e^{-\gamma})C(M)\), proving the lower bound.  For the first upper bound, choose in \eqref{eq:cap-radius} a product-Poisson center with mean \(\gamma q\); then \(D(P_{\gamma,x}\|Q_{\gamma,q})=\gamma D(W_x\|q)\), including infinite values, and minimizing over \(q\) gives \(\gamma C(M)\).  Finally, conditional on \(N_\gamma>0\), the observation can reveal at most which of the \(r\) distinct rows generated it, yielding the second upper bound.
\end{proof}

\subsection{Optimality conditions in prediction form}

For a prior $P$, write
\begin{equation*}
P_Y^{P,\gamma}=\sum_{x'}P(x')P_{\gamma,x'}
\end{equation*}
for the induced distribution of the count vector, and define
\begin{equation}
G_x(P,\gamma)
\triangleq
\KL{P_{\gamma,x}}{P_Y^{P,\gamma}}.
\end{equation}
Thus $G_x(P,\gamma)$ compares the output law generated by the fixed input $x$ with the overall output law induced by $P$. The usual capacity KKT condition \cite{cover2006elements} states that \(P^*\) is capacity achieving if and only if
\begin{equation}\label{eq:cap-kkt-standard}
G_x(P^*,\gamma)\le \mathsf C_M(\gamma)
\end{equation}
for every \(x\), with equality for every \(x\) in the support of \(P^*\).

\noindent The prediction identities turn \(G_x\) into an accumulated estimation loss.  Let \(\E_x\) mean expectation under the observation experiment with the input fixed to \(x\), even if \(P(x)=0\); the predictor continues to use the prior \(P\).  At positive budgets, whenever \(G_x(P,\gamma)\) is finite,
\begin{equation}\label{eq:cap-kkt-predictive}
\boxed{
G_x(P,\gamma)
=
\int_0^\gamma
\E_x\!\left[
\KL{W_x}{\widehat W_P^s}
\right]ds,
\text{ }
\widehat W_P^s=\E_P[W| Y_s].
}
\end{equation}
The required support condition is that some \(x'\) with \(P(x')>0\) satisfies \(\operatorname{supp}W_x\subseteq\operatorname{supp}W_{x'}\). It is equivalent to finiteness of \(G_x(P,\gamma)\) at any positive budget. If it fails, both the divergence and the integrated predictive loss are infinite at every positive budget. At zero, both sides are zero. Appendix~\ref{app:capacity-support} proves these claims and the differentiation bound, including inputs outside the prior support.

\noindent Consequently, a prior $P^*$ is capacity achieving if and only if
\begin{equation}\label{eq:cap-kkt-estimation}
\boxed{
\int_0^\gamma
\E_x\left[
\KL{W_x}{\widehat W_{P^*}^s}
\right]ds
\le \mathsf C_M(\gamma)
}
\end{equation}
for every $x\in\X$, with equality whenever $P^*(x)>0$.

\noindent Thus all channel rows actually used at capacity have equal \emph{accumulated predictive KL loss} against the Bayes predictor induced by the optimal prior.

\subsection{Derivative and endpoint behavior of capacity}

The prior maximizing mutual information may depend on the observation budget. Let
\begin{equation}
\mathcal P_\gamma^*
\triangleq
\operatorname*{argmax}_{P\in\Delta(\X)}I_P(X;Y_\gamma)
\end{equation}
denote the set of maximizing priors at budget $\gamma$. To account for this optimization, we first describe the one-sided derivatives of capacity. Compactness of the input simplex and the joint continuity established above give
\begin{align}
\mathsf C_{M,+}'(\gamma)
&=\max_{P\in\mathcal P_\gamma^*}\cR_{P,P}(\gamma),\text{ } \gamma\ge0,\label{eq:cap-right-derivative}\\
\mathsf C_{M,-}'(\gamma)
&=\min_{P\in\mathcal P_\gamma^*}\cR_{P,P}(\gamma),\text{ } \gamma>0.\label{eq:cap-left-derivative}
\end{align}
For the right derivative, bounding the capacity difference quotient below by a maximizer at $\gamma$ and above by a maximizer at $\gamma+h$ yields matching limits as $h\downarrow 0$ by compactness and joint continuity of the intensity derivative while the left derivative follows similarly at $\gamma-h$. In particular, whenever the maximizing prior $P_\gamma^*$ is unique at $\gamma>0$,
\begin{equation}\label{eq:cap-derivative}
\boxed{
\mathsf C_M'(\gamma)
=\E_{P_\gamma^*}\!\left[
\KL{W}{\E_{P_\gamma^*}[W| Y_\gamma]}
\right].
}
\end{equation}
Thus the slope of the capacity curve is the predictive KL risk for the optimizing prior at that budget.  The intensity derivative is evaluated with this prior held fixed; no derivative of the optimizing prior is taken. At zero intensity, no observation has yet arrived and
\begin{equation}
\widehat W_P^0=\E_P W.
\end{equation}
Hence
\begin{equation}
\cR_{P,P}(0)
=\sum_xP(x)\KL{W_x}{\sum_{x'}P(x')W_{x'}}
=I_P(X;Z).
\end{equation}
Applying the right-envelope formula at \(\gamma=0\) gives the particularly simple identity
\begin{equation}\label{eq:cap-zero-slope}
\boxed{
\mathsf C_M'(0+)=C(M).
}
\end{equation}
Therefore
\begin{equation}\label{eq:cap-small-gamma}
\mathsf C_M(\gamma)=\gamma C(M)+o(\gamma).
\end{equation}

\noindent To examine the next capacity term, write
\[
A(P)=I_P(X;Z_1),\text{ } B(P)=I_P(Z_1;Z_2),
\]
where \(Z_1,Z_2\) are conditionally independent outputs of \(M\) given \(X\sim P\).
Let
\[
\mathcal A=\operatorname*{argmax}_{P\in\Delta(\X)} A(P)
\]
be the set of priors achieving the ordinary one-use capacity \(C(M)\). Several priors can belong to \(\mathcal A\), while giving different amounts of redundant information in two observations. The next result identifies which of them determine the second-order capacity term.

\begin{theorem}[Minimum-redundancy selection]\label{thm:cap-selection}
Define
\[
B_*=\min_{P\in\mathcal A} B(P),\text{ }
\mathcal A_2=\operatorname*{argmin}_{P\in\mathcal A}B(P).
\]
Then
\begin{equation}\label{eq:cap-second-order}
\boxed{
\mathsf C_M(\gamma)
=\gamma C(M)-\frac{\gamma^2}{2}B_*+o(\gamma^2),
\text{ } \gamma\downarrow0.
}
\end{equation}
For any sequence of positive budgets tending to zero, every accumulation
point of corresponding capacity-achieving priors belongs to \(\mathcal A_2\).
\end{theorem}

\begin{proof}
The curvature identity and Lemma~\ref{lem:poisson-diff} give
\begin{equation}\label{eq:cap-uniform-expansion}
I_P(X;Y_\gamma)
=\gamma A(P)-\frac{\gamma^2}{2}B(P)+R_P(\gamma),
\text{ } |R_P(\gamma)|\le K\gamma^3,
\end{equation}
uniformly over all priors, with \(K=\frac43\log|\X|\).
Indeed, \(0\le I_n(P)\le\log|\X|\) implies
\(|\Delta^3 I_n(P)|\le8\log|\X|\), so the third intensity
derivative has the same bound. The finite-alphabet mutual informations
\(A\) and \(B\) are continuous on the input simplex. Compactness therefore
ensures that \(\mathcal A\) and \(\mathcal A_2\) are nonempty.

\noindent Choose a capacity-achieving prior \(P_\gamma\) at each positive budget
and fix \(P_*\in\mathcal A_2\). Comparing their uniform expansions gives
\begin{equation}\label{eq:cap-selection-bound}
0\le C(M)-A(P_\gamma)
\le\frac{\gamma}{2}\bigl(B_*-B(P_\gamma)\bigr)+2K\gamma^2.
\end{equation}
Since \(B\) is bounded, \(A(P_\gamma)\to C(M)\). Thus every accumulation
point lies in \(\mathcal A\).
The same inequality gives \(B(P_\gamma)\le B_*+4K\gamma\).
Continuity now places every accumulation point in \(\mathcal A_2\);
in particular \(B(P_\gamma)\to B_*\).
Dividing \eqref{eq:cap-selection-bound} by \(\gamma\) then shows
\(C(M)-A(P_\gamma)=o(\gamma)\).
Substitution in \eqref{eq:cap-uniform-expansion} proves
\eqref{eq:cap-second-order}.
\end{proof}

\noindent If the one-use capacity-achieving prior is unique, say \(P^*\), then \(B_*=I_{P^*}(Z_1;Z_2)\), recovering the unique-prior formula. The minimization is restricted to \(\mathcal A\): reducing redundancy by choosing an uninformative prior does not improve capacity. The theorem describes the small-budget limit; it does not assert that every member of \(\mathcal A_2\) achieves capacity at a positive budget. The secondary minimization also has a convex formulation.

\begin{proposition}[Convex formulation of the selection rule]
\label{prop:cap-convex-selection}
All priors in \(\mathcal A\) induce the same one-output law \(q^*\).
Writing
\[
r_P(z,z')=\sum_xP(x)W_x(z)W_x(z'),
\]
the minimum redundancy is
\begin{equation}\label{eq:cap-convex-selection}
\boxed{
B_*=\min_{P\in\mathcal A}
\KL{r_P}{q^*\otimes q^*}.
}
\end{equation}
This is a convex optimization over \(P\).
The minimizing two-output law \(r_P\) is unique, although the prior
realizing it need not be.
\end{proposition}

\begin{proof}
For \(q_P=\sum_xP(x)W_x\),
\[
A(P)=H(q_P)-\sum_xP(x)H(W_x).
\]
Concavity makes \(\mathcal A\) convex. If two maximizing priors had
different output laws, strict concavity of entropy would give a
strictly larger value at their mixture, a contradiction.
Thus their common output law is \(q^*\).
For \(P\in\mathcal A\), both marginals of \(r_P\) are \(q^*\), so
\(B(P)=D(r_P\|q^*\otimes q^*)\).
The map \(P\mapsto r_P\) is linear.
After omitting coordinates outside
\(\operatorname{supp}(q^*)\times\operatorname{supp}(q^*)\), which vanish
for every feasible \(r_P\), the reference law is strictly positive.
Relative entropy is convex in \(P\) through this linear map and
strictly convex in \(r_P\). Hence all minimizers induce the same
two-output law.
\end{proof}

\noindent Finally, the matched row-identifiability limit from Section~\ref{sec:high-intensity} determines the large-budget capacity.

\begin{corollary}[Capacity saturation]\label{cor:cap-saturation}
If \(r\) is the number of distinct rows of \(M\), then
\begin{equation}\label{eq:cap-infinity}
\boxed{
\lim_{\gamma\to\infty}\mathsf C_M(\gamma)=\log r.
}
\end{equation}
If all rows are distinct, the limit is \(\log|\X|\).
\end{corollary}
\begin{proof}
Choose a fixed prior inducing the uniform law on the \(r\) distinct rows.  By \eqref{eq:entropy-row}, its information tends to \(\log r\), giving a lower bound on the capacity limit.  For every prior \(P\) and every \(\gamma\),
\[
I_P(X;Y_\gamma)=I_P(W;Y_\gamma)\le H_P(W)\le\log r.
\]
These bounds give \eqref{eq:cap-infinity}.
\end{proof}

\subsection{Coupon-collector channels and balanced input designs}
\label{sec:coupon-design}
The selection rule above becomes explicit for a channel whose input
is a subset of symbols and whose output is a uniformly selected member of that subset. We ask how much information is lost when the input distribution uses only a smaller family of subsets. The answer depends on which smaller subsets appear in the observations, and how evenly the chosen input family represents them.

\subsubsection{Model and exact information loss}
Fix integers \(2\le k<m\). The input alphabet consists of all
\(k\)-element subsets of \([m]=\{1,\ldots,m\}\), and
\[
W_S(z)=\frac1k\ind\{z\in S\}.
\]
This is the interference-free coupon-collector channel
of~\cite{sokolovskii2025coupon}, here observed a Poisson number of times. An input \(S\sim P\) is selected once; conditional on \(S\), all reads sample independently and uniformly from \(S\). Write \(I_P(\gamma)=I_P(S;Y_\gamma)\) and \(\mathsf C_{m,k}(\gamma)=\max_P I_P(\gamma)\). Let \(U\) be uniform on all \(\binom mk\) possible inputs. Permuting the \(m\) symbols leaves the channel unchanged.
Averaging any input prior over these permutations yields $U$, and using concavity of mutual information in the prior shows that $U$ achieves capacity at every budget.

\noindent Let \(T\) be the set of distinct symbols observed and let \(R=|T|\). The total number of reads \(N_\gamma\) counts repetitions; \(R\) does not. For a given input \(S\), Poisson splitting gives independent \(\Pois(\gamma/k)\) counts for its \(k\) members. Thus, with
\[
p=1-e^{-\gamma/k},
\]
each member is observed with probability \(p\), independently of the
others, and \(R\sim\operatorname{Bin}(k,p)\).
This distribution does not depend on \(S\) or on its prior.

\noindent The set \(T\) is sufficient for the hidden input.
A record of length \(n\) has likelihood zero if it contains
a symbol outside \(S\), and otherwise has likelihood
\(e^{-\gamma}\gamma^n/(n!k^n)\), the same for every compatible input. Under \(U\), the posterior is therefore uniform over the
\(\binom{m-R}{k-R}\) subsets containing \(T\). It follows that
\begin{equation}\label{eq:coupon-capacity}
\mathsf C_{m,k}(\gamma)
=\log\binom mk-
\sum_{r=0}^{k}\binom kr p^r(1-p)^{k-r}
\log\binom{m-r}{k-r}.
\end{equation}
This is the Poisson average of the fixed-read capacity
in~\cite{sokolovskii2025coupon}. The same prior \(U\) maximizes
information for every fixed number of reads, so averaging those
capacities is valid for this channel.

\noindent To describe a general prior, define its distribution on \(r\)-subsets by
\begin{equation}\label{eq:coupon-marginal}
\mu_{P,r}(T)=\frac{\Pr_P(T\subseteq S)}{\binom kr},
\text{ } |T|=r,\text{ } 0\le r\le k.
\end{equation}
Equivalently, first draw \(S\sim P\), then choose an \(r\)-subset of
\(S\) uniformly. Let \(u_r\) be uniform on all \(\binom mr\)
\(r\)-subsets and put
\[
d_r(P)=D(\mu_{P,r}\|u_r).
\]
Here \(d_0(P)=0\) and \(d_k(P)=D(P\|U)\).
The next identity separates the probability of observing \(r\)
distinct symbols from the imbalance in their distribution.

\begin{theorem}[Exact information deficit]\label{thm:coupon-deficit}
For every input prior \(P\) and every \(\gamma\ge0\),
\begin{equation}\label{eq:coupon-deficit}
\boxed{
\mathsf C_{m,k}(\gamma)-I_P(\gamma)
=\sum_{r=1}^{k}\binom kr p^r(1-p)^{k-r}d_r(P).
}
\end{equation}
\end{theorem}

\begin{proof}
At zero budget both sides vanish. For \(\gamma>0\), every
\(r\)-subset of a fixed \(S\) is the observed set with probability
\(p^r(1-p)^{k-r}\). Consequently \(T\), conditional on \(R=r\),
has distribution \(\mu_{P,r}\), while \(T\), conditional on \(S\)
and \(R=r\), is uniform over the \(\binom kr\) compatible subsets.
Sufficiency of \(T\) and independence of \(R\) from \(S\) gives
\[
I_P(\gamma)=
\sum_{r=0}^{k}\Pr(R=r)
\left[H(\mu_{P,r})-\log\binom kr\right].
\]
For \(P=U\), each \(\mu_{U,r}=u_r\).
Subtracting the two expressions and using
\(\log\binom mr-H(\mu_{P,r})=d_r(P)\) proves the result.
\end{proof}
\noindent Every term is nonnegative. Moreover, for each \(\gamma>0\), the
\(r=k\) term is \(p^kD(P\|U)\). Thus \(U\) is the unique
capacity-achieving prior at every positive Poisson budget.
A smaller input family can nevertheless preserve several leading
terms of the capacity expansion, as we now show.

\subsubsection{Balance conditions, fixed reads, and support size}
Balancing individual symbols means \(\Pr_P(i\in S)=k/m\) for every \(i\). Balancing pairs imposes the corresponding uniformity on two-element subsets. More generally, for \(1\le t\le k\), call \(P\) a
\emph{weighted \(t\)-design} if \(\mu_{P,r}=u_r\) for every
\(1\le r\le t\).
The weights are the input probabilities \(P(S)\).
It suffices to check \(\mu_{P,t}=u_t\): selecting a smaller subset
uniformly from a uniform \(t\)-subset gives the uniform smaller-subset
law. If \(P\) is uniform on a family of \(k\)-subsets, this is the
usual design condition that each \(t\)-subset belongs to the same
number of members of that family.

\begin{corollary}[Order of the information deficit]\label{cor:coupon-order}
Fix an input prior \(P\), independent of \(\gamma\).
For \(1\le t<k\),
\[
\mathsf C_{m,k}(\gamma)-I_P(\gamma)=O(\gamma^{t+1})
\text{ }\text{as }\gamma\downarrow0
\]
if and only if \(P\) is a weighted \(t\)-design.
If \(P\ne U\) and \(r_*\) is the smallest index with \(d_{r_*}(P)>0\),
then
\begin{equation}\label{eq:coupon-leading-deficit}
\mathsf C_{m,k}(\gamma)-I_P(\gamma)
=\frac{\binom{k}{r_*}}{k^{r_*}}d_{r_*}(P)\gamma^{r_*}
+O(\gamma^{r_*+1}).
\end{equation}
In particular, every weighted \((k-1)\)-design satisfies
\begin{equation}\label{eq:coupon-top-deficit}
\mathsf C_{m,k}(\gamma)-I_P(\gamma)
=(1-e^{-\gamma/k})^kD(P\|U)
\end{equation}
at every budget.
\end{corollary}

\begin{proof}
Since \(p=\gamma/k+O(\gamma^2)\), the \(r\)th weight in
\eqref{eq:coupon-deficit} is
\(\binom kr\gamma^r/k^r+O(\gamma^{r+1})\).
The first nonzero divergence therefore gives
\eqref{eq:coupon-leading-deficit}; its coefficient is strictly
positive and cannot be canceled by higher-order terms.
The order criterion follows. For a weighted \((k-1)\)-design
only the \(r=k\) term remains, giving \eqref{eq:coupon-top-deficit}.
\end{proof}
\noindent The same balance condition has a fixed-read interpretation.
Let \(\mathsf C_{m,k}^{[n]}\) denote the capacity with exactly \(n\ge1\)
reads. Conditional on observing \(r\) distinct members, their set
again has distribution \(\mu_{P,r}\). The preceding entropy argument gives
\begin{equation}\label{eq:coupon-fixed-deficit}
\mathsf C_{m,k}^{[n]}-I_P(S;Z_1,\ldots,Z_n)
=\sum_{r=1}^{\min(n,k)}w_{n,r}d_r(P),
\end{equation}
where \(w_{n,r}>0\) is the probability that \(n\) independent
uniform draws from a \(k\)-set contain exactly \(r\) distinct members.
Consequently, for \(1\le t\le k\), \(P\) achieves full-input
capacity with exactly \(t\) reads if and only if it is a weighted
\(t\)-design. It then also achieves capacity for every smaller
positive read count. At \(t=k\), this condition forces \(P=U\).
Poissonization gives the explicit binomial weights in
\eqref{eq:coupon-deficit}, translating these fixed-read equalities
into an order of information loss at small mean budgets.

The connection to Theorem~\ref{thm:cap-selection} can be made explicit.
Here \(A(P)=H(Z_1)-\log k\), so \(C(M)=\log(m/k)\), and
\(\mathcal A\) consists exactly of the balanced one-symbol priors.
For any such prior, direct evaluation of the two-output law gives
\begin{equation}\label{eq:coupon-redundancy}
B(P)=B_*+\frac{k-1}{k}d_2(P),
\end{equation}
where
\[
B_*=\frac1k\log\frac mk+
\frac{k-1}{k}\log\frac{m(k-1)}{k(m-1)}.
\]
Indeed, the diagonal probabilities are \(1/(mk)\), and for \(i\ne j\)
the probability of the ordered pair \((i,j)\) is
\((k-1)\mu_{P,2}(\{i,j\})/(2k)\).
Substituting these probabilities into
\(D(r_P\|u_1\otimes u_1)\) yields \eqref{eq:coupon-redundancy}.
Thus the minimum-redundancy priors are precisely the weighted
\(2\)-designs. For \(k=2\), this already forces \(P=U\).
For larger \(k\), higher subset marginals can still distinguish
priors with the same optimal first two information terms.

\begin{corollary}[Minimum support for second-order agreement]
\label{cor:coupon-support}
Any fixed prior satisfying
\[
\mathsf C_{m,k}(\gamma)-I_P(\gamma)=O(\gamma^3)
\text{ }\text{as }\gamma\downarrow0
\]
has at least \(m\) inputs in its support.
\end{corollary}

\begin{proof}
The exact deficit forces \(d_1(P)=d_2(P)=0\).
For \(k=2\), the second condition is \(P=U\); the same argument
below applies.
Let \(V_S\in\{0,1\}^m\) be the incidence vector of \(S\).
The singleton and pair balance conditions give
\[
\E_P[V_SV_S^{\mathsf T}]
=\frac{k(m-k)}{m(m-1)}I_m+
\frac{k(k-1)}{m(m-1)}\boldsymbol1\boldsymbol1^{\mathsf T}.
\]
This matrix has rank \(m\), since \(k<m\).
Its representation
\(\sum_{S:P(S)>0}P(S)V_SV_S^{\mathsf T}\) is a sum of
\(|\operatorname{supp}P|\) rank-one matrices, so the support
must have size at least \(m\).
\end{proof}

\noindent This is a weighted form of Fisher's design
inequality~\cite{fisher1940blocks}.
Here it bounds the number of allowed inputs needed to preserve
the linear and quadratic capacity terms, even when those inputs
may be assigned unequal probabilities.

\subsubsection{Matching and Fano examples}
First take \(m=4\) and \(k=2\).
The full input alphabet consists of six pairs.
Compare its uniform prior with the prior assigning probability
\(1/2\) to each of \(\{1,2\}\) and \(\{3,4\}\).
Each symbol occurs with probability \(1/2\) under both priors,
so both achieve one-use capacity \(\log2\).
Their pair distributions differ: for the two-input prior,
\(d_1=0\) and \(d_2=\log3\). Hence
\begin{equation}\label{eq:coupon-matching-gap}
\mathsf C_{4,2}(\gamma)-I_{\mathrm{match}}(\gamma)
=\log3\,(1-e^{-\gamma/2})^2
=\frac{\log3}{4}\gamma^2+O(\gamma^3).
\end{equation}
One read identifies the input under the two-input prior, giving
\[
I_{\mathrm{match}}(\gamma)=(1-e^{-\gamma})\log2.
\]
Under the full prior, one distinct observed member leaves three
possible input pairs; two distinct members determine the pair.
The conditional information gains agree when there is one distinct
member and differ by \(\log3\) when there are two.
The probability of the latter event is \((1-e^{-\gamma/2})^2\),
which explains the exact deficit.

\noindent More explicitly, every one-use capacity-achieving prior is a mixture of the three perfect matchings
\[
\{\{1,2\},\{3,4\}\},\text{ }
\{\{1,3\},\{2,4\}\},\text{ }
\{\{1,4\},\{2,3\}\}.
\]
The one-symbol balance equations force equal masses on the two
pairs in each matching. If \(v\) is the distribution over matchings
and the pair is then chosen uniformly within the selected matching,
\(H(P)=\log2+H(v)\) and
\begin{equation}\label{eq:coupon-matching-family}
I_v(\gamma)
=(1-e^{-\gamma})\log2+
(1-e^{-\gamma/2})^2H(v).
\end{equation}
To obtain this expression, use
\(D(P\|U)=\log6-H(P)=\log3-H(v)\) in
\eqref{eq:coupon-top-deficit} and subtract the matching case.
The uniform mixture is therefore the unique optimizer at every
positive budget. In this example the second-order selection rule
already identifies the full optimal prior.
\noindent Next take \(m=7\) and \(k=3\). The full alphabet contains 35 triples. Use instead the uniform prior on the seven triples
\[
123,\ 145,\ 167,\ 246,\ 257,\ 347,\ 356,
\]
where \(123\) denotes \(\{1,2,3\}\), and similarly for the other triples.
Every symbol appears in three triples, and every pair appears in
exactly one. This is the Fano \(2\)-design.\footnote{The Fano family is also used in~\cite{barakpelleg2026fano} to study the time needed to cover all seven points by repeatedly drawing whole triples. Here one hidden triple is selected once, and we measure the information obtained by sampling its individual members.}
Its singleton and pair marginals agree with those of the uniform
35-input prior, so \(d_1=d_2=0\), while \(d_3=\log5\).
Consequently
\begin{equation}\label{eq:coupon-fano-gap}
\boxed{
\mathsf C_{7,3}(\gamma)-I_{\mathrm{Fano}}(\gamma)
=\log5\,(1-e^{-\gamma/3})^3.
}
\end{equation}
The deficit is \((\log5)\gamma^3/27+O(\gamma^4)\):
the seven-input prior matches both the linear and quadratic terms.
It also achieves full-input capacity with exactly one or two reads.
Corollary~\ref{cor:coupon-support} shows that seven is the smallest
possible support with this property, even if unequal input
probabilities are allowed. The exact gap can also be read directly from the posterior. With one distinct observed symbol, the full and Fano priors leave 15 and three possible inputs, respectively. Both information gains are \(\log(7/3)\). With two distinct symbols, they leave five and one possible inputs, respectively; both information gains are \(\log7\). When all three members are observed, the gains are \(\log35\) and \(\log7\), a difference of \(\log5\). This event has probability \((1-e^{-\gamma/3})^3\).
It is the event of seeing three distinct members, not merely of
taking at least three reads. The different posterior uncertainties in the first two cases are consistent with equal information gains because the two priors start with different entropies.

\noindent Figure~\ref{fig:coupon-designs} evaluates these exact formulas.
At \(\gamma=1\), the Fano prior gives approximately \(0.7504\)
nats, compared with capacity \(0.7870\) nats, retaining \(95.3\%\)
of the full-input information with seven possible inputs.
The restriction changes the messages that can be represented;
it does not encode all 35 input labels into seven labels.
As the budget grows, the two information curves tend to
\(\log7\) and \(\log35\).
The Fano prior thus also illustrates the qualification in
Theorem~\ref{thm:cap-selection}: matching the minimum-redundancy
coefficient need not make a prior optimal at any positive budget.

\begin{figure}[htbp]
\centering
\includegraphics[width=\linewidth]{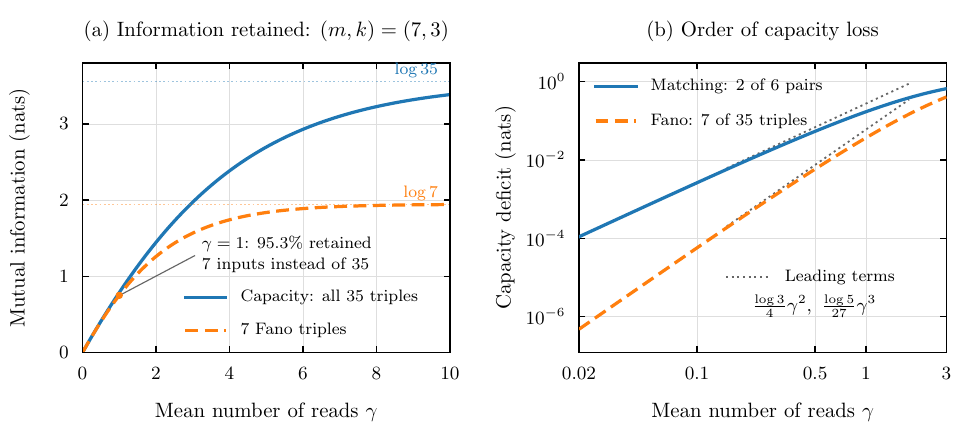}
\caption{Exact information curves for balanced input families in the
coupon-collector channel. \\
(a) For \(m=7,k=3\), the uniform prior on seven Fano triples is
compared with capacity, achieved by the uniform prior on all
35 triples. Their limiting values are \(\log7\) and \(\log35\). 
\\
(b) Information deficits for the two-pair matching
\((m,k)=(4,2)\) and the seven-triple Fano family \((m,k)=(7,3)\).
Their respective leading orders are quadratic and cubic in
\(\gamma\); dotted curves show the leading terms
\((\log3)\gamma^2/4\) and \((\log5)\gamma^3/27\).
Each deficit is relative to the full-input capacity
of its own channel. The horizontal axis is the mean total
number of reads, including repetitions.}
\label{fig:coupon-designs}
\end{figure}

The same calculation applies to any family of \(k\)-subsets in which
each \((k-1)\)-subset occurs exactly once, whenever such a family
exists. Counting these incidences gives
\(b=\binom{m}{k-1}/k\) members.
The uniform prior is a weighted \((k-1)\)-design and satisfies
\(D(P\|U)=\log(\binom mk/b)=\log(m-k+1)\).
Thus
\[
\mathsf C_{m,k}(\gamma)-I_P(\gamma)
=\log(m-k+1)(1-e^{-\gamma/k})^k.
\]
It matches the first \(k-1\) derivatives of the full-input
capacity at zero using fewer possible inputs.
The matching and Fano constructions are the \(k=2\) and \(k=3\)
instances of this statement.

\FloatBarrier

\section{Continuous-time marked-Poisson DMC model}\label{sec:ct-model}
We now replace the single hidden input by a process that evolves in physical time, while keeping the channel \(M\) fixed. The estimation target is the channel-output distribution associated with the current state. We compare estimation from observations already available with retrospective estimation from the full observation record.
Fix \(T>0\), and let
\begin{equation}
X^T=\{X_t:0\le t\le T\},\text{ } X_t\in\X,
\end{equation}
be a finite-state process with right-continuous paths having left limits, viewed as a random element of the standard Borel path space \(D([0,T],\X)\).  The true path law \(P\) and any assumed path law \(Q\) remain fixed as \(\gamma\) varies, and there is no feedback from the observations into the input.
For \(\gamma\ge0\), generate a homogeneous Poisson clock of rate \(\gamma\), independently of the entire input path.  Conditional on the path and clock times, generate independent marks, with the mark at time \(t\) having law \(M(\cdot| X_t)\).  The complete observation \(Y_\gamma^T\) retains both timestamps and marks; its expected total count is \(\gamma T\).
Write \(W_t=W_{X_t}=M(\cdot|X_t)\) for the channel row selected by the state at time \(t\), and let \(Y_z(t)\) count the observations with mark \(z\) up to time \(t\). Conditional on \(X^T\), these counting processes are independent inhomogeneous Poisson processes with intensity densities
\begin{equation}\label{eq:ct-intensity}
\boxed{
\lambda_z(t)=\gamma W_t(z)=\gamma M(z| X_t).
}
\end{equation}
Because
\begin{equation}
\sum_z\lambda_z(t)=\gamma,
\end{equation}
the superposed event clock has the same law under every input path.  Thus this is continuous-time random sampling of the input process through the DMC.  When the input varies in time, timestamps generally cannot be discarded in favor of only aggregate symbol counts.

Let \(\F_t^{Y,\gamma}\) be the observation filtration and \(\F_{t-}^{Y,\gamma}\) its strict past.  Under a possibly mismatched path law \(Q\), define
\begin{align}
\widehat W_{Q,t}^{\mathrm c,\gamma}
&\triangleq \E_Q[W_{t-}|\F_{t-}^{Y,\gamma}],\label{eq:ct-cpred}\\
\widehat W_{Q,t}^{\mathrm{nc},\gamma}
&\triangleq \E_Q[W_t|\F_T^{Y,\gamma}].\label{eq:ct-ncpred}
\end{align}
For the first conditional mean we use a predictable version, with \(W_{0-}=W_0\).  It is the causal predictor; the second is a smoother that exploits all observations, inclusive of those occurring after time \(t\).  Since \(W_{t-}=W_t\) for Lebesgue-almost every \(t\) on each path, using the left limit does not change the integrated losses.  The independent clock almost surely hits no input jump.  At an arrival, the strict-past convention ensures that the mark being scored is not already used by its own predictor.
Define the integrated causal and noncausal KL losses
\begin{align}
\cC_{P,Q}(\gamma)
&\triangleq
\E_P\int_0^T
\KL{W_t}{\widehat W_{Q,t}^{\mathrm c,\gamma}}\,dt,
\label{eq:ct-C}\\
\cS_{P,Q}(\gamma)
&\triangleq
\E_P\int_0^T
\KL{W_t}{\widehat W_{Q,t}^{\mathrm{nc},\gamma}}\,dt.
\label{eq:ct-S}
\end{align}
These are expected cumulative losses over \([0,T]\), not per-unit-time averages. The predictors are constructed under the path law \(Q\), while their performance is evaluated under the true law \(P\). When \(Q=P\), they are the minimum causal and noncausal Bayes risks under KL loss. For mismatched laws, the support and finiteness conditions needed to evaluate these risks under \(P\) are specified in Section~\ref{sec:ct-mismatch}.

\section{Causal and noncausal information--estimation identities}\label{sec:ct-identities}

\begin{theorem}[Matched continuous-time identities]\label{thm:ct-matched}
For the marked-Poisson DMC observation model of Section~\ref{sec:ct-model}, the mutual information is finite and continuously differentiable on \([0,\infty)\), with a right derivative at zero.  For \(\gamma\ge0\),
\begin{equation}\label{eq:ct-duncan}
\boxed{
I_P(X^T;Y_\gamma^T)
=
\gamma\cC_{P,P}(\gamma),
}
\end{equation}
and
\begin{equation}\label{eq:ct-derivative}
\boxed{
\frac{d}{d\gamma}I_P(X^T;Y_\gamma^T)
=
\cS_{P,P}(\gamma).
}
\end{equation}
Consequently,
\begin{equation}\label{eq:ct-average}
\boxed{
\cC_{P,P}(\gamma)
=
\frac1\gamma\int_0^\gamma\cS_{P,P}(\alpha)\,d\alpha,
\text{ } \gamma>0.
}
\end{equation}
\end{theorem}

\begin{proof}
The underlying argument uses two descriptions of the same record: adding a uniformly timed observation gives the smoothing derivative, while factoring the record likelihood in chronological order and applying Poisson compensation gives the causal identity.
Take \(Q=P\) in Theorem~\ref{thm:ct-mismatch}. Its finite path-relative-entropy hypothesis then holds with value zero, and the output relative entropy vanishes. Equations~\eqref{eq:ct-duncan}--\eqref{eq:ct-average} follow, including the stated endpoint convention. The general proof in Appendix~\ref{app:ct-proof} permits zero channel entries and does not require finite Shannon entropy of the entire path.
\end{proof}

\noindent The first identity expresses the information in the complete record as the sampling rate times the causal prediction loss accumulated over \([0,T]\). The second identifies the marginal information gained by increasing the sampling rate with the integrated smoothing loss. Combining them shows that the filtering loss at rate \(\gamma\) equals the average smoothing loss over rates \(\alpha\in[0,\gamma]\), with the horizon \(T\), path law \(P\), and channel \(M\) held fixed. This average is over sampling rates, not physical time; each loss already includes integration over \([0,T]\).

\section{Continuous-time with mismatch}\label{sec:ct-mismatch}

We extend the matched filtering and smoothing identities of Section~\ref{sec:ct-identities} in the same order as the static results in Section~\ref{sec:static-mismatch}: first change the law assumed for the hidden input, then allow an incorrect channel. Throughout, the true path and observation experiment remain fixed as the estimator's model changes.

\subsection{Path-law mismatch with a fixed channel}\label{sec:ct-prior-mismatch}

Let \(Q\) be the path law used by the estimator. It may misdescribe the initial state distribution, the temporal evolution of the process, or both; the channel and independent sampling clock remain unchanged. For finite-valued identities, impose the sufficient condition
\[
D_0\triangleq D(P_{X^T}\|Q_{X^T})<\infty.
\]
This implies \(P\ll Q\), but is stronger than absolute continuity alone on path space.  It is automatic when \(Q=P\), permits zero channel entries, and does not assume finite Shannon entropy of the path. As in Theorem~\ref{thm:mismatch}, we combine the true mutual information with the divergence between the true and nominal observation laws. Here those laws describe the complete time-mark record, and we define
\begin{equation}
\cJ_{P,Q}(\gamma)
\triangleq
I_P(X^T;Y_\gamma^T)
+
\KL{P_{Y_\gamma^T}}{Q_{Y_\gamma^T}}.
\end{equation}

\begin{theorem}[Mismatched causal and noncausal bridge]\label{thm:ct-mismatch}
The function \(\cJ_{P,Q}\) is finite and continuously differentiable on \([0,\infty)\), with a right derivative at zero.  For \(\gamma\ge0\),
\begin{equation}\label{eq:ct-mismatch-master}
\boxed{
\cJ_{P,Q}(\gamma)
=
\gamma\cC_{P,Q}(\gamma)
=
\int_0^\gamma\cS_{P,Q}(\alpha)\,d\alpha.
}
\end{equation}
Equivalently,
\begin{equation}\label{eq:ct-mismatch-derivative}
\boxed{
\frac{d}{d\gamma}\cJ_{P,Q}(\gamma)
=
\cS_{P,Q}(\gamma),
}
\end{equation}
and
\begin{equation}\label{eq:ct-mismatch-average}
\boxed{
\cC_{P,Q}(\gamma)
=
\frac1\gamma\int_0^\gamma\cS_{P,Q}(\alpha)\,d\alpha,\text{ }\gamma>0.
}
\end{equation}
\end{theorem}

\begin{proof}
Set \(\widetilde M=M\) in Theorem~\ref{cm:thm:ct}, proved in Appendix~\ref{app:ct-proof}. Then \(K_T=0\), so its hypotheses reduce to \(D_0<\infty\), and its predictors and risks reduce to \eqref{eq:ct-cpred}--\eqref{eq:ct-S}. With \(\mathcal J_T=\cJ_{P,Q}\), the identities and regularity assertions follow, including the right derivative at zero.
\end{proof}

Subtracting the matched formulas gives two complementary regret identities:
\begin{equation}\label{eq:ct-regret-s}
\boxed{
\KL{P_{Y_\gamma^T}}{Q_{Y_\gamma^T}}
=
\int_0^\gamma
\left[
\cS_{P,Q}(\alpha)-\cS_{P,P}(\alpha)
\right]d\alpha,
}
\end{equation}
and
\begin{equation}\label{eq:ct-regret-c}
\boxed{
\KL{P_{Y_\gamma^T}}{Q_{Y_\gamma^T}}
=
\gamma\left[
\cC_{P,Q}(\gamma)-\cC_{P,P}(\gamma)
\right].
}
\end{equation}
Thus the excess causal loss is the running average of the excess smoothing loss:
\begin{equation}\label{eq:ct-regret-average}
\boxed{
\cC_{P,Q}(\gamma)-\cC_{P,P}(\gamma)
=
\frac1\gamma\int_0^\gamma
\left[
\cS_{P,Q}(\alpha)-\cS_{P,P}(\alpha)
\right]d\alpha,\text{ }\gamma>0.
}
\end{equation}

These formulas measure the additional prediction cost of using \(Q\). In each difference, the nominal and matched predictors use the same observation history and are evaluated under \(P\); the conditional KL decomposition \eqref{eq:conditional-kl-decomposition} makes this excess cost nonnegative. Equation~\eqref{eq:ct-regret-s} accumulates excess smoothing loss over sampling rates, while \eqref{eq:ct-regret-c} expresses the same output divergence as the current excess causal loss multiplied by \(\gamma\). Their combination gives \eqref{eq:ct-regret-average}. The horizon \(T\) is fixed throughout: both losses already integrate over physical time.

\subsection{Simultaneous path-law and channel mismatch}\label{sec:ct-joint-mismatch}

We now allow the estimator to use a nominal channel \(\widetilde M\) under the nominal path law \(Q\). As in Section~\ref{sec:joint-mismatch}, the channel changes the nominal likelihood and the rows averaged by its predictors. The true marks still use \(W_{X_t}\), the nominal marks use \(\widetilde W_{X_t}\), and both retain the same independent rate-\(\gamma\) clock.

\noindent The true and nominal conditional intensity densities are
\begin{equation}
\lambda_z(t)=\gamma W_{X_t}(z),
\text{ }
\widetilde\lambda_z(t)=\gamma\widetilde W_{X_t}(z).
\end{equation}
Both total intensities equal \(\gamma\); thus the unmarked clock has the same law under every path and under both models.  We impose the sufficient conditions
\begin{equation}\label{cm:eq:ct-assumptions}
\begin{aligned}
P&\ll Q,\text{ } D_0\triangleq D(P\|Q)<\infty,\\
K_T&\triangleq\int_0^T\E_P\KL{W_{X_t}}{\widetilde W_{X_t}}\,dt<\infty.
\end{aligned}
\end{equation}
Here \(D_0\) compares the laws of the entire hidden path. The quantity \(K_T\) is the cumulative channel discrepancy along the true path: even if that path were known, using the nominal row at each time would incur this integrated KL loss. It is a cumulative quantity over \([0,T]\), not a time average.

\noindent The first absolute-continuity condition follows already from \(D_0<\infty\), but is displayed to emphasize the common path space.  The condition on \(K_T\) permits zero channel entries while ensuring the required same-label channel support for \(P\)-almost every path and almost every time.  These conditions imply absolute continuity of the true joint path--observation law with respect to its nominal counterpart, as Appendix~\ref{app:ct-proof} shows.  In particular, nominal conditional estimates can be evaluated on true observation records without ambiguity on sets of positive true probability.  No finite Shannon entropy assumption on the entire path is needed.

\noindent Write \(\F_t^{Y,\gamma}=\sigma(Y_\gamma^t)\) for the observation filtration and \(\F_{t-}^{Y,\gamma}\) for its strict past.  Define the nominal causal and smoothing predictors by
\begin{align}
\widehat{\widetilde W}_{Q,t}^{\mathrm c,\gamma}
&\triangleq
\E_{Q,\widetilde M}[\widetilde W_{X_{t-}}|\F_{t-}^{Y,\gamma}],\\
\widehat{\widetilde W}_{Q,t}^{\mathrm{nc},\gamma}
&\triangleq
\E_{Q,\widetilde M}[\widetilde W_{X_t}|\F_T^{Y,\gamma}].
\end{align}
Use a predictable version of the first predictor, with \(X_{0-}=X_0\), and a jointly measurable version of the second obtained from a regular conditional path law. The left-limit and strict-past conventions are those of Section~\ref{sec:ct-model}, following \eqref{eq:ct-cpred}--\eqref{eq:ct-ncpred}: they leave the integrated losses unchanged and ensure that an arriving mark is not used by its own predictor. The fully mismatched risks are
\begin{align}
\Cmis(\gamma)
&\triangleq
\int_0^T\E_{P,M}
\KL{W_{X_t}}{\widehat{\widetilde W}_{Q,t}^{\mathrm c,\gamma}}\,dt,
\label{cm:eq:Cmis}\\
\Smis(\gamma)
&\triangleq
\int_0^T\E_{P,M}
\KL{W_{X_t}}{\widehat{\widetilde W}_{Q,t}^{\mathrm{nc},\gamma}}\,dt.
\label{cm:eq:Smis}
\end{align}

\noindent The expectation under \((Q,\widetilde M)\) constructs each predictor from the nominal posterior and nominal rows. The expectation under \((P,M)\) evaluates its loss against the true row, so the estimation target remains \(W_{X_t}\).

\noindent With \(\widetilde M=M\), these predictors and risks reduce to \eqref{eq:ct-cpred}--\eqref{eq:ct-S}. The causal and noncausal bridge below retains the form of \eqref{eq:ct-mismatch-master}; only the nominal observation law and predictors change.

\begin{theorem}[Continuous-time fully mismatched identities]\label{cm:thm:ct}
Under \eqref{cm:eq:ct-assumptions}, let
\begin{equation}
\mathcal J_T(\gamma)
\triangleq
I_{P,M}(X^T;Y_\gamma^T)
+\KL{P_{Y_\gamma^T}^M}{Q_{Y_\gamma^T}^{\widetilde M}}.
\end{equation}
This function is finite and continuously differentiable on \([0,\infty)\), with a right derivative at zero.  The risks are finite and continuous, and for \(\gamma\ge0\),
\begin{align}
\boxed{
\mathcal J_T(\gamma)
=\gamma\Cmis(\gamma)
=\int_0^\gamma\Smis(\alpha)\,d\alpha
},\label{cm:eq:ct-causal}\\[1mm]
\boxed{\mathcal J_T'(\gamma)=\Smis(\gamma)}.
\label{cm:eq:ct-smooth}
\end{align}
In particular,
\begin{equation}\label{cm:eq:ct-average}
\boxed{
\Cmis(\gamma)
=\frac1\gamma\int_0^\gamma\Smis(\alpha)\,d\alpha,
\text{ }\gamma>0.
}
\end{equation}
At zero intensity, both risks equal
\begin{equation}\label{cm:eq:ct-zero}
\Cmis(0)=\Smis(0)
=\int_0^T\E_P\KL{W_{X_t}}{\E_Q\widetilde W_{X_t}}\,dt<\infty.
\end{equation}
\end{theorem}

\begin{proof}
Appendix~\ref{app:ct-proof} gives the full argument. The joint-law identity \eqref{cm:eq:ct-joint-fixed} gives the growth estimate \eqref{eq:ct-proof-growth}; Poisson differentiation of the increment \eqref{eq:ct-proof-increment} then gives the smoothing derivative. Chronological likelihood factorization and the compensation formula \eqref{eq:ct-proof-compensation} give the causal identity. The endpoint argument following \eqref{eq:ct-proof-zero} establishes continuity at zero, including zero channel entries.
\end{proof}
Define the true predictors using the same information sets:
\[
\widehat W_{P,t}^{\mathrm c,\gamma}
=\E_{P,M}[W_{X_{t-}}|\F_{t-}^{Y,\gamma}],
\text{ }
\widehat W_{P,t}^{\mathrm{nc},\gamma}
=\E_{P,M}[W_{X_t}|\F_T^{Y,\gamma}].
\]
Let \(\cC_{P,P}(\gamma)\) and \(\cS_{P,P}(\gamma)\) denote the matched risks of \eqref{eq:ct-C}--\eqref{eq:ct-S}, obtained here by setting \((Q,\widetilde M)=(P,M)\). Subtract the matched identities \eqref{eq:ct-duncan}--\eqref{eq:ct-derivative} from \eqref{cm:eq:ct-causal}--\eqref{cm:eq:ct-smooth}, and apply the conditional KL decomposition \eqref{eq:conditional-kl-decomposition} at each time. This gives
\begin{align}
\KL{P_{Y_\gamma^T}^M}{Q_{Y_\gamma^T}^{\widetilde M}}
&=\gamma\big[\Cmis(\gamma)-\Cmat(\gamma)\big],
\label{cm:eq:ct-output-causal}\\
\frac{d}{d\gamma}
\KL{P_{Y_\gamma^T}^M}{Q_{Y_\gamma^T}^{\widetilde M}}
&=\Smis(\gamma)-\Smat(\gamma)
=\int_0^T\E_{P,M}
\KL{\widehat W_{P,t}^{\mathrm{nc},\gamma}}
{\widehat{\widetilde W}_{Q,t}^{\mathrm{nc},\gamma}}\,dt.
\label{cm:eq:ct-output-smooth}
\end{align}
All quantities are finite, so these subtractions are legitimate; the derivative at zero is again from the right. Subtracting the averaging formulas \eqref{cm:eq:ct-average} and \eqref{eq:ct-average} also gives the same relation between excess causal and smoothing losses as \eqref{eq:ct-regret-average}.

\noindent As a consequence of the same joint-law calculation, the prediction discrepancy can be separated from posterior path-law discrepancy. The path-space counterpart of \eqref{cm:eq:conservation} is
\begin{equation}\label{cm:eq:ct-conservation}
\boxed{
\begin{aligned}
D_0+\gamma K_T
&=\KL{P_{Y_\gamma^T}^M}{Q_{Y_\gamma^T}^{\widetilde M}}\\
&\text{ }+\E_{P_{Y_\gamma^T}^M}
\KL{P^M_{X^T| Y_\gamma^T}}
{Q^{\widetilde M}_{X^T| Y_\gamma^T}}.
\end{aligned}
}
\end{equation}
\noindent Apply the KL chain rule in the opposite order to \eqref{cm:eq:ct-joint-divergence}. This gives \eqref{cm:eq:ct-conservation}, with finite expected posterior divergence.  At zero intensity that posterior divergence is \(D_0\).  Thus mismatch in the process law contributes an initial path-space discrepancy, while channel mismatch contributes \(\gamma K_T\) along the path.

\noindent The output term concerns the law of the observed time-mark record; the posterior term compares two conditional distributions over entire hidden trajectories given that same record. As in the label-swapping Example~\ref{cm:ex:label-swap}, agreement about observable predictions need not imply agreement about the hidden labels. Equation~\eqref{cm:eq:ct-conservation} expresses this distinction on path space, under the same-label support conditions \eqref{cm:eq:ct-assumptions}.

\section{Overview of results}\label{sec:dictionary}

The tables collect the main identities; the cited results give their support and finiteness conditions. All derivatives are with respect to \(\gamma\), with right derivatives at zero. In the static experiment, \(I_P(\gamma)=I_P(X;Y_\gamma)\). Prior-only mismatch follows by setting \(\widetilde M=M\).

\begin{center}
\renewcommand{\arraystretch}{1.3}
\small
\begin{tabular}{>{\raggedright\arraybackslash}p{0.34\textwidth} >{\raggedright\arraybackslash}p{0.56\textwidth}}
\toprule
\multicolumn{2}{l}{\textbf{Static prediction}}\\
\midrule
\textbf{Information quantity} & \textbf{Prediction interpretation}\\
\midrule
\(\displaystyle I_P'(\gamma)\)
& Matched predictive KL risk \(\cR_{P,P}(\gamma)\), \eqref{eq:first}.\\
\(\displaystyle -I_P''(\gamma)\)
& Dependence between two fresh outputs, \(I_P(Z_1';Z_2'|Y_\gamma)\), \eqref{eq:second}; equivalently, predictive-update KL, \eqref{eq:predictive-update}.\\
\(\displaystyle I_P^{(k)}(\gamma),\ k\ge2\)
& Signed conditional \(k\)-way co-information of fresh outputs, \eqref{eq:higher}.\\
\(\displaystyle I_{P,M}'(\gamma)+\Dout'(\gamma)\)
& Total mismatched risk \(\Rmis(\gamma)\), \eqref{cm:eq:main-identity}. The excess risk is \(\Dout'=\Rmis-\Rmat\), equal to the expected KL divergence between the true and nominal predictors, \eqref{cm:eq:Dout-predictive}.\\
\(\displaystyle D(P\|Q)+\gamma\kappa\)
& \(\Dout+\Dpost\): output discrepancy plus label-posterior discrepancy, under same-label support, \eqref{cm:eq:conservation}.\\
\bottomrule
\end{tabular}
\end{center}

For static capacity, each candidate prior uses its own matched predictor. Write \(C(M)\) for one-use DMC capacity and \(r\) for the number of distinct channel rows.

\begin{center}
\renewcommand{\arraystretch}{1.3}
\small
\begin{tabular}{>{\raggedright\arraybackslash}p{0.34\textwidth} >{\raggedright\arraybackslash}p{0.56\textwidth}}
\toprule
\multicolumn{2}{l}{\textbf{Static capacity}}\\
\midrule
\textbf{Information quantity} & \textbf{Prediction interpretation}\\
\midrule
\(\displaystyle \mathsf C_M(\gamma)\)
& \(\max_P\int_0^\gamma\cR_{P,P}(s)\,ds\): largest accumulated matched risk, \eqref{eq:cap-envelope}. Its large-budget limit is \(\log r\), \eqref{eq:cap-infinity}.\\
\(\displaystyle \mathsf C_M'(\gamma)\)
& With a unique optimizer at \(\gamma>0\), its matched risk at that budget, \eqref{eq:cap-derivative}. At zero, \(\mathsf C_M'(0+)=C(M)\), without uniqueness, \eqref{eq:cap-zero-slope}.\\
\(\displaystyle
\lim_{\gamma\downarrow0}
\frac{2[\gamma C(M)-\mathsf C_M(\gamma)]}{\gamma^2}
\)
& Minimum two-output redundancy \(B_*\) among one-use capacity-achieving priors; a convex secondary optimization, \eqref{eq:cap-second-order}--\eqref{eq:cap-convex-selection}.\\
\(\displaystyle \mathsf C_{m,k}(\gamma)-I_P(\gamma)\)
& Coupon-collector deficit \(\sum_{r=1}^k\Pr(R=r)D(\mu_{P,r}\|u_r)\), with \(R\) the number of distinct observed symbols, \eqref{eq:coupon-deficit}. Balance through order \(t\) preserves the first \(t\) capacity terms, Corollary~\ref{cor:coupon-order}.\\
\bottomrule
\end{tabular}
\end{center}

For processes, \(\gamma\) is the sampling rate and \(T\) is fixed. Both filtering and smoothing losses include integration over physical time. The quantity \(\mathcal J_T\) is mutual information plus output relative entropy under path-law and channel mismatch; its common-channel specialization is \(\cJ_{P,Q}\).

\begin{center}
\renewcommand{\arraystretch}{1.3}
\small
\begin{tabular}{>{\raggedright\arraybackslash}p{0.34\textwidth} >{\raggedright\arraybackslash}p{0.56\textwidth}}
\toprule
\multicolumn{2}{l}{\textbf{Continuous-time prediction}}\\
\midrule
\textbf{Information quantity} & \textbf{Prediction interpretation}\\
\midrule
\(\displaystyle I_P(X^T;Y_\gamma^T)\)\newline
and its rate derivative
& \(\gamma\cC_{P,P}(\gamma)\) and \(\cS_{P,P}(\gamma)\), respectively: cumulative causal loss times the rate, and cumulative smoothing loss, \eqref{eq:ct-duncan}--\eqref{eq:ct-derivative}.\\
\(\displaystyle \mathcal J_T(\gamma)\)
& \(\gamma\Cmis(\gamma)=\int_0^\gamma\Smis(\alpha)\,d\alpha\): causal/noncausal averaging under mismatch, \eqref{cm:eq:ct-causal}. For a common channel, this is \eqref{eq:ct-mismatch-master}.\\
\bottomrule
\end{tabular}
\end{center}

\section{Examples}\label{sec:examples}
The preceding sections developed the static and continuous-time
identities and applied the capacity analysis to balanced input
designs (Section~\ref{sec:coupon-design}). We now work out further
examples of prediction, mismatch, and Poisson sampling. The BSC illustrates gradual learning in the static model; the erasure channel makes the static matched and mismatched identities explicit and turns the dynamic KL loss into conditional entropy, and a binary symmetric Markov input yields closed-form filtering, smoothing, and information-rate expressions. A numerical illustration uses finite image catalogues to connect information with reconstruction and observation prediction under channel mismatch.

\subsection{Binary input through a BSC}\label{sec:bsc-example}

For a binary symmetric channel, symmetry makes the uniform prior capacity achieving for every \(\gamma\).  Hence in this section \(I(\gamma)=\mathsf C_M(\gamma)\).  Consider \(X\sim\Ber(1/2)\) and a binary symmetric channel with crossover \(\varepsilon\in(0,1/2)\):
\begin{equation}
W_0=(1-\varepsilon,\varepsilon),
\text{ }
W_1=(\varepsilon,1-\varepsilon),
\end{equation}
where coordinates correspond to outputs \(0,1\).  Let \(C_1\) and \(C_0\) be the Poissonized counts of ones and zeros, and define
\begin{equation}
D=C_1-C_0,
\text{ }
a=\log\frac{1-\varepsilon}{\varepsilon}.
\end{equation}
Conditional on \(X=1\),
\begin{equation}
C_1\sim\Pois(\gamma(1-\varepsilon)),
\text{ }
C_0\sim\Pois(\gamma\varepsilon),
\end{equation}
independently, while the two Poisson means are interchanged under \(X=0\). Hence \(D\) has a Skellam distribution conditional on \(X\). Each observed one adds \(a\) to the log likelihood ratio favoring \(X=1\), and each zero subtracts \(a\); thus \(D\) measures the net evidence for the hidden bit. With the uniform prior, the posterior log odds are \(aD\), and Bayes' rule gives
\begin{equation}\label{eq:bsc-posterior}
\pi(D)
\triangleq
P(X=1| Y_\gamma)
=
\frac{1}{1+e^{-aD}}.
\end{equation}
The predictive probability of observing a one next is
\begin{equation}\label{eq:bsc-q}
q(D)
=
P(Z_{\rm new}=1| Y_\gamma)
=
\varepsilon+(1-2\varepsilon)\pi(D).
\end{equation}
All expectations below average over the true observation law with the uniform input prior. Equivalently, \(D\) is averaged over the equally weighted mixture of its two conditional Skellam laws. Thus
\begin{equation}\label{eq:bsc-I}
I(\gamma)
=
\log2-\E[h_2(\pi(D))],
\end{equation}
where \(h_2(u)=-u\log u-(1-u)\log(1-u)\), and
\begin{equation}\label{eq:bsc-risk}
\begin{aligned}
I'(\gamma)
&=\E\Big[
\pi(D)D(\Ber(1-\varepsilon)\|\Ber(q(D)))\\
&\hspace{3.5em}+(1-\pi(D))D(\Ber(\varepsilon)\|\Ber(q(D)))
\Big].
\end{aligned}
\end{equation}

For the curvature, conditioned on the current posterior \(\pi\), the joint law \(r_\pi\) of two fresh output bits is
\begin{align}
r_\pi(1,1)&=\pi(1-\varepsilon)^2+(1-\pi)\varepsilon^2,\\
r_\pi(1,0)&=r_\pi(0,1)=\varepsilon(1-\varepsilon),\\
r_\pi(0,0)&=\pi\varepsilon^2+(1-\pi)(1-\varepsilon)^2.
\end{align}
With \(q=q(D)\), Theorem~\ref{thm:second} becomes
\begin{equation}\label{eq:bsc-curvature}
-I''(\gamma)
=
\E\left[
D\big(r_{\pi(D)}\|\Ber(q(D))^{\otimes2}\big)
\right].
\end{equation}

Table~\ref{tab:bsc} gives representative values for \(\varepsilon=0.15\), computed by summing the corresponding Skellam probabilities.  The identities show up numerically as the increasing, concave behavior of \(I(\gamma)\), with \(I'(\gamma)\downarrow0\) and \(-I''(\gamma)\to0\) as the hidden bit becomes essentially known.

\begin{table}[t]
\centering
\caption{Capacity of the Poissonized BSC with crossover \(\varepsilon=0.15\), together with its first two derivatives.  Values are in nats.}
\label{tab:bsc}
\begin{tabular}{rrrr}
\toprule
\(\gamma\) & \(I(\gamma)\) & \(I'(\gamma)=\cR_{P,P}(\gamma)\) & \(-I''(\gamma)\)\\
\midrule
0    & 0.00000 & 0.27044 & 0.12539\\
0.5  & 0.12097 & 0.21604 & 0.09418\\
1    & 0.21820 & 0.17469 & 0.07248\\
2    & 0.36188 & 0.11718 & 0.04524\\
4    & 0.52683 & 0.05605 & 0.01991\\
8    & 0.64834 & 0.01438 & 0.00473\\
\bottomrule
\end{tabular}
\end{table}

\subsection{The q-ary erasure channel}\label{sec:erasure}

The erasure channel makes the preceding identities fully explicit and, in the continuous-time model, reduces the predictive KL loss to ordinary conditional entropy.  Let $|\X|=q$ and let the output alphabet be $\Z=\X\cup\{\mathtt e\}$.  For erasure probability $\varepsilon\in[0,1)$, write
\begin{equation}
 a\triangleq 1-\varepsilon,
\end{equation}
and define
\begin{equation}\label{eq:erasure-row}
M(z| x)=
\begin{cases}
a, & z=x,\\
\varepsilon, & z=\mathtt e,\\
0, & \text{otherwise}.
\end{cases}
\end{equation}

\subsubsection{Static model: effective erasure probability}

By Poisson thinning, the number of non-erased observations is Poisson with mean $a\gamma$. Every non-erased observation equals $X$. Consequently the entire repeated-observation experiment is statistically equivalent to a single $q$-ary erasure channel whose effective erasure probability is
\begin{equation}\label{eq:effective-erasure}
\boxed{
\delta_\gamma=e^{-a\gamma}.
}
\end{equation}
If all observations are erased, an event of probability $e^{-a\gamma}$, the posterior remains equal to the prior; otherwise the first non-erased symbol reveals $X$ exactly.

For an arbitrary input prior $P$,
\begin{equation}\label{eq:erasure-mutual-info}
\boxed{
I_P(X;Y_\gamma)=\left(1-e^{-a\gamma}\right)H_P(X).
}
\end{equation}
Thus a uniform prior achieves capacity, and
\begin{equation}\label{eq:erasure-capacity}
\boxed{
\mathsf C_M(\gamma)=\left(1-e^{-a\gamma}\right)\log q.
}
\end{equation}
In particular,
\begin{equation}
\mathsf C_M'(0+)=a\log q=C(M),
\text{ }
\lim_{\gamma\to\infty}\mathsf C_M(\gamma)=\log q,
\end{equation}
which realizes both endpoint statements of Section~\ref{sec:capacity} with equality.  More generally,
\begin{align}
\mathsf C_M'(\gamma)&=a e^{-a\gamma}\log q,\\
\mathsf C_M''(\gamma)&=-a^2e^{-a\gamma}\log q.
\end{align}
\noindent The information--estimation identity is equally explicit.  Prior to the first non-erased observation the posterior remains $P$, while afterward $X$ is known.  If no informative observation has arrived, the predictive row is
\begin{equation}
\widehat W_P
=\varepsilon\delta_{\mathtt e}
+a\sum_x P(x)\delta_x,
\end{equation}
and for $X=x$,
\begin{equation}
D(W_x\|\widehat W_P)=a\log\frac1{P(x)}.
\end{equation}
Hence the matched predictive risk is
\begin{equation}\label{eq:erasure-risk}
\boxed{
\cR_{P,P}(\gamma)
=a e^{-a\gamma}H_P(X)
=a H_P(X| Y_\gamma).
}
\end{equation}
Equation~\eqref{eq:erasure-mutual-info} therefore gives directly
\begin{equation}
\frac{d}{d\gamma}I_P(X;Y_\gamma)=\cR_{P,P}(\gamma).
\end{equation}
All higher derivatives are explicit:
\begin{equation}\label{eq:erasure-higher}
\boxed{
\frac{d^k}{d\gamma^k}I_P(X;Y_\gamma)
=(-1)^{k-1}a^k e^{-a\gamma}H_P(X),
\text{ } k\ge1.
}
\end{equation}
In particular,
\begin{equation}\label{eq:erasure-second}
-I_P''(\gamma)=a^2e^{-a\gamma}H_P(X)
=I_P(Z_1';Z_2'| Y_\gamma).
\end{equation}
The last equality has a direct operational explanation: residual dependence between two fresh outputs remains only if $X$ has not yet been revealed, and both fresh outputs must then avoid erasure.
For the causal interpretation of Figures~\ref{fig:erasure-losses} and \ref{fig:erasure-average}, embed the static experiment in the continuous-time model by taking \(X_t=X\) on \([0,1]\).  The integrated losses are then
\begin{equation*}
\cC_{P,P}(\gamma)=\frac{1-e^{-a\gamma}}{\gamma}H_P(X),
\text{ }
\cS_{P,P}(\gamma)=a e^{-a\gamma}H_P(X),
\text{ } \gamma>0.
\end{equation*}
For \(H_P(X)>0\), writing \(\eta=a\gamma\) gives
\begin{equation*}
\frac{\cC_{P,P}(\eta/a)}{aH_P(X)}=\frac{1-e^{-\eta}}{\eta},
\text{ }
\frac{\cS_{P,P}(\eta/a)}{aH_P(X)}=e^{-\eta}.
\end{equation*}
Both normalized curves have continuous value \(1\) at \(\eta=0\).  If \(H_P(X)=0\), the unnormalized losses vanish and these ratios are not defined.  The figures visualize the distinction between endpoint and running-average losses and the exact averaging relation.

\begin{figure}[t]
\centering
\includegraphics[width=0.70\linewidth]{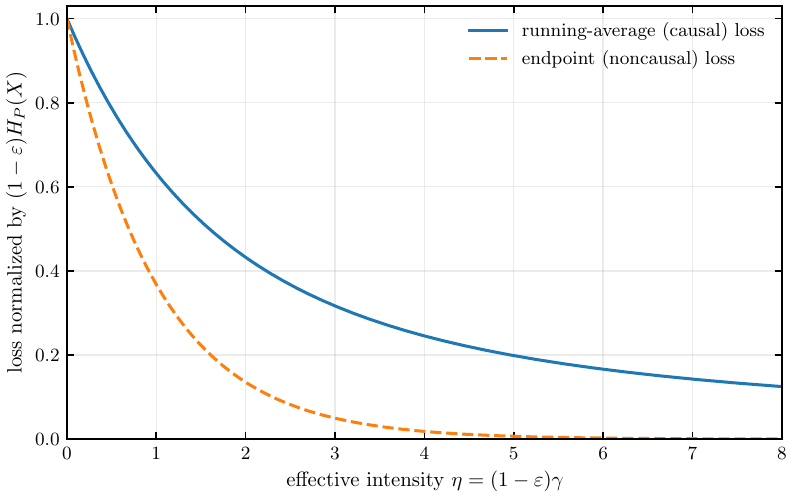}
\caption{Static erasure channel embedded as a constant input path on $[0,1]$, in the effective-intensity variable $\eta=(1-\varepsilon)\gamma$.  The endpoint predictive loss, equal to the integrated smoothing loss in this embedding, decays exponentially, while the integrated causal loss decreases more slowly.  Both curves are normalized by $(1-\varepsilon)H_P(X)$, with $H_P(X)>0$.}
\label{fig:erasure-losses}
\end{figure}

\begin{figure}[t]
\centering
\includegraphics[width=0.70\linewidth]{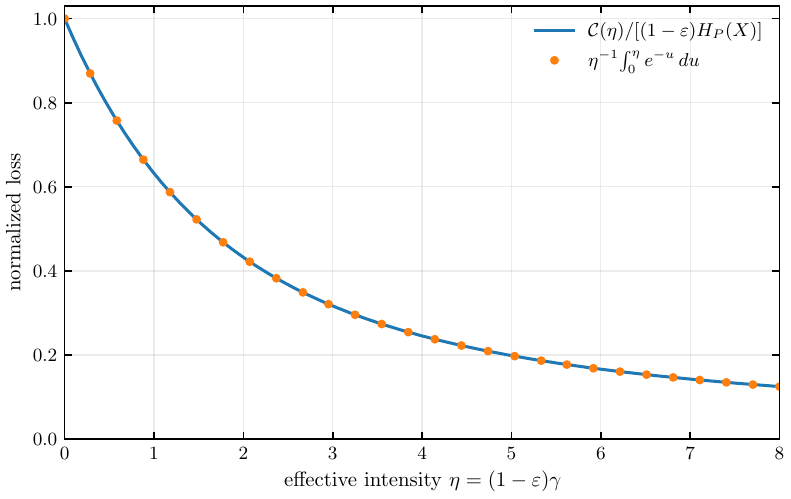}
\caption{Constant-path erasure model on $[0,1]$: the running-average identity.  With $a=1-\varepsilon$ and $H_P(X)>0$, the solid curve is $\cC_{P,P}(\eta/a)/(aH_P(X))$; the markers evaluate $\eta^{-1}\int_0^\eta e^{-u}\,du$.  In the legend, $\cC(\eta)$ abbreviates $\cC_{P,P}(\eta/a)$.  The two coincide exactly, with continuous value $1$ at zero.}
\label{fig:erasure-average}
\end{figure}

\subsubsection{Mismatch and relative entropy}

Returning to the static experiment with mean observation count $\gamma$, we first allow an incorrect prior while keeping the erasure probability fixed. Let $P$ be the true prior and $Q$ the prior used by the estimator, with $P\ll Q$.  Before the first informative observation the mismatched predictive row is
\begin{equation}
\widehat W_Q
=\varepsilon\delta_{\mathtt e}
+a\sum_xQ(x)\delta_x,
\end{equation}
and therefore
\begin{equation}
D(W_x\|\widehat W_Q)=a\log\frac1{Q(x)}.
\end{equation}
After a non-erased observation, both matched and mismatched posteriors collapse to the revealed symbol.  It follows that
\begin{equation}\label{eq:erasure-mismatch-risk}
\boxed{
\cR_{P,Q}(\gamma)
=a e^{-a\gamma}\big[H(P)+D(P\|Q)\big].
}
\end{equation}
Consequently the excess risk is
\begin{equation}\label{eq:erasure-excess}
\boxed{
\cR_{P,Q}(\gamma)-\cR_{P,P}(\gamma)
=a e^{-a\gamma}D(P\|Q).
}
\end{equation}
The output relative entropy is also elementary:
\begin{equation}\label{eq:erasure-output-kl}
\boxed{
D(P_{Y_\gamma}\|Q_{Y_\gamma})
=\left(1-e^{-a\gamma}\right)D(P\|Q).
}
\end{equation}
Thus
\begin{equation}
\frac{d}{d\gamma}D(P_{Y_\gamma}\|Q_{Y_\gamma})
=a e^{-a\gamma}D(P\|Q),
\end{equation}
exactly matching \eqref{eq:erasure-excess}.  The posterior mismatch that remains hidden after observing $Y_\gamma$ is
\begin{equation}\label{eq:erasure-posterior-kl}
\boxed{
\E_P D(P_{X| Y_\gamma}\|Q_{X| Y_\gamma})
=e^{-a\gamma}D(P\|Q),
}
\end{equation}
so the KL conservation identity simply partitions the initial prior mismatch according to whether or not an informative observation has occurred.

\subsubsection{Wrong prior and wrong erasure probability}

The erasure channel makes simultaneous prior and channel mismatch completely explicit. Let $\X=\{1,\ldots,q\}$ and let the true channel have erasure probability $\varepsilon\in(0,1)$,
\begin{equation}
W_x=\varepsilon\,\delta_{\mathtt e}+(1-\varepsilon)\,\delta_x.
\end{equation}
The nominal channel uses erasure probability $\widetilde\varepsilon\in(0,1)$,
\begin{equation}
\widetilde W_x=\widetilde\varepsilon\,\delta_{\mathtt e}+(1-\widetilde\varepsilon)\,\delta_x,
\end{equation}
and the nominal prior is $Q$, with $P\ll Q$. Write
\begin{equation}
a=1-\varepsilon,
\text{ }
\widetilde a=1-\widetilde\varepsilon,
\end{equation}
and define the Bernoulli channel discrepancy
\begin{equation}\label{cm:eq:erasure-d}
d_{\varepsilon,\widetilde\varepsilon}
\triangleq
\varepsilon\log\frac{\varepsilon}{\widetilde\varepsilon}
+a\log\frac{a}{\widetilde a}.
\end{equation}

The effective-erasure reduction above describes learning about $X$ with a fixed channel. When the erasure probability is misspecified, the full observation record also carries evidence about the channel through its erased and non-erased counts, even after $X$ is known. This discrepancy accounts for the term $\gamma d_{\varepsilon,\widetilde\varepsilon}$ in the output divergence below.

\noindent A non-erased true observation reveals $X$ exactly, and the number of such observations is Poisson with mean $a\gamma$. Before the first reveal, the true and nominal posteriors remain $P$ and $Q$; after a reveal, both true and nominal posteriors collapse to the revealed label. The only permanent mismatch is therefore the wrong erasure probability.

\begin{proposition}[Fully explicit erasure formulas]\label{cm:prop:erasure}
For the model above,
\begin{align}
I_{P,M}(X;Y_\gamma)
&=(1-e^{-a\gamma})H(P),\label{cm:eq:er-I}\\[2mm]
\Dout(\gamma)
&=
\gamma d_{\varepsilon,\widetilde\varepsilon}
+
(1-e^{-a\gamma})D(P\|Q),\label{cm:eq:er-Dout}\\[2mm]
\Rmis(\gamma)
&=
 d_{\varepsilon,\widetilde\varepsilon}
+
 a e^{-a\gamma}\big[H(P)+D(P\|Q)\big],\label{cm:eq:er-Rmis}\\[2mm]
\Rmat(\gamma)
&=a e^{-a\gamma}H(P),\label{cm:eq:er-Rmat}\\[2mm]
\Dpost(\gamma)
&=e^{-a\gamma}D(P\|Q).\label{cm:eq:er-Dpost}
\end{align}
\end{proposition}

\begin{proof}
The probability of no non-erased mark is $e^{-a\gamma}$. On this event the true and nominal label posteriors remain $P$ and $Q$; otherwise both are the same point mass. This proves the information and posterior-divergence formulas. Before a reveal, the true-row loss against the nominal predictor is
\begin{equation}
D(W_x\|\overline{\widetilde W}_Q)
=d_{\varepsilon,\widetilde\varepsilon}+a\log\frac1{Q(x)};
\end{equation}
after a reveal it is $D(W_x\|\widetilde W_x)=d_{\varepsilon,\widetilde\varepsilon}$. Averaging gives \eqref{cm:eq:er-Rmis}; the same calculation with the matched predictor gives \eqref{cm:eq:er-Rmat}. Finally, $\kappa=d_{\varepsilon,\widetilde\varepsilon}$ and Theorem~\ref{cm:thm:conservation} yield \eqref{cm:eq:er-Dout}.
\end{proof}

\noindent The additive separation here comes from the special structure of the erasure channel: erasure probabilities do not depend on the input, and an informative symbol reveals the same label under both models. Section~\ref{sec:posterior-mismatch} explains why prior and channel errors can interact in general. In this example, prior mismatch contributes only until the first informative observation, while channel mismatch contributes the permanent floor $d_{\varepsilon,\widetilde\varepsilon}$. In particular,
\begin{align}
\Dout'(\gamma)
&=
 d_{\varepsilon,\widetilde\varepsilon}
+
 a e^{-a\gamma}D(P\|Q),\\[1mm]
\lim_{\gamma\to\infty}\Rmis(\gamma)
&=
 d_{\varepsilon,\widetilde\varepsilon},
\end{align}
exactly as predicted by Proposition~\ref{cm:prop:high}.\\
The conservation law reads
\begin{equation}
D(P\|Q)+\gamma d_{\varepsilon,\widetilde\varepsilon}
=
\left[\gamma d_{\varepsilon,\widetilde\varepsilon}+(1-e^{-a\gamma})D(P\|Q)\right]
+
e^{-a\gamma}D(P\|Q),
\end{equation}
which makes its interpretation transparent.

\subsubsection{Continuous-time process: random noiseless sampling}

We now return to matched estimation with the true erasure channel \eqref{eq:erasure-row} and allow the hidden input to evolve in physical time, as in Section~\ref{sec:ct-model}. Let $X_t\in\X$ be a stochastic process. Informative, non-erased events form a thinned Poisson process of rate $a\gamma$, independent of the input path, and at each such event time the mark reveals the current state $X_t$ exactly. Erasure events form an independent Poisson process of rate $\varepsilon\gamma$ and carry no information. Thus the model is exactly \emph{random noiseless sampling} of $X_t$ at rate $a\gamma$. Each informative event reveals the state at its own sampling time; subsequent evolution can restore uncertainty before the next reveal. \\
More generally, let $\mathcal G$ be any observation sigma-field and put
\begin{equation}
\pi_t(x)=P(X_t=x|\mathcal G).
\end{equation}
The posterior predictive row is
\begin{equation}
\E[W_t|\mathcal G]
=\varepsilon\delta_{\mathtt e}
+a\sum_x\pi_t(x)\delta_x,
\end{equation}
and hence
\begin{equation}\label{eq:erasure-entropy-loss}
\boxed{
\E D\!\left(W_t\middle\|\E[W_t|\mathcal G]\right)
=a H(X_t|\mathcal G).
}
\end{equation}
Accordingly, the matched causal and noncausal risks of Section~\ref{sec:ct-identities} reduce to integrated filtering and smoothing entropies.  With $Y_\gamma^{t-}$ denoting the observation history strictly before $t$,
\begin{align}
\cC_{P,P}(\gamma)
&=a\int_0^T H_P(X_t| Y_\gamma^{t-})\,dt,\label{eq:erasure-causal}\\
\cS_{P,P}(\gamma)
&=a\int_0^T H_P(X_t| Y_\gamma^T)\,dt.\label{eq:erasure-noncausal}
\end{align}
The general causal and noncausal identities therefore become
\begin{align}
I_P(X^T;Y_\gamma^T)
&=\gamma a\int_0^T H_P(X_t| Y_\gamma^{t-})\,dt,\label{eq:erasure-duncan}\\
\frac{d}{d\gamma}I_P(X^T;Y_\gamma^T)
&=a\int_0^T H_P(X_t| Y_\gamma^T)\,dt.\label{eq:erasure-smoothing-derivative}
\end{align}
Eliminating mutual information yields, for \(\gamma>0\), the especially simple entropy relation
\begin{equation}\label{eq:erasure-causal-noncausal}
\boxed{
\int_0^T H_P(X_t| Y_\gamma^{t-})\,dt
=
\frac1\gamma\int_0^\gamma
\int_0^T H_P(X_t| Y_\alpha^T)\,dt\,d\alpha.
}
\end{equation}
At \(\gamma=0\), the relation is understood by continuous extension; both sides equal \(\int_0^T H_P(X_t)\,dt\).
Thus, for random noiseless Poisson sampling of a finite-state process, the time-integrated filtering entropy at sampling intensity $a\gamma$ equals the uniform average, over lower intensities, of the time-integrated smoothing entropy.  The erasure channel therefore provides a closed-form bridge between the predictive-KL formulation and conditional entropy.

\FloatBarrier

\subsection{Binary symmetric Markov process through an erasure channel}\label{sec:markov-erasure}

The dynamic erasure model becomes unusually explicit when the hidden process is a stationary binary symmetric continuous-time Markov chain.  Let
\begin{equation}
X_t\in\{-1,+1\},
\text{ }
-1 \xrightleftharpoons[\nu]{\nu} +1,
\end{equation}
where \(\nu>0\) is the flip rate in either state. The stationary law is uniform, and for \(s\ge0\),
\begin{equation}\label{eq:telegraph-corr}
\Pp(X_{t+s}=x'| X_t=x)
=\frac{1+xx'e^{-2\nu s}}{2},
\text{ }
\E[X_tX_{t+s}]=e^{-2\nu s}.
\end{equation}
Pass \(X_t\) through the binary erasure channel of Section~\ref{sec:erasure}, with erasure probability \(\varepsilon\), and write
\begin{equation}
a=1-\varepsilon,
\text{ }
\lambda=a\gamma.
\end{equation}
By thinning, erased events carry no information and the informative events form a Poisson process of rate \(\lambda\), independent of the input path, at whose event times the current state \(X_t\) is observed perfectly.  Thus the problem is exactly Poisson sampling of the binary ``telegraph'' process at rate \(\lambda\).

To distinguish stationary quantities from finite-horizon losses, consider the bi-infinite stationary process and observation record. The nearest-sample formulas below are for \(\lambda>0\); at zero rate, both entropies have continuous value \(\log2\). Let
\begin{align}
c(\lambda)&\triangleq H(X_0| Y^{0-}_\lambda),\label{eq:telegraph-c-def}\\
s(\lambda)&\triangleq H(X_0| Y^{(-\infty,\infty)}_\lambda),\label{eq:telegraph-s-def}
\end{align}
where the first conditions on all informative samples strictly before time zero and the second on the full two-sided record.  These are, respectively, the stationary filtering and smoothing entropies. Both are entropies of the state at a typical stationary time, measured in nats. Multiplying by \(a\) gives the corresponding predictive KL risks per unit physical time; these are denoted by overlined quantities below. The mutual-information rate \(\mathcal I(\gamma)\) is measured in nats per unit physical time. Here \(\gamma\) is the total observation rate, \(\lambda=a\gamma\) is the informative observation rate, and \(\nu\) is the state flip rate. Their ratio \(\lambda/\nu\) compares how frequently the state is revealed with how frequently it changes.

\subsubsection{Causal filtering and a closed-form filtering entropy}

Let \(A\) be the age of the most recent informative sample before time zero and let \(S\in\{-1,+1\}\) be the state revealed there.  Since the sampling process is Poisson and independent of \(X\),
\begin{equation}
A\sim\operatorname{Exp}(\lambda).
\end{equation}
By the Markov transition law \eqref{eq:telegraph-corr},
\begin{align}
\Pp(X_0=S| A,S)&=\frac{1+e^{-2\nu A}}{2},\\
\E[X_0| A,S]&=S e^{-2\nu A}.
\end{align}
Hence the last revealed state and its age are a sufficient statistic for causal filtering. Once that state is known, earlier observations provide no additional information about the subsequent evolution. As the latest sample ages, possible intervening flips restore uncertainty about the current state. With
\begin{equation}
h_2(p)=-p\log p-(1-p)\log(1-p),
\end{equation}
the stationary filtering entropy is
\begin{equation}\label{eq:telegraph-c-integral}
\boxed{
c(\lambda)
=\int_0^\infty
\lambda e^{-\lambda u}
h_2\!\left(\frac{1-e^{-2\nu u}}{2}\right)du.
}
\end{equation}
This integral can be evaluated in terms of the digamma function.  Writing \(\psi\) for the digamma function and \(\gamma_E\) for Euler's constant,
\begin{equation}\label{eq:telegraph-c-closed}
\boxed{
c(\lambda)
=
\frac{\nu}{\lambda+2\nu}
\left[
\psi\!\left(1+\frac{\lambda}{4\nu}\right)
+\gamma_E+2\log2
\right].
}
\end{equation}
Appendix~\ref{app:digamma} evaluates \eqref{eq:telegraph-c-integral} by a convergent entropy series and digamma identities.
For the original erasure-channel loss, equation~\eqref{eq:erasure-entropy-loss} shows that the stationary causal KL risk per unit time is
\begin{equation}\label{eq:telegraph-causal-risk}
\boxed{
\overline{\cC}(\gamma)
=a\,c(a\gamma).
}
\end{equation}
Here and below the overline denotes the infinite-horizon per-unit-time limit of the corresponding integrated quantity.  To justify these limits for \(\lambda>0\), couple a finite observation record on \([0,T]\) to the two-sided stationary record.  The causal posteriors agree once an informative sample has occurred after zero.  By the nearest-neighbor Markov property used below, the smoothing posteriors agree whenever there is an informative sample on each side of the estimation time within \([0,T]\).  At time \(t\), the respective exceptional-event probabilities are \(e^{-\lambda t}\) and at most \(e^{-\lambda t}+e^{-\lambda(T-t)}\).  The conditional entropies lie between \(0\) and \(\log2\); averaging these bounds over \([0,T]\) makes the boundary contributions vanish as \(T\to\infty\).

\subsubsection{Mutual-information rate}

Define the stationary mutual-information rate
\begin{equation}
\mathcal I(\gamma)
\triangleq
\lim_{T\to\infty}\frac1T I(X^T;Y_\gamma^T).
\end{equation}
The causal identity \eqref{eq:ct-duncan} and \eqref{eq:telegraph-causal-risk} give
\begin{equation}\label{eq:telegraph-I-rate-basic}
\boxed{
\mathcal I(\gamma)
=\gamma a\,c(a\gamma)
=\lambda c(\lambda).
}
\end{equation}
Thus
\begin{equation}\label{eq:telegraph-I-rate}
\boxed{
\mathcal I(\gamma)
=
\frac{\lambda\nu}{\lambda+2\nu}
\left[
\psi\!\left(1+\frac{\lambda}{4\nu}\right)
+\gamma_E+2\log2
\right],
\text{ } \lambda=a\gamma.
}
\end{equation}
There is also a direct renewal interpretation.  Consecutive informative samples are separated by \(\Delta\sim\operatorname{Exp}(\lambda)\).  Given the previous revealed state, the entropy of the next revealed state is
\(
 h_2((1-e^{-2\nu\Delta})/2)
\).
Averaging gives \(c(\lambda)\) nats per informative arrival, and informative arrivals occur at rate \(\lambda\), recovering \eqref{eq:telegraph-I-rate-basic}.

\subsubsection{Noncausal smoothing}

For a direct description of the smoother, let the nearest informative observations on the left and right of time zero occur at distances
\begin{equation}
A,B\stackrel{\mathrm{iid}}{\sim}\operatorname{Exp}(\lambda),
\end{equation}
with revealed states \(L,R\in\{-1,+1\}\). By the Markov property, these two nearest observations are sufficient for smoothing. Given their states and times, samples farther into the past or future provide no additional information about \(X_0\); the smoother only needs to infer the state between these two exact reveals. Put
\begin{equation}
r=e^{-2\nu A},
\text{ }
q=e^{-2\nu B}.
\end{equation}
Then
\begin{equation}\label{eq:telegraph-smooth-mean}
\boxed{
\E[X_0| L,R,A,B]
=
\frac{Lr+Rq}{1+LRrq}.
}
\end{equation}
Moreover,
\begin{equation}
\Pp(LR=d| A,B)=\frac{1+d\,rq}{2},
\text{ } d\in\{-1,+1\}.
\end{equation}
Equations above provide a direct two-dimensional integral for \(s(\lambda)\).  A conditional-entropy identity reduces it to one dimension and justifies the derivative of the stationary information rate directly.  Put
\begin{equation*}
f(u)=h_2\!\left(\frac{1-e^{-2\nu u}}{2}\right).
\end{equation*}
For fixed gaps \(u,v>0\), the entropy chain rule and the Markov property give
\begin{equation*}
H(X_0| L,R,A=u,B=v)=f(u)+f(v)-f(u+v),
\end{equation*}
where the conditional entropy averages over the revealed endpoint states \(L,R\).  Averaging over the independent exponential gaps, whose sum has density \(\lambda^2u e^{-\lambda u}\), yields
\begin{equation*}
s(\lambda)=2c(\lambda)-\lambda^2\int_0^\infty u e^{-\lambda u}f(u)\,du
=c(\lambda)+\lambda c'(\lambda).
\end{equation*}
Here differentiation of \eqref{eq:telegraph-c-integral} is justified locally for every \(\lambda>0\) by \(0\le f\le\log2\) and exponential domination.  Consequently,
\begin{equation}\label{eq:telegraph-s-derivative}
\boxed{
s(\lambda)
=\frac{d}{d\lambda}\big[\lambda c(\lambda)\big].
}
\end{equation}
Let
\begin{equation}
A_\lambda
\triangleq
\psi\!\left(1+\frac{\lambda}{4\nu}\right)
+\gamma_E+2\log2,
\end{equation}
and write \(\psi_1\) for the trigamma function.  Differentiating \eqref{eq:telegraph-I-rate} yields
\begin{equation}\label{eq:telegraph-s-closed}
\boxed{
s(\lambda)
=
\frac{2\nu^2}{(\lambda+2\nu)^2}A_\lambda
+
\frac{\lambda}{4(\lambda+2\nu)}
\psi_1\!\left(1+\frac{\lambda}{4\nu}\right).
}
\end{equation}
Thus a two-sided hidden-Markov smoothing calculation collapses to a digamma-like expression by differentiating the causal information formula. Rewriting in terms of the original sampling-rate parameter \(\gamma\), the stationary noncausal KL risk per unit time is
\begin{equation}\label{eq:telegraph-smoothing-risk}
\boxed{
\overline{\cS}(\gamma)
=a\,s(a\gamma).
}
\end{equation}
The general matched continuous-time identities specialize to
\begin{align}
\mathcal I(\gamma)&=\gamma\overline{\cC}(\gamma),\label{eq:telegraph-ct-causal}\\
\mathcal I'(\gamma)&=\overline{\cS}(\gamma),\label{eq:telegraph-ct-smooth}\\
\overline{\cC}(\gamma)
&=\frac1\gamma\int_0^\gamma\overline{\cS}(\alpha)\,d\alpha,\text{ }\gamma>0.\label{eq:telegraph-ct-average}
\end{align}
Equivalently, in the effective sampling-rate parameter,
\begin{equation}\label{eq:telegraph-cs-average}
\boxed{
c(\lambda)
=\frac1\lambda\int_0^\lambda s(u)\,du,\text{ }\lambda>0.
}
\end{equation}
Both averaging formulas are understood at zero by continuous extension, with \(c(0)=s(0)=\log2\) and \(\overline{\cC}(0)=\overline{\cS}(0)=a\log2\).

The filtering-smoothing gap and the two information--estimation identities are shown in Figures~\ref{fig:telegraph-filter-smooth}--\ref{fig:telegraph-info-rate}.  The horizontal normalization $\lambda/\nu$ removes the arbitrary time scale of the Markov chain.

\begin{figure}[ht]
\centering
\includegraphics[width=0.70\linewidth]{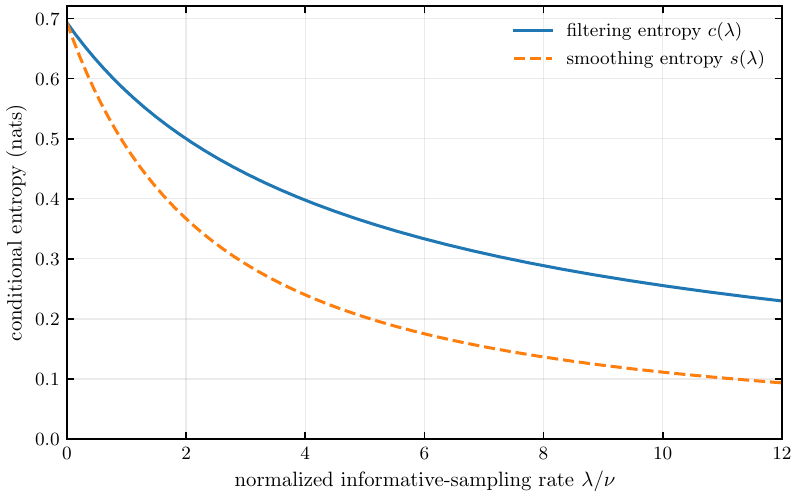}
\caption{Binary symmetric continuous-time Markov process observed through an erasure channel.  The filtering entropy $c(\lambda)$ exceeds the smoothing entropy $s(\lambda)$ because the smoother uses informative samples on both sides of the estimation time.}
\label{fig:telegraph-filter-smooth}
\end{figure}

\begin{figure}[ht]
\centering
\includegraphics[width=0.70\linewidth]{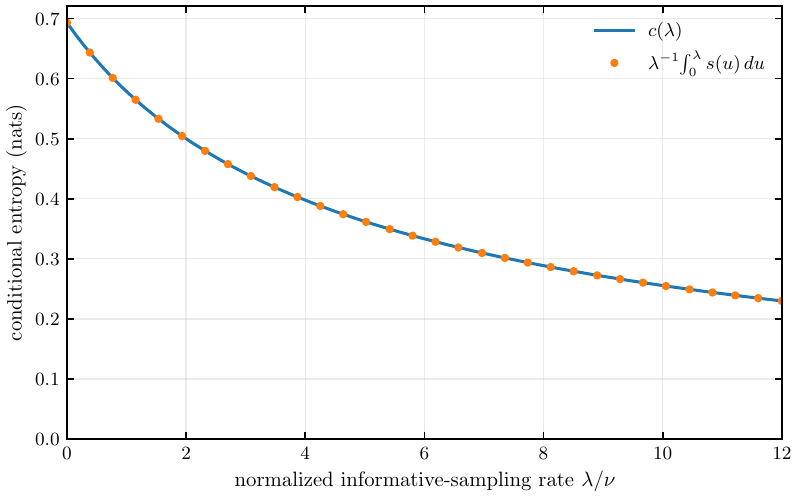}
\caption{Causal/noncausal averaging identity for the Markov/erasure model.  The solid curve is $c(\lambda)$; the markers evaluate $\lambda^{-1}\int_0^\lambda s(u)\,du$.}
\label{fig:telegraph-average}
\end{figure}

\begin{figure}[ht]
\centering
\includegraphics[width=0.70\linewidth]{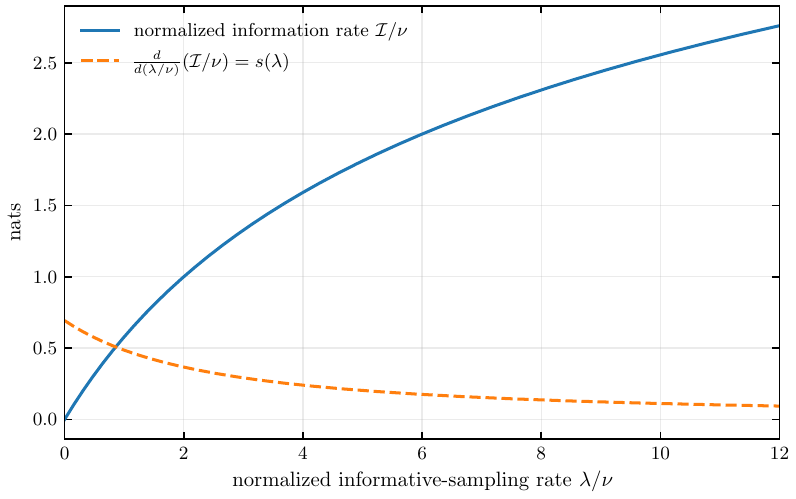}
\caption{Normalized mutual-information rate for the Markov/erasure model.  With $\rho=\lambda/\nu$, $\mathcal I/\nu=\rho c(\lambda)$ and $d(\mathcal I/\nu)/d\rho=s(\lambda)$.}
\label{fig:telegraph-info-rate}
\end{figure}

\FloatBarrier

\subsubsection{Finite-horizon form and asymptotic regimes}
The stationary formulas arise as limits of an exact finite-horizon expression.  Suppose \(X_0\) is stationary but there are no observations before time zero.  At time \(t\), either no informative sample has yet occurred, or the age of the latest informative sample is \(u<t\).  Therefore
\begin{equation}\label{eq:telegraph-finite-filter}
H(X_t| Y_\lambda^{t-})
=
e^{-\lambda t}\log2
+
\int_0^t \lambda e^{-\lambda u}
h_2\!\left(\frac{1-e^{-2\nu u}}{2}\right)du.
\end{equation}
For the equivalent noiseless Poisson-sampling model,
\begin{equation}\label{eq:telegraph-finite-I}
\boxed{
\begin{aligned}
I(X^T;Y_\lambda^T)
={}&(1-e^{-\lambda T})\log2+\lambda^2\int_0^T
(T-u)e^{-\lambda u}
h_2\!\left(\frac{1-e^{-2\nu u}}{2}\right)du.
\end{aligned}
}
\end{equation}
In \eqref{eq:telegraph-finite-I}, the first term is the entropy of the first revealed state, weighted by the probability that an informative sample occurs. Later samples contribute the conditional entropy of each revealed state given the preceding reveal, producing the integral term. The first term is bounded by \(\log2\) and vanishes after division by \(T\) in the information-rate limit. Differentiation with respect to \(\lambda\) gives the integrated finite-horizon smoothing entropy by Theorem~\ref{thm:ct-matched}. \\
The closed forms also reveal two informative limits.  At low sampling rate,
\begin{equation}\label{eq:telegraph-low}
\lambda\ll\nu:
\text{ }
c(\lambda),s(\lambda)\longrightarrow\log2,
\text{ }
\mathcal I(\gamma)\sim\lambda\log2.
\end{equation}
Each rare informative sample therefore contributes essentially one bit, or \(\log2\) nats. At high sampling rate, with \(\rho=\lambda/\nu\to\infty\),
\begin{align}
 c(\lambda)
 &=
 \frac{\log\rho+\gamma_E}{\rho}
 +O\!\left(\frac{\log\rho}{\rho^2}\right),\label{eq:telegraph-high-c}\\
 s(\lambda)
 &=
 \frac1\rho
 +O\!\left(\frac{\log\rho}{\rho^2}\right),\label{eq:telegraph-high-s}\\
 \mathcal I(\gamma)
 &=
 \nu\left[\log\rho+\gamma_E
 +O\!\left(\frac{\log\rho}{\rho}\right)\right].\label{eq:telegraph-high-I}
\end{align}
The extra logarithmic factor in the filtering entropy reflects uncertainty about a rare unobserved flip since the latest sample: the entropy of a rare alternative contains such a factor. The future reveal helps the smoother determine whether a flip occurred between the neighboring samples, reducing its leading entropy to order \(1/\rho\), as in \eqref{eq:telegraph-high-s}. The information rate does not saturate despite the binary state alphabet: increasingly dense observations localize the continuously valued jump times more accurately, producing logarithmically increasing timing information. \\
For comparison with more conventional estimation losses, the causal posterior mean is \(S e^{-2\nu A}\). Under the \(\{-1,+1\}\) encoding, the causal minimum mean-square error and minimum classification error are, respectively,
\begin{align}
\operatorname{mmse}_{\rm c}(\lambda)
&=\frac{4\nu}{\lambda+4\nu},\label{eq:telegraph-mmse}\\
P_{\rm e,c}(\lambda)
&=\frac{\nu}{\lambda+2\nu}.\label{eq:telegraph-map}
\end{align}
These quantities are not the losses appearing in the Poisson information identity, but they provide useful checks on the filtering posterior.

\FloatBarrier
\subsection{An image example: reconstruction and prediction}\label{sec:catalogue-example}
We return to the static model with an image selected uniformly from a
collection of candidate images known to the estimator. From sparse,
noisy pixel measurements, we compare how accurately the image can be
reconstructed and how well another measurement can be predicted. The
example shows that accurate reconstruction can coexist with persistent
prediction error when the assumed measurement noise is incorrect.

\noindent We use two test catalogues of 4-level quantized images: the first contains 1,000 CIFAR-10 images, with 100 images from each class, at native $32\times32$ RGB resolution~\cite{krizhevsky2009tiny} (so each image has $d=3072$ scalar color coordinates) and the second contains all 400 images of the AT\&T face database, with ten photographs of each of 40 subjects, at native $92\times112$ grayscale resolution, provided by AT\&T Laboratories Cambridge~\cite{att_faces,samaria1994parameterisation} so here $d=10304$. 
Each 8-bit intensity $u_j$ is quantized to $x_j=\lfloor u_j/64\rfloor\in\{0,1,2,3\}$.

\noindent An elementary observation selects a coordinate $J$ uniformly with replacement and returns a noisy category $Z$ through a fixed, strictly
positive four-symbol channel $A$. Thus the whole image is the hidden input, and $(J,Z)$ is one DMC output, with row
\begin{equation}\label{eq:catalogue-row}
W_x(j,z)=\frac{1}{d}A(z|x_j).
\end{equation}
Conditional on the image, successive readings are independent, and their
number is $N_\gamma\sim\Pois(\gamma)$, independently of the input.
The budget $\gamma$ is therefore the expected total number of scalar
readings: one color component for CIFAR, or one grayscale intensity for
the faces. The channel and computation details appear in
Appendix~\ref{app:catalogue-methods}.

\noindent All estimators receive the same record generated through $A$ under the uniform prior. For channel mismatch, retain that prior but let the estimator assume $\widetilde A=0.6A+0.4I_4$, where $I_4$ is the identity matrix. Reconstruction computes the posterior mean over the candidate images in the catalogue, whereas prediction averages future noisy output probabilities under the model's assumed channel.

Figure~\ref{fig:catalogue-montage} illustrates three fixed examples from each catalogue, at $\gamma=32$ for CIFAR and $\gamma=64$ for faces. In the second column, unmeasured coordinates are rendered in gray. The matched (column 3) and mismatched (column 5) catalogue estimators update posterior probabilities over complete images. A reading at one coordinate can therefore change estimates at coordinates that have not been observed. In contrast, the coordinate-factorized baseline (column 4) updates each coordinate independently from its marginal prior, using the true channel $A$. Lacking whole-image coupling, its unmeasured coordinates remain blurred at their marginal means, illustrating how prior structure allows a sparse record to inform the entire image.

\begin{figure}[htbp]
\centering
\includegraphics[width=\linewidth]{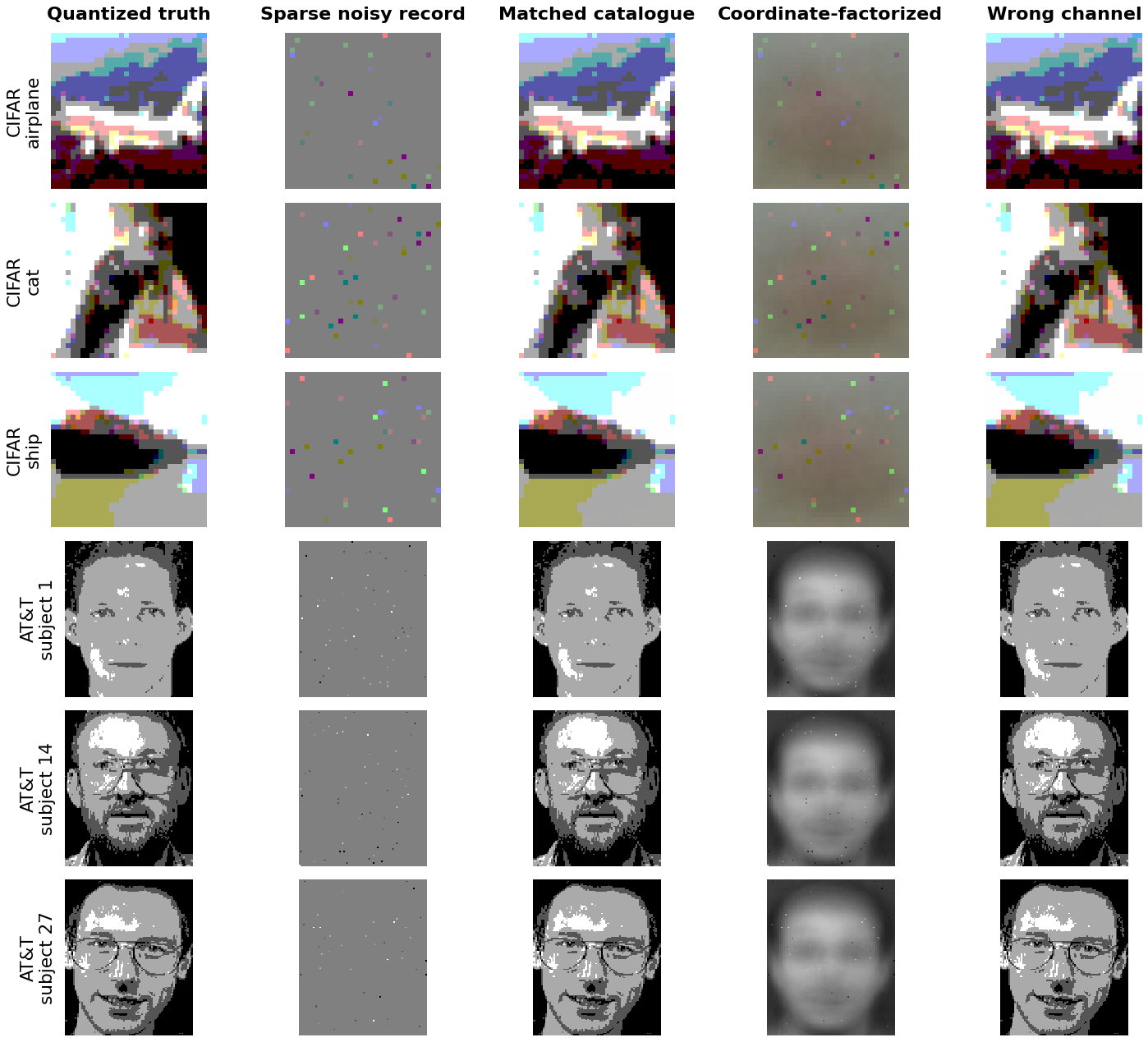}
\caption{Reconstruction from sparse, noisy pixel measurements of an image
selected from a collection of candidates known to the estimator. The upper
three rows are CIFAR-10 images~\cite{krizhevsky2009tiny} at $\gamma=32$;
the lower three are faces from AT\&T Laboratories Cambridge~\cite{att_faces}
at $\gamma=64$. Columns show the quantized reference, the noisy-reading
display, matched catalogue posterior mean, coordinate-factorized
reconstruction, and incorrect-channel catalogue posterior mean. Gray
placeholders fill unobserved scalar components. Realized reading counts
are 27, 39, and 31 for CIFAR, and 58, 55, and 44 for faces. All reconstructions in a row use the same record.}
\label{fig:catalogue-montage}
\end{figure}

Figure~\ref{fig:catalogue-prediction} summarizes 6,000 CIFAR trajectories. In panel~(a), normalized mutual information and accumulated matched predictive KL loss~\eqref{eq:integrated-matched} agree closely. The interpolated information curve reaches $95\%$ of the catalogue entropy at $\gamma\simeq 21.4$ readings. This threshold concerns uncertainty about which catalogue image was selected.

Panel~(b) shows both matched and mismatched reconstruction errors approaching zero. Predictive KL, rather than reconstruction MSE, is the loss appearing in the information identity. For each image, the nominal channel row closest to its true row in KL divergence corresponds to that same image. Consequently, both posteriors concentrate on the true candidate.

Panel~(c) exposes the contrast between state identification and output prediction. While matched predictive risk decays to zero, mismatched predictive risk plateaus at the theoretical floor $\kappa \simeq 0.02903$ nats/reading (Proposition~\ref{cm:prop:high}). Even with the image identified, postulating $\widetilde A$ instead of $A$ permanently impairs future observation prediction. This complements Example~\ref{cm:ex:label-swap}: there, identical predictions coexist with different beliefs about the hidden input. Here, both posteriors concentrate on the same image while their predictions remain different. Because the catalogue is known in advance, this experiment illustrates candidate identification rather than arbitrary image restoration.

\begin{figure}[htbp]
\centering
\includegraphics[width=\linewidth]{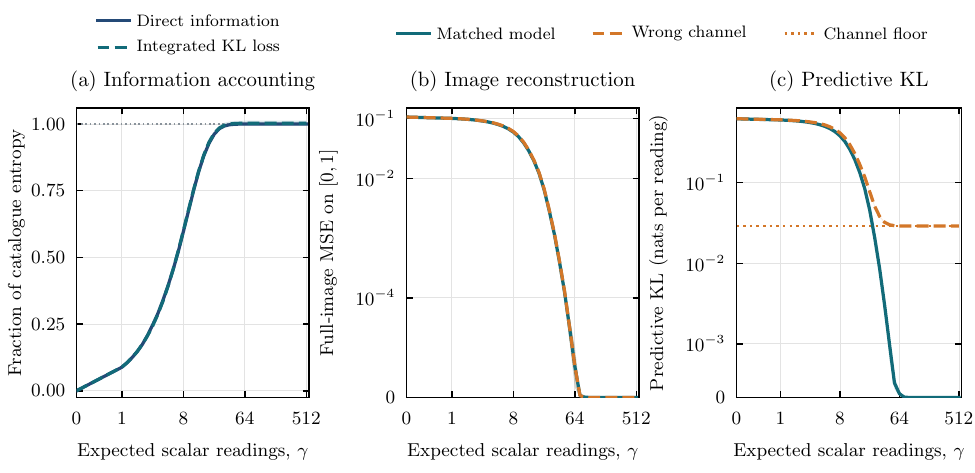}
\caption{Information, reconstruction, prediction for the 1,000-image
CIFAR catalogue. (a) Normalized mutual information and integrated matched predictive
KL loss. (b) Matched and incorrect-channel
posterior-mean reconstruction MSE over all coordinates of the quantized
image scaled to $[0,1]$. (c) Predictive KL loss and the calculated
incorrect-channel floor $\kappa$.
The horizontal axis gives the expected total number of scalar readings. Budget axes are linear up to one reading and log. thereafter; the MSE and KL axes have corresponding linear regions below $10^{-5}$ and $10^{-3}$.}
\label{fig:catalogue-prediction}
\end{figure}
\FloatBarrier

\section{Discussion \& Conclusion}\label{sec:discussion}
\begingroup
\clubpenalties 1 10000
\widowpenalties 1 10000

Poissonization turns the number of observations through a fixed DMC into a continuous budget, where the posterior-mean channel row predicts the next output and the minimum expected KL loss equals the rate of mutual-information gain. This finite-alphabet formulation of established Poisson information--estimation relations yields results on curvature, higher-order dependence, row identifiability, and capacity: dependence between fresh outputs explains diminishing information gains, conditional co-information gives the higher derivatives, and maximizing accumulated matched prediction loss over the input prior gives the static capacity.

Among priors achieving one-use capacity, the second-order term selects those with the least redundancy between two outputs. For the coupon-collector channel, the exact decomposition through observed-subset distributions extends the comparison to every budget. Balanced input families preserve the leading capacity terms with fewer possible inputs, as shown by the explicit losses for the matching and Fano constructions; at the same time, the Fano prior illustrates the limit of second-order agreement, as its information deficit remains positive at every positive budget.

Under the stated support and finiteness conditions, output relative entropy equals the accumulated excess prediction loss, while the joint-law decomposition separates this observable discrepancy from disagreement between the latent posteriors. In the label-swap and erasure examples, input posteriors can disagree even when output predictions agree, or agree even when output predictions differ. Total prediction regret remains finite precisely when every true row that can occur appears among the nominal rows with positive prior weight. Otherwise, regret grows linearly at a rate given by the average minimum KL approximation error; under the stated uniqueness condition, that error is also the limiting expected predictive loss.

For independently Poisson-sampled processes, filtering and smoothing estimate the same changing channel row from different observation histories. Their losses, integrated over time, satisfy an exact averaging relation across sampling rates, both in the matched model and under simultaneous process-law and channel mismatch. The explicit Markov and erasure formulas distinguish static from dynamic information acquisition. While static capacity is bounded by the logarithm of the number of distinct rows, information per unit time in the binary Markov example grows logarithmically with the sampling rate, because denser samples continue to localize the jump times.

The image example uses collections of CIFAR-10 images and AT\&T faces observed through sparse, noisy pixel measurements. A posterior over complete images lets observed coordinates inform unobserved ones; removing dependence leaves unobserved coordinates at their marginal means (Figure~\ref{fig:catalogue-montage}). The aggregate CIFAR results show agreement between information computed from posterior entropy and accumulated matched predictive KL loss (Figure~\ref{fig:catalogue-prediction}). Although both the matched and mismatched estimators recover the true image as the budget grows, the mismatched predictor hits a nonzero KL-loss floor. Knowing the image does not fix the predictor when the assumed measurement channel is incorrect.

\par\endgroup

\section*{Acknowledgments}
\addcontentsline{toc}{section}{Acknowledgments}
OpenAI’s GPT 5.6 Sol and 6 Astra were used to assist with editing and checking mathematical proofs and arguments. We thank Aayush Rajesh for his comments.
\phantomsection
\addcontentsline{toc}{section}{References}
\nocite{cover2006elements}
\bibliographystyle{IEEEtran-linked}
\begingroup
\bibliography{references}
\endgroup
\clearpage
\appendix

\section{High-intensity channel-mismatch limits}\label{app:channel-high}

This is the full proof of Proposition~\ref{cm:prop:high}. The notation \(J_n\) is as in the proof of Theorem~\ref{cm:thm:main}.

\begin{proof}
First fix \(x\) with \(P(x)>0\) and condition on \(X=x\).  Put
\[
S_x=\{z:W_x(z)>0\},\text{ }
L_{j,n}=\prod_{i=1}^n\widetilde W_j(Z_i),\text{ }
L^{\mathrm{true}}_{x,n}=\prod_{i=1}^n W_x(Z_i).
\]
If \(\KL{W_x}{\widetilde W_j}<\infty\), the log likelihood ratio has bounded increments on \(S_x\), so the strong law gives
\[
\frac1n\log\frac{L_{j,n}}{L^{\mathrm{true}}_{x,n}}
\longrightarrow-\KL{W_x}{\widetilde W_j}
\text{ }\text{almost surely}.
\]
If this divergence is infinite, some \(z\in S_x\) has \(\widetilde W_j(z)=0\).  That symbol is eventually observed almost surely, after which \(L_{j,n}=0\).  Comparing the logarithm of the finite sum
\[
Q_n^{\widetilde M}(Z^n)=\sum_{j\in\mathcal A_Q}Q(j)L_{j,n}
\]
with its largest summand therefore yields
\[
\frac1n\log\frac{L^{\mathrm{true}}_{x,n}}{Q_n^{\widetilde M}(Z^n)}
\longrightarrow d(x)
\text{ }\text{almost surely}.
\]
This convergence also holds in expectation.  Choose an active \(j_0\) with finite divergence and set
\[
a_x=\min_{z\in S_x}W_x(z)>0,\text{ }
b_x=\min_{z\in S_x}\widetilde W_{j_0}(z)>0.
\]
On every history of positive true probability,
\begin{equation}\label{cm:eq:high-logratio-bound}
n\log a_x
\le
\log\frac{L^{\mathrm{true}}_{x,n}}{Q_n^{\widetilde M}(Z^n)}
\le
-\log Q(j_0)+n\log(1/b_x).
\end{equation}
The lower bound uses \(Q_n^{\widetilde M}(Z^n)\le1\), and the upper bound uses \(Q_n^{\widetilde M}(Z^n)\ge Q(j_0)b_x^n\).  Dominated convergence and the finiteness of \(\X\) now give \(J_n/n\to\bar d\).  Since
\[
J_n=I_{P,M}(X;Z^n)+\KL{P_n^M}{Q_n^{\widetilde M}},
\text{ } 0\le I_{P,M}(X;Z^n)\le H(P),
\]
the fixed-count output divergence divided by \(n\) also tends to \(\bar d\).  Its growth is at most linear by \eqref{cm:eq:high-logratio-bound}.  Poisson averaging preserves this slope: for every \(\epsilon>0\), the fixed-count error is bounded by \(\epsilon n+C_\epsilon\), and hence its averaged error is bounded by \(\epsilon\gamma+C_\epsilon\).  Dividing by \(\gamma\) proves \eqref{cm:eq:high-Dout}. \\
For the risk limit, fix a unique minimizer \(j^*=j^*(x)\).  For any other active \(j\) with finite divergence, the bounded log likelihood-ratio increments
\[
\log\frac{\widetilde W_j(Z_i)}{\widetilde W_{j^*}(Z_i)}
\]
have strictly negative mean
\(-[\KL{W_x}{\widetilde W_j}-d(x)]\).
Writing this positive divergence gap as \(\Delta_j\), Hoeffding's inequality gives
\[
\Pp_{W_x}\left\{
\log\frac{L_{j,n}}{L_{j^*,n}}\ge-\frac{n\Delta_j}{2}
\right\}\le e^{-c_jn}
\]
for some \(c_j>0\).  If the divergence for \(j\) is infinite, the probability that its likelihood has not yet become zero is at most \((1-W_x(z))^n\) for a symbol \(z\in S_x\) with \(\widetilde W_j(z)=0\).  A union bound over the finite active family shows that, for constants \(C_x,c_x>0\), there are events \(E_{x,n}\) such that
\[
\Pp_{W_x}(E_{x,n}^c)\le C_xe^{-c_xn},
\text{ }
1-Q_{\widetilde M}(j^*| Z^n)\le C_xe^{-c_xn}
\text{ }\text{on }E_{x,n}.
\]
The posterior includes fixed factors like \(Q(j)/Q(j^*)\), which are absorbed into \(C_x\). \\
Expected KL loss requires control outside these events. Let
\[
\ell_{x,n}(Z^n)=
\KL{W_x}{\widehat{\widetilde W}_{Q,n}(\cdot| Z^n)}.
\]
Using the same \(j_0,b_x\) as above, the nominal posterior weight of \(j_0\) is at least \(Q(j_0)b_x^n\).  Therefore its predictive mass at every \(z\in S_x\) is at least \(Q(j_0)b_x^{n+1}\), and
\begin{equation}\label{cm:eq:high-risk-domination}
0\le\ell_{x,n}(Z^n)
\le-\log Q(j_0)+(n+1)\log(1/b_x).
\end{equation}
On \(E_{x,n}\), the predictor is exponentially close to \(\widetilde W_{j^*}\), which is positive on \(S_x\); the logarithm is Lipschitz on a sufficiently small neighborhood of its positive coordinates.  Hence, after enlarging constants to cover small \(n\),
\[
\E_{W_x}\big|\ell_{x,n}(Z^n)-d(x)\big|
\le C_x(1+n)e^{-c_xn}.
\]
This proves convergence in expected loss, not only posterior convergence.  Averaging over \(X\) and the independent Poisson count proves \eqref{cm:eq:high-risk}.  The same concentration bound proves the asserted convergence in probability.
\end{proof}

\section{Continuous-time sampling and compensation proof}\label{app:ct-proof}

We prove Theorem~\ref{cm:thm:ct} under \eqref{cm:eq:ct-assumptions}. The common-channel and matched identities follow by specialization. Adding a uniformly timed observation gives the smoothing derivative; factoring the record likelihood in chronological order gives the causal identity.

\begin{proof}
Let \(A=(U,Z)\), where \(U\) is uniform on \([0,T]\), independently of the path. Conditional on \(X^T\) and \(U=t\), the true mark law is \(W_t\) and the nominal mark law is \(\widetilde W_{X_t}\). The corresponding kernels are
\begin{equation}\label{eq:ct-proof-kernel}
K_{X^T}(dt,z)=\frac1T W_t(z)\,dt,
\text{ }
\widetilde K_{X^T}(dt,z)=\frac1T\widetilde W_{X_t}(z)\,dt.
\end{equation}
Under each model, generate \(A_1,A_2,\ldots\) independently conditional on \(X^T\), and independently generate \(N\sim\Pois(\gamma T)\). Sorting the \(N\) pairs by time gives the actual observation record. Conversely, a uniformly random permutation of a sorted record restores the artificial labels with a law independent of the path and common to both models. These operations preserve information, output and joint relative entropy, and the posterior on the path. \\
A uniformly timed observation contributes \(K_T/T\) to the expected conditional channel divergence. Conditional independence and the KL chain rule therefore give
\begin{equation}\label{cm:eq:ct-joint-fixed}
D(P^M_{X^T,A^n}\|Q^{\widetilde M}_{X^T,A^n})
=D_0+\frac nT K_T.
\end{equation}
Mixing over the common Poisson count, whose mean is \(\gamma T\), gives
\begin{equation}\label{cm:eq:ct-joint-divergence}
D(P^M_{X^T,Y_\gamma^T}\|Q^{\widetilde M}_{X^T,Y_\gamma^T})
=D_0+\gamma K_T.
\end{equation}
These finite divergences imply absolute continuity of the true joint and output laws with respect to their nominal counterparts. In particular, nominal posterior predictors are well defined on true records up to the null sets relevant to the integrated losses. \\
For a fixed number of pairs, we have
\[
G_n=I_{P,M}(X^T;A^n)+D(P^M_{A^n}\|Q^{\widetilde M}_{A^n}).
\]
The times \(U^n\) are independent of the path, so
\[
I_{P,M}(X^T;A^n)=I_{P,M}(X^T;Z^n| U^n)\le n\log|\Z|.
\]
Data processing in \eqref{cm:eq:ct-joint-fixed} bounds the output divergence and yields
\begin{equation}\label{eq:ct-proof-growth}
0\le G_n\le D_0+n\left(\log|\Z|+\frac{K_T}{T}\right).
\end{equation}
Also, \(G_n=D(P^M_{X^T,A^n}\|P_{X^T}\otimes Q^{\widetilde M}_{A^n})\). Its log likelihood ratio \(L_n\) is therefore integrable: the expected negative part of a probability log density ratio is at most \(1/e\), and hence \(\E_{P,M}|L_n|\le G_n+2/e\). Relative to uniform-time measure and counting measure on marks, write \(q_n\) for the density of \(Q^{\widetilde M}_{A^n}\); then
\begin{equation}\label{eq:ct-proof-likelihood}
G_n=\E_{P,M}\log\frac{\prod_{i=1}^n W_{U_i}(Z_i)}{q_n(A^n)}.
\end{equation}
Factoring off one additional pair gives
\begin{equation}\label{eq:ct-proof-increment}
G_{n+1}-G_n
=\frac1T\int_0^T
\E_{P,M}\KL{W_t}{\E_{Q,\widetilde M}[\widetilde W_{X_t}| A^n]}\,dt.
\end{equation}
The new time is independent of the path and the existing pairs under both models, so conditioning on its value does not change the existing path posterior. By the observation equivalence,
\[
\mathcal J_T(\gamma)=\E[G_N]
\le D_0+\gamma\bigl(T\log|\Z|+K_T\bigr).
\]
Lemma~\ref{lem:poisson-diff} applies by \eqref{eq:ct-proof-growth}. Differentiating the mean parameter \(\gamma T\) contributes a factor \(T\), which cancels the \(1/T\) in the increment:
\begin{equation}\label{eq:ct-proof-smoothing-derivative}
\mathcal J_T'(\gamma)
=T\E[G_{N+1}-G_N]
=\Smis(\gamma).
\end{equation}
The nominal posterior after the random number of pairs is exactly the full-record smoother. Uniform convergence of the Poisson derivative series gives continuity, including the right derivative at zero. Since \(G_0=0\), integration proves
\(\mathcal J_T(\gamma)=\int_0^\gamma\Smis(\alpha)\,d\alpha\). \\
For the causal identity, let \(\gamma>0\) and consider a record
\(y=((\tau_i,z_i))_{i=1}^n\), with \(0<\tau_1<\cdots<\tau_n<T\).
Relative to ordered-time Lebesgue measure and counting measure on marks, its true conditional density and its nominal marginal density are
\begin{equation}\label{eq:ct-proof-record-density}
\begin{aligned}
p(y| X^T)&=e^{-\gamma T}\gamma^n\prod_{i=1}^n W_{\tau_i}(z_i),\\
q(y)&=e^{-\gamma T}\gamma^n\prod_{i=1}^n
\widehat{\widetilde W}_{Q,\tau_i}^{\mathrm c,\gamma}(z_i).
\end{aligned}
\end{equation}
The clock is independent of the path under both models. Thus neither the absence of arrivals nor knowledge of future unmarked times changes a path posterior based on preceding marks; sequential averaging gives precisely the strict-past predictors in the second product. The independent clock hits no input jump, so left-limit and current-row values agree at its arrivals. The common clock factor cancels, yielding
\begin{equation}\label{eq:ct-proof-causal-score}
\log\frac{p(Y_\gamma^T| X^T)}{q(Y_\gamma^T)}
=\sum_{\tau_i\le T}
\log\frac{W_{\tau_i}(Z_i)}{\widehat{\widetilde W}_{Q,\tau_i}^{\mathrm c,\gamma}(Z_i)}.
\end{equation}
Its expectation is \(\mathcal J_T(\gamma)\), and its absolute expectation is finite by the preceding bound on \(L_n\), averaged over the Poisson count. \\
Include the input path in the initial sigma-field and use predictable versions throughout. For any nonnegative predictable \(g(t,z)\), conditional Poisson increments and monotone approximation give the compensation formula
\begin{equation}\label{eq:ct-proof-compensation}
\E_{P,M}\sum_{\tau_i\le T}g(\tau_i,Z_i)
=\gamma\E_{P,M}\int_0^T\sum_zW_{t-}(z)g(t,z)\,dt.
\end{equation}
Apply this first to the negative part of the log score; scores where \(W_{t-}(z)=0\) may be set to zero. For probability vectors \(u,v\),
\begin{equation}\label{eq:ct-proof-negative-score}
\sum_{z:u(z)>0}u(z)\left[\log\frac{v(z)}{u(z)}\right]_+
\le\frac1e\sum_zv(z)=\frac1e,
\end{equation}
because \(a\log(b/a)\le b/e\) for \(0<a<b\). The expected sum of negative event scores is consequently at most \(\gamma T/e\). Since their total log score is already integrable, the sum of positive event scores is integrable as well. Compensation may therefore be applied to both parts, giving
\[
\mathcal J_T(\gamma)
=\gamma\E_{P,M}\int_0^T
\KL{W_{t-}}{\widehat{\widetilde W}_{Q,t}^{\mathrm c,\gamma}}\,dt
=\gamma\Cmis(\gamma).
\]
The linear terms of the coordinatewise Poisson loss cancel because both rows have total mass one, giving the same KL loss as in the Poisson information--estimation framework of \cite{atar2012poisson}. The direct argument here also permits zero channel entries. At \(\gamma=0\), the record is empty and \(\mathcal J_T(0)=0\). The increment formula with \(n=0\), together with \eqref{eq:ct-proof-growth}, proves that
\begin{equation}\label{eq:ct-proof-zero}
\begin{aligned}
\Cmis(0)=\Smis(0)
&=\E_P\int_0^T\KL{W_t}{\E_Q\widetilde W_{X_t}}\,dt\\
&=T G_1\le T D_0+T\log|\Z|+K_T<\infty.
\end{aligned}
\end{equation}
The derivative series and the integral identity show that \(\Smis\) and the continuous extension of \(\mathcal J_T(\gamma)/\gamma=\Cmis(\gamma)\) have this value at zero. For positive rates, continuity of \(\Cmis\) follows from \(\mathcal J_T(\gamma)/\gamma\). This proves all assertions of Theorem~\ref{cm:thm:ct}. \\
Setting \(\widetilde M=M\) makes \(K_T=0\) and reduces \(\mathcal J_T,\Cmis,\Smis\) to \(\cJ_{P,Q},\cC_{P,Q},\cS_{P,Q}\), respectively, proving Theorem~\ref{thm:ct-mismatch}. Taking \(Q=P\) as well removes the output divergence and proves Theorem~\ref{thm:ct-matched}.
\end{proof}

\section{Posterior covariance expansion and remainder}\label{app:covariance}

Use \(q_\pi,r_\pi\) from Section~\ref{sec:curvature}.\\
Remove zero coordinates of \(q_\pi\), whose corresponding rows and columns in \(r_\pi\) also vanish.  On the remaining support, set \(b=q_\pi\otimes q_\pi\) and \(\delta=r_\pi-b\).  If \(|\delta_i|\le b_i/2\) coordinatewise, Taylor's theorem gives the expansion \eqref{eq:quad-second}, with remainder
\[
|\mathcal R|\le\frac23\sum_i\frac{|\delta_i|^3}{b_i^2}.
\]
Indeed, for \(g(u)=(1+u)\log(1+u)\), one has \(|g'''(u)|\le4\) when \(|u|\le1/2\), and the linear terms cancel because \(\sum_i\delta_i=0\).  For a family of posteriors whose positive predictive coordinates stay uniformly bounded away from zero, \(\mathcal R=o(\|\delta\|_2^2)\) as \(\delta\to0\).  Full support at each posterior alone does not imply this unweighted remainder.

\section{Support and growth in the capacity optimality conditions}\label{app:capacity-support}

We justify \eqref{eq:cap-kkt-predictive}, with the input fixed at \(x\) and the predictor using \(P\). Write \(S_x=\operatorname{supp}W_x\).  Finiteness at any positive budget implies that some \(x'\) with \(P(x')>0\) satisfies \(S_x\subseteq\operatorname{supp}W_{x'}\): otherwise, a record containing every symbol in \(S_x\) would have positive probability under \(x\) and zero probability under the \(P\)-mixture.  Conversely, such an \(x'\) gives, with \(m=\min_{z\in S_x}W_{x'}(z)>0\) and \(P_n\) the \(P\)-mixture law of \(Z^n\),
\[
0\le d_n\triangleq D(W_x^{\otimes n}\|P_n)
\le-\log P(x')+n\log(1/m).
\]
Thus this support condition is also sufficient for finiteness at every finite positive budget.  Factoring off the last observation gives
\[
d_{n+1}-d_n
=\E_x D\!\left(W_x\middle\|\E_P[W| Z^n]\right).
\]
Since \(G_x(P,\gamma)=\E[d_{N_\gamma}]\) and \(d_0=0\), the linear-growth form of Lemma~\ref{lem:poisson-diff} justifies differentiation and integration, proving \eqref{eq:cap-kkt-predictive}. \\
If the support condition fails, take a shortest record having positive probability under \(x\) but zero probability under \(P\).  Its preceding history is \(P\)-positive, and the \(P\)-predictor assigns zero probability to an \(x\)-possible next symbol.  This history has positive probability under the fixed-input experiment at every \(s>0\), so \(\E_xD(W_x\|\widehat W_P^s)=\infty\).  Hence both sides of \eqref{eq:cap-kkt-predictive} are infinite for every \(\gamma>0\); this conclusion is independent of arbitrary predictor versions on \(P\)-null histories.  At \(\gamma=0\), both the divergence and the zero-length integral are zero.

\section{Evaluation of the stationary filtering entropy}\label{app:digamma}

We evaluate the integral \eqref{eq:telegraph-c-integral}; the smoothing formula then follows from \eqref{eq:telegraph-s-derivative}. For \(0\le v\le1\),
\[
h_2\!\left(\frac{1-v}{2}\right)
=\log2-\sum_{k=1}^{\infty}\frac{v^{2k}}{2k(2k-1)}.
\]
This follows by integrating the power series for the derivative on \(0<v<1\), with continuity at the endpoints. The summands are nonnegative, so Tonelli's theorem permits termwise integration against the exponential age density. For \(\lambda>0\), put \(\alpha=\lambda/(4\nu)\). Then
\[
c(\lambda)=\log2-\sum_{k=1}^{\infty}
\frac{\alpha}{2k(2k-1)(k+\alpha)}.
\]
The partial-fraction identity
\[
\frac{\alpha}{2k(2k-1)(k+\alpha)}
=-\frac{1}{2k}
+\frac{\alpha}{2\alpha+1}\frac{1}{k-\frac12}
+\frac{1}{2(2\alpha+1)}\frac{1}{k+\alpha}
\]
is used with the divergent harmonic terms grouped before summation. The convergent digamma identity
\[
\sum_{k=1}^{\infty}\left(\frac{1}{k+\beta}-\frac1k\right)
=-\psi(1+\beta)-\gamma_E,\text{ } \beta>-1,
\]
together with \(\psi(1/2)=-\gamma_E-2\log2\), gives
\[
c(\lambda)
=\frac{\psi(1+\alpha)+\gamma_E+2\log2}{2(2\alpha+1)}
=\frac{\nu}{\lambda+2\nu}
\left[\psi\!\left(1+\frac{\lambda}{4\nu}\right)+\gamma_E+2\log2\right].
\]
This proves \eqref{eq:telegraph-c-closed}. Differentiating \(\lambda c(\lambda)\), justified by the dominated integral calculation in the main text, gives \eqref{eq:telegraph-s-closed}. Continuity at zero yields \(c(0)=s(0)=\log2\); the usual large-argument expansions of \(\psi\) and \(\psi_1\) give \eqref{eq:telegraph-high-c}--\eqref{eq:telegraph-high-I}.

\FloatBarrier

\section{Image reconstruction and prediction computations}\label{app:catalogue-methods}

\paragraph{Image collections and channel parameters:}
The CIFAR catalogue comprises 100 images selected without replacement from each class of the CIFAR-10 training set~\cite{krizhevsky2009tiny}. The face catalogue comprises all 400 images from AT\&T Laboratories Cambridge~\cite{att_faces}. Images remain at native resolution and are quantized as in Section~\ref{sec:catalogue-example}; all quantized images
within each catalogue are distinct. The true elementary measurement channel is
\begin{equation}\label{eq:catalogue-sensing-channel}
A=
\begin{pmatrix}
0.82&0.12&0.04&0.02\\
0.08&0.76&0.12&0.04\\
0.03&0.09&0.75&0.13\\
0.01&0.04&0.12&0.83
\end{pmatrix},
\end{equation}
with rows indexed by the input category and columns by the observed
category. The nominal channel is $\widetilde A=0.6A+0.4I_4$. Comparing
each true image row with every nominal image row verifies that
$\widetilde W_x$ uniquely minimizes $\KL{W_x}{\widetilde W_{x'}}$ over
$x'\in\X$. For a uniform catalogue of $K$ images with $d$ coordinates,
Proposition~\ref{cm:prop:high} therefore gives the limiting predictive KL loss
\begin{equation}\label{eq:catalogue-floor}
\kappa
=\frac{1}{Kd}\sum_{x\in\X}\sum_{j=1}^d
\KL{A(\cdot|x_j)}{\widetilde A(\cdot|x_j)}.
\end{equation}

\paragraph{Montage and independent-coordinate baseline:}
In Figure~\ref{fig:catalogue-montage}, unmeasured components are displayed at $1/2$ and observed categories are scaled by $1/3$. Repeated readings display the most frequent category, with ties resolved toward the lowest category. Inference retains every reading; the catalogue reconstructions are exact posterior means under the true or nominal channel. \\
For the independent-coordinate baseline, define the empirical marginals
\begin{equation*}
\rho_j(a)=\frac{1}{K}\sum_{x\in\X}\ind\{x_j=a\}.
\end{equation*}
Under their product prior, the local posterior after a record
$(J_i,Z_i)_{i=1}^n$ is
\begin{equation*}
\pi_j(a)\propto\rho_j(a)\prod_{i:J_i=j}A(Z_i|a),
\end{equation*}
and the reconstructed coordinate is $\sum_{a=0}^3(a/3)\pi_j(a)$.
An unobserved coordinate retains its prior $\rho_j$.

\paragraph{Numerical details:}
Figure~\ref{fig:catalogue-prediction} averages 6,000 independent CIFAR acquisition trajectories. Within each trajectory, increasing the budget extends the same Poisson observation stream. Posterior probabilities are calculated exactly over the candidate images. Mutual information is estimated as the initial image entropy minus the average posterior entropy.

Reconstruction error is measured against the quantized images scaled to $[0,1]$ and averaged over all coordinates and over the possible true images, weighted by their probabilities under the true posterior. Prediction loss uses the same posterior averaging, but is evaluated on 64 coordinates chosen uniformly without replacement, independently of the observations. This coordinate sample stays fixed within each stream and is used only to evaluate loss; inference uses every measurement. Averaging over coordinate samples gives the full-image prediction loss in expectation. The incorrect-channel estimate includes a correction that reduces sampling variability.\footnote{We add the difference between the full-coordinate and sampled-coordinate averages of $\KL{A(\cdot|x_j)}{\widetilde A(\cdot|x_j)}$, weighted over images by the true posterior. For each fixed record, this correction averages to zero over the independent coordinate sample.}

Between observations, the posterior and estimated loss remain constant. We obtain the accumulated loss in Figure~\ref{fig:catalogue-prediction}(a) by multiplying each estimated loss by the budget interval over which it applies and summing. This integrates the estimate exactly while simulation uncertainty remains. The shaded pointwise 95\% intervals are obtained by resampling complete streams, preserving the dependence between estimates at different budgets. Reconstruction errors displayed as zero reflect numerical precision.

\end{document}